\documentclass{article}%[letter,11pt,final] % Ellis: I commented this out because it was preventing me from compiling

\usepackage{algorithm}
\usepackage[noend]{algpseudocode}
\usepackage{amsmath,amssymb,amsthm}
\usepackage{bm}
\usepackage{bbm}
\usepackage{bbold}
\usepackage{booktabs}
\usepackage{caption}
\usepackage{enumerate}
\usepackage{float}
\usepackage[colorlinks,citecolor=blue,linkcolor=BrickRed]{hyperref}
\usepackage{cleveref}
\usepackage[utf8]{inputenc}
\usepackage{parskip}
\usepackage{subcaption}
\usepackage[normalem]{ulem}
\usepackage[usenames,dvipsnames]{xcolor}
\usepackage{xspace}
\usepackage[numbers,sort&compress]{natbib}

\newtheorem{theorem}{Theorem}

\newtheorem{lemma}{Lemma}
\newtheorem{definition}{Definition}
\newtheorem{corollary}{Corollary}

\newcounter{note}[section]
\renewcommand{\thenote}{\thesection.\arabic{note}}

\newcommand{\gnote}[1]{\refstepcounter{note}$\ll${\bf Goran~\thenote:}
  {\sf \color{red} #1}$\gg$\marginpar{\tiny\bf GZ~\thenote}}
\newcommand{\enote}[1]{\refstepcounter{note}$\ll${\bf Ellis~\thenote:}
  {\sf \color{blue} #1}$\gg$\marginpar{\tiny\bf EH~\thenote}}

\newcommand{\gznote}[1]{\gnote{#1}}

\DeclareMathOperator*{\argmin}{arg\,min}

\newcommand{\E}{\mathop{{}\mathbb{E}}}%Expectation

\newcommand{\OPT}{\mathrm{OPT}}
\newcommand{\poly}{\text{poly}}
\newcommand{\eps}{\varepsilon}%Variance

\newcommand{\mcC}{\mathcal{C}}
\newcommand{\mcD}{\mathcal{D}}

\newcommand{\mcI}{\mathcal{I}}

\newcommand{\mcP}{\mathcal{P}}

\newcommand{\mcS}{\mathcal{S}}

\newcommand{\calT}{\mathcal{T}}
\newcommand{\calC}{\mathcal{C}}
\newcommand{\calI}{\mathcal{I}}

\allowdisplaybreaks

\usepackage{fullpage}
\usepackage{appendix}
\usepackage{thm-restate}
\usepackage[font=itshape]{quoting}
\usepackage{MnSymbol}
\usepackage{graphicx}
\usepackage{booktabs}
\usepackage[justification=raggedright]{caption}

\newcommand{\at}[1]{^{(#1)}}

\newcommand{\1}[1]{\mathbb{I}[#1]}
\newcommand{\hop}{\mathrm{hop}}
\newcommand{\dist}{d}

\newcommand{\ALG}{\mathrm{ALG}} % probably should not be used

\newcommand{\important}[1]{\textbf{#1}}

\newcommand{\pad}{\rho_{\mathrm{pad}}}
\newcommand{\unmatched}{\mathrm{unmatched}}
\newcommand{\alphawc}{\alpha_{\mathrm{WC}}}

\newcommand{\FullOrShort}{full}%Change from full to short for versioning
\ifthenelse{\equal{\FullOrShort}{full}}{
    \newcommand{\fullOnly}[1]{#1}
    \newcommand{\shortOnly}[1]{}
}{
    
    \newcommand{\fullOnly}[1]{}
    \newcommand{\shortOnly}[1]{#1}
}

\title{Tree Embeddings for Hop-Constrained Network Design}

\date{}

\author{\begin{tabular}[t]{c@{\extracolsep{1.5em}}cc} 
        Bernhard Haeupler  & \qquad D Ellis Hershkowitz & Goran Zuzic \\
        \small Carnegie Mellon University & \small \qquad Carnegie Mellon University & \small ETH Z\"urich \\
        \small\texttt{haeupler@cs.cmu.edu} & \small \qquad \texttt{dhershko@cs.cmu.edu} & \small \texttt{goran.zuzic@inf.ethz.ch}
\end{tabular}}

\begin{document}

\maketitle
\begin{abstract}
Network design problems aim to compute low-cost structures such as routes, trees and subgraphs. Often, it is natural and desirable to require that these structures have small hop length or hop diameter. Unfortunately, optimization problems with hop constraints are much harder and less well understood than their hop-unconstrained counterparts. A significant algorithmic barrier in this setting is the fact that hop-constrained distances in graphs are very far from being a metric.

\medskip 

We show that, nonetheless, hop-constrained distances can be approximated by distributions over ``partial tree metrics.'' We build this result into a powerful and versatile algorithmic tool which, similarly to classic probabilistic tree embeddings, reduces hop-constrained problems in general graphs to hop-unconstrained problems on trees. We then use this tool to give the first poly-logarithmic bicriteria approximations for the hop-constrained variants of many classic network design problems. These include Steiner forest, group Steiner tree, group Steiner forest, buy-at-bulk network design as well as online and oblivious versions of many of these problems.

\end{abstract}
\thispagestyle{empty}
\newpage

\section{Introduction}\setcounter{page}{1}

The field of network design studies how to efficiently construct and use large networks. Over the past several decades researchers have paid particular attention to the construction of low-cost computer and transportation networks that enable specified communication and delivery demands. 

Formally, these problems require computation of low-cost structures in graphs, such as paths, trees or subgraphs, that satisfy specified connectivity requirements. For example, there has been extensive work on how, given a weighted graph $G = (V, E, w)$ with $n$ nodes, one can compute a subgraph $H \subseteq G$ of minimum weight that connects: all vertices (\important{minimum spanning tree (MST)}); all vertices in an input $S \subseteq V$ (\important{Steiner tree}); at least $k$ nodes (\important{$k$-MST}); at least $k$ terminals from an input $S \subseteq V$ (\important{$k$-Steiner tree}); at least one vertex from each set $S_i$ for a given collection of vertex sets $S_1, \ldots, S_k \subset V$  (\important{group Steiner tree}); $s_i \in V$ to $t_i \in V$ for every pair in $\{(s_i, t_i)\}_i$ (\important{Steiner forest}); and some vertex in $S_i \subseteq V$ to some vertex in $T_i \subseteq V$ for every pair in $\{(S_i, T_i)\}_i$ (\important{group Steiner forest a.k.a.\ generalized connectivity}). %All of the above problems are NP-hard with the exception of MST. 
%Furthermore, to model the ``economies of scale'' phenomena that occurs in the construction of such networks, a \important{buy-at-bulk} cost structure is often used where the cost of bandwidth on an edge is a subadditive function of the load on that edge. 
To model the uncertainty and dynamic nature of networks, these problems are often generalized to their \important{online} variants where the demands to be connected are revealed over discrete time steps. An even stronger model of uncertainty is the \important{oblivious} setting where an algorithm must specify how it will satisfy each possible demand before it even knows the demands; demands are then revealed and the algorithm buys its pre-specified solution.

However, connectivity alone is often not sufficient for fast and reliable networks. Indeed, we often also desire that our networks be \important{hop-constrained}; namely we desire that demands are not just appropriately connected but connected with a path consisting of a low number of edges (a.k.a.\ hops). By reducing the number of traversed edges, hop constraints facilitate fast communication \cite{akgun2011new,de2018extended}. 
%For instance, \cite{akgun2011new} notes how transmission times in local access networks can be kept low by restricting the number of hops in a network and \cite{de2018extended} notes a similar role of hop constraints in telecommunications networks. 
Furthermore, low-hop networks tend to also be more reliable: if a transmission over an edge fails with some probability, the greater the number of hops between the source and destination, the greater the probability that this transmission fails \cite{rossi2012connectivity,woolston1988design}. 
%\cite{rossi2012connectivity} notes how this phenomena occurs and is moderated by hop constraints in electricity distribution networks and \cite{woolston1988design} demonstrates that hop-constrained spanning trees more successfully transmit messages than spanning trees without hop constraints.

Unfortunately, adding hop constraints to network design problems makes them significantly harder. MST is solvable in polynomial time but MST with hop constraints is known to admit no $o(\log n)$ poly-time approximation algorithm \cite{bar2001generalized}. Similarly, Steiner forest has a constant approximation \cite{agrawal1995trees} but hop-constrained Steiner forest has no poly-time $o(2^{\log ^ {1-\eps} n})$-approximation for any constant $\eps > 0$ \cite{dinitz2015label}.\footnote{Both under standard complexity assumptions.} Indeed, although there has been extensive work on approximation algorithms for simple connectivity problems like spanning tree and Steiner tree with hop constraints \cite{althaus2005approximating,konemann2005approximating,kortsarz1997approximating,kantor2009approximate,ravi1994rapid,marathe1998bicriteria,hajiaghayi2009approximating,khani2011improved}, nothing is known regarding algorithms for many well-studied generalizations of these problems with hop constraints. For instance, no non-trivial algorithms are known for Steiner forest, group Steiner tree, group Steiner forest or online MST with hop constraints.

By allowing an algorithm to ``pretend'' that the input graph is a tree, probabilistic tree embeddings have had enormous success as the foundation of many poly-log approximation algorithms for network design; thus, we might naturally expect them to be useful for hop-constrained network design. Specifically, a long and celebrated line of work \cite{karp19892k,alon1995graph,bartal1996probabilistic,fakcharoenphol2004tight} culminated in the embedding of Fakcharoenphol, Rao and Talwar \cite{fakcharoenphol2004tight}---henceforth ``FRT''---which showed that any metric can be $O(\log n)$-approximated by a distribution $\mcD$ over trees.\footnote{See \Cref{sec:RW} for a formal statement.} Consequently, a typical template for many network design algorithms is to (1) embed the metric induced by weighted graph $G$ into a $T \sim \mcD$; (2) solve the input problem on $T$ (which is typically much easier than the problem on $G$)
and; (3) project the solution on $T$ back into $G$. For example, such a template gives poly-log approximations for group Steiner tree and group Steiner forest \cite{garg2000polylogarithmic,naor2011online}. In the $h$-hop-constrained setting for some $h \geq 1$, the natural notion of distance to consider between vertices $u$ and $v$ is the $h$-hop-constrained distance---the length of the shortest path between $u$ and $v$ according to $w$ with at most $h$ hops. Thus, to use tree embeddings for hop-constrained network design we must first understand how to approximate these distances with trees.

\subsection{Our Contributions}
In this paper we initiate the study of metric approximations for hop-constrained distances and their use in algorithms for hop-constrained network design. Broadly, our results fall into four categories.

\subsubsection{Impossibility of Approximating Hop-Constrained Distances with Metrics}
    We begin by observing that hop-constrained distances are inapproximable by metrics (\Cref{sec:inapxMetrics}). 
    
    \textbf{Results:} Not only are hop-constrained distances not a metric (since they do not satisfy the triangle inequality) but as we show any metric that approximates hop-constrained distances does so with an $\Omega(L)$ multiplicative error where $L$ is the aspect ratio of the input weighted graph (\Cref{lem:noMetric}). This lower bound is matched by a trivial upper bound (\Cref{lem:trivApx}).
     
    \textbf{Discussion:} Since the expected distance between two nodes in a distribution over metrics is itself a metric, our impossibility result also rules out approximating hop-constrained distances with distributions over metrics as in FRT.
    
    \textbf{Techniques:} This observation is proved by careful analysis of a simple example: a path graph.
    
    \subsubsection{Approximating Hop-Constrained Distances with Partial Tree Metrics}
    Despite these apparent roadblocks, we show that---somewhat surprisingly---it is indeed possible to to approximate hop-constrained distances with trees (Sections \ref{sec:distInd}, \ref{sec:ApproximatingHopMetrics}).
    
   \textbf{Results:} We show that a distribution over ``partial tree metrics'' can approximate hop-constrained distances with an expected distance stretch of $O(\log n \log \log n)$ and a worst-case distance stretch of $O(\log ^ 2 n)$ with an $O(\log ^ 2 n)$ relaxation in the hop constraint (\Cref{thm:mainmetric}).
    
    \textbf{Discussion:} This result differs from FRT in two notable ways: (1) our partial tree metrics are partial in the sense that they contain only a constant fraction of nodes from the input graph---indeed, this is what allows us to overcome the impossibility of approximating hop-constrained distances with metrics and; (2) our result provides a worst-case guarantee, unlike FRT which only gives a guarantee in expectation.
    
    \textbf{Techniques:} We show this result by first proving a decomposition lemma (\Cref{lem:hopconstrained-decomp}), which applies padded decompositions to a ``mixture metric'' that combines hops and (unconstrained) distances. We then recursively apply this decomposition, using different combinations of hops and distances in our recursive calls.
    
    \subsubsection{$h$-Hop Partial Tree Embeddings}
    We next build embeddings for hop-constrained network design from our metric approximations (Sections \ref{sec:hop-bounded-hsts}, \ref{sec:repTrees}).
    
    \textbf{Results:} Specifically, we show that one can construct a distribution over ``$h$-hop partial tree embeddings'' of hop-constrained distances with expected distance stretch $O(\log n \log \log n)$ and a worst-case distance stretch $O(\log ^ 2 n)$ with an $O(\log ^3 n)$ relaxation in the hop constraint (\Cref{thm:mainEmbed}). Further, we show that these embeddings can be used for hop-constrained network design as in the above template for network design that uses FRT. Notably, our embeddings reduce many hop-constrained network design problems to their \emph{non-hop-constrained} versions on trees. Since our embeddings, like our partial tree metrics, are also partial, we build on these embeddings by constructing ``$h$-hop partial tree embeddings,'' which represents many draws from our distribution over partial tree embeddings as a single tree.
    
    \textbf{Discussion:} Like our tree metrics and unlike FRT, our tree embeddings are partial and give worst-case guarantees. Moreover, our embeddings follow almost immediately from our metric approximations. However, a notable difference between our embeddings and those of FRT is that demonstrating that they can be used for hop-constrained network design requires a non-trivial amount of work. In particular, while appropriately projecting from an input graph to a tree embedding is trivial in the FRT case, the partialness of our embeddings makes this projection significantly more troublesome. Thus, we develop a projection theorem (\Cref{thm:projSoln}), which informally shows that a natural projection from $G$ to one of our tree embeddings appropriately preserves cost and connectivity. 
    
    \textbf{Techniques:} We prove our projection theorem using ``$h$-hop-connectors'' which are, informally, a hop-constrained version of Euler tours. We emphasize that this projection theorem is only used in the analysis of our algorithms.
    
    \subsubsection{Applications to Hop-Constrained Network Design}
    Lastly, we use our embeddings to develop the first non-trivial approximation algorithms for the hop-constrained versions of many classic network design problems (Sections \ref{sec:firstAppl}, \ref{sec:secondAppl}).
   
   \textbf{Results:} 
   As detailed in \Cref{fig:resultsOverview}, we give numerous (poly-log, poly-log) bicriteria algorithms for hop-constrained network design problems that relax both the cost and hop constraint of the solution.
   
   \renewcommand{\arraystretch}{1.3}
   \begin{table}[htbp]
       \centering
       \label{fig:replication-as-requested}
       \begin{tabular}{|l|c|c|c|c|}
           \hline
           \textbf{Hop-Constrained Problem} & \textbf{Cost Apx.} & \textbf{Hop Apx.} & \textbf{Cost In $\E$} & \textbf{Section} \\ \hline
           \multicolumn{5}{|l|}{\qquad\textbf{Offline Problems}} \\ \hline
           \multicolumn{1}{|l|}{\qquad \qquad Relaxed $k$-Steiner Tree} & $O(\log ^ 2 n)$ & $O(\log ^ 3 n)$ & & \ref{sec:HCkST}  \\ \hline
           \multicolumn{1}{|l|}{\qquad \qquad $k$-Steiner Tree} & $O(\log ^ 3 n)$ & $O(\log ^ 3 n)$ & &  \ref{sec:HCkST}  \\ \hline
           \multicolumn{1}{|l|}{\qquad\qquad  Group Steiner Tree} & $O(\log ^ 5 n)$ & $O(\log ^ 3 n)$ & & \ref{sec:groupST}, \ref{sec:groupSTBetter}\\ \hline
           \multicolumn{1}{|l|}{\qquad\qquad  Group Steiner Forest} & $O( \log ^ 7 n)$ & $O(\log ^ 3 n)$ & & \ref{sec:GSF} \\ \hline
           \multicolumn{5}{|l|}{\qquad \textbf{Online Problems}} \\ \hline
           \multicolumn{1}{|l|}{\qquad\qquad Group Steiner Tree} & $O(\log ^ 6 n)$ & $O(\log ^ 3 n)$ & $\checkmark$ & \ref{sec:groupSTOnline}\\ \hline
           \multicolumn{1}{|l|}{\qquad\qquad Group Steiner Forest} & $O(\log ^ 8 n)$ & $O(\log ^ 3 n)$ & $\checkmark$ & \ref{sec:onlineGSF} \\ \hline
           \multicolumn{5}{|l|}{\qquad \textbf{Oblivious Problems}} \\ \hline
           \multicolumn{1}{|l|}{\qquad\qquad Steiner Forest} & $O(\log ^ 3 n)$ & $O(\log ^ 3 n)$ & &\ref{sec:steinFor}\\ \hline
           \multicolumn{1}{|l|}{\qquad\qquad Network Design} & $O(\log ^ 4 n)$ & $O(\log ^ 3 n)$ &  & \ref{sec:HCOND}\\ \hline
       \end{tabular}\caption{Our bicriteria approximation results. All results are for poly-time algorithms that succeed with high probability (at least $1 - \frac{1}{\poly (n)}$). For some of the problems we assume certain parameters are $\poly(n)$ to simplify presentation; see the relevant sections for more details. All results are new except for the $k$-Steiner tree result which is implied by \cite{khani2011improved}.} \label{fig:resultsOverview}
   \end{table}
   
   \textbf{Discussion:}
   As noted above, bicriterianess is necessary for any poly-log approximation for Steiner forest and its generalizations. Furthermore, while the results in \Cref{fig:resultsOverview} are stated in utmost generality, many special cases of our results were to our knowledge not previously known. For example, our algorithm for hop-constrained oblivious Steiner forest immediately gives new algorithms for hop-constrained Steiner forest, hop-constrained online Steiner tree and hop-constrained online Steiner forest, as well as min-cost $h$-spanner (see \Cref{sec:steinFor} for details). Similarly, our algorithm for oblivious network design immediately gives new algorithms for the hop-constrained version of the well-studied buy-at-bulk network design problem \cite{awerbuch1997buy}.

   \textbf{Techniques:} All of our algorithms for these problems use the above mentioned tree embedding template with either our $h$-hop partial tree embeddings or our $h$-hop repetition tree embeddings.

\section{Related Work}\label{sec:RW}
Before proceeding to our results we give a brief overview of additional related work on approximation algorithms for hop-constrained network design and tree embeddings.  We also later give related work for each problem whose hop-constrained versions we study in that problem's section.

%\subsubsection{Applied Work in Operations Research}
% These works highlight the practicality of network design with hop constraints; notably their effect on network latency and network reliability. By reducing the number of traversed edges, hop constraints facilitate fast communication. For instance, \cite{akgun2011new} notes how transmission times in local access networks can be kept low by restricting the number of hops in a network and \cite{de2018extended} notes a similar role of hop constraints in telecommunications networks. Furthermore, low-hop networks tend to also be more reliable: if we imagine that a transmission over an edge fails with some probability, the greater the number of hops between the source and destination of a transmission the greater the probability that this transmission fails. \cite{rossi2012connectivity} notes how this phenomena occurs and is moderated by hop constraints in electricity distribution networks and \cite{woolston1988design} demonstrates that hop-constrained spanning trees more successfully transmit messages than spanning trees without hop constraints.

\subsubsection{Hop-Constrained Network Design}

For some simple hop-constrained network design problems non-trivial (unicriteria) approximation algorithms are known. \cite{althaus2005approximating} gave an $O(\log n)$ approximation for minimum depth spanning tree on complete graphs which define a metric. \cite{konemann2005approximating} gave a $O(\sqrt{\log n})$ for the degree-bounded minimum diameter spanning tree problem. \cite{kortsarz1997approximating} gave a $O(d \log n)$ approximation for computing a minimum cost Steiner tree with depth at most $d$. \cite{kantor2009approximate} gave a constant approximation for the minimum depth Steiner tree problem if the input graph is a complete graph defining a metric.

However, hop constraints often make otherwise easy problems so challenging that the only non-trivial approximation algorithms known or possible are bicriteria. %As noted above, while Steiner forest admits a constant approximation \cite{agrawal1995trees}, hop-constrained Steiner forest admits no (unicriteria) $O(2^{\log ^ {1 - \epsilon}}n)$ approximation for any $\epsilon > 0$ under standard complexity assumptions \cite{dinitz2015label}. 
The apparent necessity of bicriterianess in hop-constrained optimization is highlighted by the existence of many bicriteria algorithms. For example, \cite{ravi1994rapid} and \cite{marathe1998bicriteria} gave an $(O(\log n), O(\log n))$ bicriteria approximation algorithms for MST and Steiner tree with hop constraints.\footnote{A later paper of \cite{naor1997improved} claimed to improve this result to a $(O(\log n), 2)$-approximation but it is our understanding that this paper was retracted due to a bug.}
Similarly, \cite{hajiaghayi2009approximating} gave a $(O(\log ^ 4 n), O(\log ^ 2 n))$-bicriteria algorithm for $k$-Steiner tree with hop constraints which was later improved to $(O(\log ^ 2 n), O(\log n))$ by \cite{khani2011improved}; here the first term is the approximation in the cost while the second term is the approximation in the hop constraint.

Lastly, while we have given results from the theory community, we note that hop-constrained network design has received considerable attention from the operations research community; see, for example, \cite{akgun2011new,gouveia1995using,gouveia1996multicommodity,gouveia2003network,voss1999steiner,gouveia2001new,botton2013benders,botton2015hop,diarrassouba2016integer,thiongane2015formulations,leitner2016layered,bley2013ip,de2018extended,diarrassouba2018k,rossi2012connectivity,woolston1988design} among many other papers.

%However, while there has been extensive work on simple connectivity problems like spanning tree and Steiner tree with hop constraints, relatively little is known about well-studied generalizations of these problems with hop constraints. For instance, no non-trivial algorithms are known for Steiner forest, group Steiner tree, group Steiner Forest or online Steiner tree with hop constraints.

\subsubsection{Tree Embeddings}
The celebrated embedding of \cite{fakcharoenphol2004tight} showed that for any metric $(V,d)$ there is a distribution $\mcD$ of weighted trees on $V$ so that for any $u,v \in V$ we have $d(u,v) \leq d_T(u,v)$ for any tree $T$ in the support of $\mcD$ and $\E_{T \sim \mcD}d_T(u,v) \leq O(\log n \cdot d(u,v))$; here, $d_T$ indicates the distance according to the weight function in $T$. Using these tree embeddings with the above template reduces many graph problems to their tree versions at the cost of $O(\log n)$ in the quality of the resulting solution. This has lead to a myriad of algorithms with poly-logarithmic approximation and competitive ratios for NP-hard problems including, among many others, the $k$-server \cite{bansal2011polylogarithmic}, metrical task systems \cite{bartal1997polylog}, offline and online group Steiner tree and group Steiner forest \cite{alon2006general,naor2011online, garg2000polylogarithmic}, buy-at-bulk network design \cite{awerbuch1997buy} and oblivious routing problems \cite{racke2002minimizing}.

There has also been considerable work on extending the power of tree embeddings to a variety of other settings including tree embeddings for planar graphs \cite{konjevod2001approximating}, online tree embeddings \cite{barta2020online}, dynamic tree embeddings \cite{forster2020dynamic,chechik2020dynamic} and distributed tree embeddings \cite{khan2012efficient}. Additionally, there has been significant work on tree embeddings where the trees into which the graph is embedded are subtrees of the input graph \cite{alon1995graph,elkin2008lower,abraham2008nearly,koutis2011nearly}. This line of work culminated in the petal-decomposition of \cite{abraham2012using} which shows that, up to an $O(\log \log n)$, there are subtree embeddings which match the results of FRT.

%There is also a long line of work on universal guarantees for low-depth low-weight spanning and Steiner trees \cite{dinitz2008shallow,elkin2015steiner,solomon2014euclidean,kantor2009approximate,hajiaghayi2009approximating}. 

%Spanners?

\section{Preliminaries}\label{sec:prelims}
Before proceeding to our formal results we define conventions we use throughout this work.

\textbf{General:}
We let $[k] := \{1, 2, \ldots, k\}$ for any non-negative integer $k$. We let $A \sqcup B$ denote the disjoint union of $A$ and $B$. We often use the Iverson bracket notation $\1{\text{condition}}$ which evaluates to $1$ when the $\mathrm{condition}$ is true and $0$ otherwise.

%\textbf{Subadditive functions}
%\gznote{maybe define subadditive functions here?}

\textbf{Graphs:}
Given a graph $G = (V, E)$ we denote its vertex set by $V(G)$ and $E(G)$, or simply $V$ and $E$ if $G$ is clear from context. We let $n := |V|$. All graphs considered in this paper are undirected. Most commonly, we consider undirected weighted graphs $G = (V, E, w)$ with weights $w : E \to \{1, 2, \ldots, L\}$. The value $L$ is called the aspect ratio and throughout this paper we assume $L = \poly(n)$. We will let $w_G$ be $G$'s weight function if $G$ is not clear from context. Generally, weighted graphs in this paper are assumed to be complete, i.e., $E = \binom{V}{2}$. In the context of this paper this is without loss of generality. In particular, one can transform any non-complete weighted graph $G = (V, E, w)$ with aspect ratio $L$ into an equivalent complete graph $G'$ with aspect ratio $L' = n^2 \cdot L$ which gives a weight of $L'$ to any edge not in $E$ without affecting any of the results in this paper.

% In this case it might be helpful to think about $L$ as a sufficiently large number so that $w(u, v) = L$ effectively models the case when $u$ and $v$ are not connected via an edge. In other words, the completeness assumption is without loss of generality since we can model general (non-complete) graphs this way and this distinction is immaterial in all problems considered throughout this paper. For this reason, when we simply refer to ``weighted graphs'' we mean complete, weighted undirected graphs. We will also let $w_G$ stand for the weight function of graph $G$ if $G$ is not clear from context. Finally, we consider (undirected) \textbf{unweighted graphs} $G = (V, E)$ which are not necessarily complete.

\textbf{Subgraphs:}
Given a weighted graph $G = (V(G), E(G), w_G)$, we will often consider a subgraph $H = (V(H), E(H))$ where $V(H) \subseteq V(G)$ and $E(H) \subseteq E(G)$. Unlike $G$, such subgraphs will not necessarily be complete. We will often identify a subset of edges $E' \subseteq E(G)$ of a graph $G$ with the subgraph induced by these edges, i.e., the subgraph $H$ with $E(H) = E'$ and $V(H) = \bigcup_{e \in E(H)} e$. Given a collection of vertices $U \subseteq V(G)$ we will let $G[U]$ be the ``induced'' subgraph with vertex set $U$ and edge set $\{\{u,v\} : u,v \in U \text{ and } \{u,v\} \in E(G)\}$. We define the weight of a subgraph $w_G(H) := \sum_{e \in E(H)} w_G(e)$ as the sum of weights of its edges.

\textbf{Well-Separated Trees:}
We will often work with well-separated rooted trees. We say that a weighted rooted tree $T = (V, E, w)$ with root $r \in V$ is well-separated if every root-to-leaf path has weights that are decreasing powers of $2$. That is, if $e'$ is a child edge of $e$ in $T$ then $w(e') = \frac{1}{2}w(e)$.

\textbf{Distances and Metrics:}
For a set $V$ we call any positive real function $d : V \times V \to R_{\ge 0}$ which is symmetric, i.e., satisfies $d(u,v)=d(v,u)$ for all $u,v \in V$, and satisfies the identity of indiscernibles, i.e., $d(u,v) = 0 \Leftrightarrow u=v$, a distance function (such a function is also often called a semimetric). If $d$ also satisfies the triangle inequality $d(u,w) \leq d(u,v) + d(v,u)$ for all $u,v,w \in V$ then $d$ is called a metric. We also extend the definition of $d$ to sets in the standard way: $d(U, U') := \min_{u \in U, u' \in U'} d(u, u')$.

\textbf{Paths, Path Length, and Hop Length:}
A sequence $P = (v_0, v_1, \ldots, v_\ell)$ of nodes in a graph $G$ is called a path if for all $i \in [\ell]$ we have $\{v_{i-1},v_i\} \in E(G)$ and we say $E(P) := \bigcup_i \{\{p_{i-1},p_i\}\} \subseteq E(G)$ is the edge set of $P$. If the nodes in $P$ are distinct we say that $P$ is simple. In this paper paths are not assumed to be simple. We denote the number of hops in $P$ with $\hop(P) := \ell$ and call $\hop(P)$ the hop length of $P$. If $G = (V, E, w)$ is weighted, we define the weight of a path $P$ in $G$ to be the sum of weights of its edges: $w(P) := \sum_{e \in E(P)} w(e)$.

\textbf{Hop Distance and Hop Diameter:}
For a (non-complete) subgraph $H = (V(H), E(H))$ of a (complete) graph $G$ we let $\hop_H(u, v)$ be the minimum number of edges of a path between $u$ and $v$ in $H$ (i.e., using only the edges $E(H)$). We also define the hop diameter of $H$ as $\hop(H) := \max_{u, v \in V(H)} \hop_H(u, v)$.
 
\textbf{Shortest-Path Metric and Tree Metric:}
For a weighted graph $G$ the distance between any two nodes $u, v \in V$ is defined as $\dist_G(u, v) := \min \{ w(P) \mid \text{path $P$ between } u, v \}$. It is easy to verify that $\dist_G$ is a metric on $V$ and for this reason $\dist_G$ is called the shortest path metric of $G$. Any metric $d$ on a set $V$ which is identical to a shortest path metric of a weighted tree $T=(V,E,w)$ is called a tree metric; for this reason we will sometimes conflate a tree metric with its corresponding tree.

% SHIFTED TO TECHNICAL SECTION FOR NOW
%\textbf{Hop-Constrained Distances.}
%For a weighted graph $G=(V,E,w)$ and a hop constraint $h \geq 1$ we define the $h$-hop distance between any two nodes $u, v \in V$ as $$\dist_G\at{h}(u, v) := \min \{ w_G(p) \mid \text{path $p$ between } u, v \text{ with } \hop(w) \le k \},$$ that is, the length of the shortest path (according to $w$) between $u$ and $v$ with at most $h$ hops. This assumes that there exists an $h$-hop path between $u$ and $v$ in $G$, which is always the case if $G$ is complete. It is easy to verify that $\dist_G\at{h}$ is a distance function on $V(G)$ (but not necessarily a metric). 

\section{Approximating Hop-Constrained Distances}\label{sec:apxHCD}
In this section we show that even though hop-constrained distances are not well-approximated by any metric, they are approximated by a distribution over what we call partial tree metrics. More specifically, we consider hop-constrained distances defined as follows.
\begin{definition}[Hop-Constrained Distances]
    For a (complete) weighted graph $G=(V,E,w)$ and a hop constraint $h \geq 1$ we define the $h$-hop distance between any two nodes $u, v \in V$ as $$\dist_G\at{h}(u, v) := \min \{ w(P) \mid \text{path $P$ in $G$ between } u, v \text{ with } \hop(P) \le h \}.$$
\end{definition}
As we have assumed that our graph $G$ is complete without loss of generality (see \Cref{sec:prelims}), the above is always well-defined for any $u,v \in V$.
\shortOnly{We note that all omitted proofs in this section appear in \Cref{sec:defProofsApxMetrics}.}

\subsection{Hop-Constrained Distances Are Inapproximable by Metrics}\label{sec:inapxMetrics}
We begin by observing that, not only is $d\at{h}_G$ not a metric, but it is, in general, innaproximable by any metric.

It is easy to verify that $\dist_G\at{h}$ is a valid distance function on $V(G)$. Indeed $\dist_G\at{h}$ is clearly symmetric, i.e., $d^{(h)}(u,v) = d^{(h)}(v,u)$, and satisfies the identity of indiscernibles, i.e., $d^{(h)}(u,v) = 0 \Leftrightarrow u=v$. However, it is also simple to see that hop-constrained distances are not necessarily metrics since they do not obey the triangle inequality. Indeed, the existence of a short $h$-hop path from $u$ to $v$ and a short $h$-hop path from $v$ to $w$ does not imply that the existence of a short $h$-hop path between $u$ and $w$. More formally it is possible that $d^{(h)}(u,w) \gg d^{(h)}(u,v) + d^{(h)}(v,w)$. 

Of course with a factor $2$ relaxation in the hop constraint the relaxed triangle inequality $d^{(2h)}(u,w) \leq d^{(h)}(u,v) + d^{(h)}(v,w)$ holds for any graph $G$ and any $u,v,w \in V(G)$. This suggests---albeit incorrectly---that one might be able to approximate hop-constrained distance by allowing constant slack in the hop constraint and length approximation as in the following definition. 

\begin{definition}\label{def:metricapproximation}
A distance function $\tilde{d}$ % : V \times V \to R_{\ge 0}$ 
 approximates the $h$-hop constrained distances $d_G^{(h)}$ for a weighted graph $G=(V,E,w)$ where $h \ge 1$ with \textbf{distance stretch} $\alpha \ge 1$ and \textbf{hop stretch} $\beta \ge 1$ if for all $u,v \in V$ we have \[d_G^{(\beta h)}(u,v) \leq \tilde{d}(u,v) \leq \alpha \cdot d_G^{(h)}(u,v).\]
\end{definition}
As we next observe, no metric provides such an approximation without a very large hop or distance stretch.

\begin{restatable}{lemma}{noMetric}\label{lem:noMetric}
For any hop constraint $h \geq 1$, distance stretch $\alpha$, hop stretch $\beta$ and any $L > 1$, there exists a graph $G = (V,E,w)$ with aspect ratio $L$ such that if a metric $\tilde{d}$ approximates $d_G\at{h}$ with distance stretch $\alpha$ and hop stretch $\beta$ then $\alpha (\beta h + 1) \ge L$.
% an approximation of $d_G\at{h}$ with distance stretch $\alpha$ and hop stretch $\beta$ and $\alpha (\beta h + 1) \ge L$ is always trivial: For such values of $\alpha$, $\beta$, $h$, and $L$ and any weighted graph $G$ with aspect ratio $L$ the scaled unconstrained shortest-path metric $\alpha \cdot d_G$ ``approximates'' $d_G\at{h}$ with distance stretch $\alpha$ and hop stretch $\beta$.
\end{restatable}
\gdef\noMetricProof{
\begin{proof}
Set $k := \beta h + 1$ and consider the path graph with vertices $v_0, v_1, \ldots, v_k$ where the edges have a uniform weight of $1$ (and all other edges have length $L$). Note that $\tilde{d}(v_i, v_{i+1}) \le \alpha \cdot d\at{h}(v_i, v_{i+1}) = \alpha$. Applying the triangle inequality $k$ times gives $\tilde{d}(v_0, v_k) \le \alpha k$. However, $\alpha k \ge \tilde{d}(v_0, v_k) \ge d\at{h \beta}(v_0, v_k) = L$, giving us that $\alpha ( \beta h + 1) \ge L$.
\end{proof}
}
\fullOnly{\noMetricProof}

Indeed, an approximation with the above large stretch is always trivially attainable. In particular, no metric can approximate $d\at{h}$ any better than the trivial approximation by the scaled shortest-path metric $\alpha \cdot d_G$ which gives value $\alpha \cdot d_G(u,v)$ to each $u,v \in V$, as shown by the following.

\begin{restatable}{lemma}{trivApx}\label{lem:trivApx}
    Given any graph $G = (V, E, w)$ with aspect ratio $L$ and a distance stretch $\alpha$ and hop stretch $\beta$ satisfying $\alpha ( \beta h + 1) \ge L$, we have that $\alpha \cdot d_G$ approximates $d_G\at{h}$ with distance stretch $\alpha$ and hop stretch $\beta$.
\end{restatable}
\gdef\trivApxProof{
\begin{proof}
%    To show the second part of \Cref{lem:noMetric} assume that $h,\alpha,\beta \ge 1$ with $\alpha (\beta h + 1) \ge L$ and any graph $G$ are given. 
    Set $\tilde{d}(u, v) := \alpha \cdot d_G(u, v)$ where $d_G$ is the standard shortest-path metric on $G$. It remains to check that $\tilde{d}(u, v)$ satisfies the requirements of \Cref{def:metricapproximation}. The right hand side of the inequality in \Cref{def:metricapproximation} clearly holds since $d_G(u, v) \le d_G\at{h}(u, v)$ implies that $\tilde{d}(u, v) = \alpha d_G(u, v) \le \alpha d_G\at{h}(u, v)$. We now argue the left hand side, i.e., $\tilde{d}(u, v) \ge d_G\at{\beta h}(u, v)$, by considering two cases: on the one hand, if $d_G(u, v) > \beta h$ then $\tilde{d}(u, v) = \alpha d_G(u, v) \geq \alpha (\beta h + 1) \ge L \ge d_G\at{\beta h}(u, v)$. On the other hand, if $d_G(u, v) \le \beta h$, this value must come from an unconstrained shortest-path with hop distance (and length) of at most $\beta h$ in which case $d_G\at{\beta h}(u, v) = d_G(u, v)$ and therefore $\tilde{d}(u, v) = \alpha \cdot d_G(u, v) \ge d_G(u, v) = d_G\at{\beta h}(u, v)$ as desired. 
\end{proof}
}
\fullOnly{\trivApxProof}

Thus, hop-constrained distances can be maximally far from \emph{any} metric in the sense that the only way to approximate them by a metric requires so much slack in the hop and distance stretch that the approximation becomes trivial. Moreover, since the expected distance between two nodes in a distribution over metrics is itself a metric, the above result also rules out approximating $d^{(h)}$ in a non-trivial way with distributions over metrics as in FRT. This impossibility remains even when one allows for relaxations of the hop constraint.

\subsection{Distances Induced by Distributions Over Partial Metrics}\label{sec:distInd}
While \Cref{lem:noMetric} shows that no metric can approximate $d\at{h}_G$ on all vertices, it does not rule out the possibility that some metric approximates $d\at{h}_G$ on a large subset of $V$. Thus, we introduce the following concept of partial metrics.

\begin{definition}[Partial Metric]
    Any metric $d$ defined on a set $V_d$ is called a \textbf{partial metric} on $V$ if $V_d \subseteq V$.
\end{definition}

We will often talk about how partial metric $d$ approximates $d\at{h}_G$ on $V_d$ with hop and distance stretches $\alpha$ and $\beta$ by which we mean that the inequality of \Cref{def:metricapproximation}---$d_G^{(\beta h)}(u,v) \leq d(u,v) \leq \alpha \cdot d_G^{(h)}(u,v)$---holds for every $u,v \in V_d$. Of course, a partial metric on the empty set trivially approximates $d\at{h}_G$ and we are ultimately interested in estimating $d\at{h}$ on all pairs of nodes. For this reason, we give the following notions of exclusion probability and how a distribution over partial metrics can induce a distance function between all nodes.

%We extend the notion of approximating $d\at{h}_G$ from \Cref{def:metricapproximation} to partial metrics in the natural way.

%\begin{definition}
%    A partial metric $d$ approximates $d\at{h}_G$ with distance stretch $\alpha \geq 1$ and hop stretch $\beta \geq 1$ for a weighted graph $G=(V,E,w)$ and $h \ge 1$ if for every $u,v \in V_d$ we have 
%    \[d_G^{(\beta h)}(u,v) \leq d(u,v) \leq \alpha \cdot d_G^{(h)}(u,v).\]
%\end{definition}

\begin{definition}[Distances of Partial Metric Distributions]\label{lem:partialMetricDist}
 Let $\mathcal{D}$ be a distribution of partial metrics of $V$ for weighted graph $G = (V, e, w)$. We say $\mathcal{D}$ has \textbf{exclusion probability} $\eps$ if for all $v \in V$ we have $\Pr_{d \sim \mathcal{D}}[v \in V_d] \geq 1 - \eps$. If $\epsilon \leq \frac{1}{3}$ then we say that $\mathcal{D}$ induces the distance function $d_{\mathcal{D}}$ on $V$, defined as
$$d_{\mathcal{D}}(u,v) :=  \E_{d \sim \mathcal{D}}\left[d(u,v) \cdot \1{ u,v \in V_d} \right ].$$
\end{definition}
 
It is easy to verify that $d_\mcD$ is indeed a distance function. %, showing that the above definition is well-defined.
%\begin{lemma}\label{lem:distGivesDist}
%    $d_{\mcD}$ is a distance function if $\mcD$ is a distribution over partial metrics with exclusion probability $\frac{1}{3}$.
%\end{lemma}
%\begin{proof}
    In particular, we trivially have that $d_\mcD(v,v) = 0$ since $d(v, v) = 0$ for all $d$ in the support of $\mcD$.
    An exclusion probability bounded above by $\frac{1}{2}$ guarantees that $\Pr_{d \sim \mathcal{D}}[u,v \in V_d] > 0$ for any $u,v \in V$. This guarantees that $d_{\mathcal{D}}(u,v) > 0$ for $u\neq v$ which makes $d_{\mathcal{D}}$ a valid distance function. It may be useful for the reader to think of $d_{\mathcal{D}}(u,v)$ as a conditional expected distance where we condition on $u$ and $v$ both being in the metric drawn from $\mcD$.%It also guarantees that $\E_{d \sim \mathcal{D}}\left[d(u,v)  \cdot \1{u, v \in V_d} \right] = \Theta(\E_{d \sim \mathcal{D}}\left[d(u,v) \mid u, v \in V_d \right])$, which allows one to think of $d_{\mathcal{D}}$ as a conditional expected distance. 
%\end{proof}
 
With these definitions in place we can define what it means for a distribution of partial metrics to approximate hop-constrained distances. %Our definition is essentially \Cref{def:metricapproximation} applied to the support of $\mathcal{D}$.
%, except that we want every metric $d$ in the support of $\mathcal{D}$ to dominate $d_G^{(\beta h)}$ on $V_d$.

\begin{definition}[Stretch of Partial Metric Distribution]\label{def:partialmetricapproximation}
A distribution $\mathcal{D}$ of partial metrics on $V$ with exclusion probability at most $\frac{1}{3}$ approximates $d^{(h)}$ on weighted graph $G=(V,E,w)$ for hop constraint $h \ge 1$ with \textbf{worst-case distance stretch} $\alpha_{WC} \ge 1$ and \textbf{hop stretch} $\beta \ge 1$ if each $d$ in the support of $\mathcal{D}$ approximates $d^{(h)}_G$ on $V_d$ with distance stretch $\alpha_{WC}$ and hop stretch $\beta$, i.e. for each $d$ in the support of $\mcD$ and all $u,v \in V_d$ we have \[d_G^{(\beta h)}(u,v) \leq d(u,v) \leq \alpha \cdot d_G^{(h)}(u,v).\]
Furthermore, $\mcD$ has \textbf{expected distance stretch} $\alpha_{\E}$ if for all $u,v \in V$ we have \[d_{\mathcal{D}}(u,v) \leq \alpha_{\E} \cdot d_G^{(h)}(u,v).\]

%and any partial metric $d$ in the support of $\mathcal{D}$ , i.e., if for every $u,v \in V_d$ we have 
%$$d_G^{(\beta h)}(u,v) \leq d(u,v) \leq \alphawc \cdot d_G^{(h)}(u,v).$$

%\begin{itemize}
	%\item $$d_{\mathcal{D}}(u,v) \leq \alpha \cdot d_G^{(h)}(u,v)$$
	%\item any partial metric $d$ in the support of $\mathcal{D}$ approximates $d^{(h)}$ on $V_d$ with distance stretch $\alpha$ and hop stretch $\beta$, i.e., if for every $u,v \in V_d$ we have 
%$$d_G^{(\beta h)}(u,v) \leq d(u,v) \leq \alpha d_G^{(\beta h)}(u,v)$$
%\end{itemize}
%$$d_{\mathcal{D}}(u,v) \leq \alpha \cdot d_G^{(h)}(u,v)$$
%and for any partial metric $d$ in the support of $\mathcal{D}$ and any $u,v \in V_d$ satisfies
%$$d_G^{(\beta h)}(u,v) \leq d(u,v).$$
\end{definition}

\subsection{Approximating Hop-Constrained Distances with Partial Tree Metrics}\label{sec:ApproximatingHopMetrics}

Even though $h$-hop distances are generally inapproximable by distributions over metrics, we now show that they are well-approximated by distributions over very simple partial metrics, namely well-separated partial tree metrics.
\begin{theorem}\label{thm:mainmetric}
For any (complete) weighted graph $G$, any hop-constraint $h \geq 1$, and any $0 < \eps < \frac{1}{3}$ there is a distribution $\mathcal{D}$ over well-separated tree metrics each of which is a partial metric on $V(G)$ such that $\mathcal{D}$ has exclusion probability at most $\eps$ and approximates $d_G\at{h}$ with expected distance stretch $\alpha_{\E} = O(\log n \cdot \log \frac{\log n}{\eps})$, worst-case distance stretch $\alphawc = O(\frac{\log^2 n}{\eps})$ and hop stretch $\beta = O(\frac{\log^2 n}{\eps})$. 
\end{theorem}

%\gznote{btw, why isn't $0 < \eps < 1$ enough?}\bnote{we assume eps < 1/3 for our definition of d_{D} and its approximation guarantees} 

The rest of \Cref{sec:ApproximatingHopMetrics} is dedicated to the proof of \Cref{thm:mainmetric}. In \Cref{sec:decomposition-lemma} we define simple ``mixture metrics'' and show how combining these metrics with padded decompositions leads to random decompositions with desirable properties regarding hop-constrained distances. In \Cref{sec:mainproof} we show how recursively refining these partitions gives a random partial tree metric which proves \Cref{thm:mainmetric}.

%  Given a weighted graph $G = (V(G), E(G), w_G)$ we call a distribution $\calT$ over partial tree embeddings $T = (V(T), E(T), w_T)$ of $G$ a \textbf{partial tree distribution}. Given a hop constraint $h \ge 1$ we say $\calT$ is a $h$-hop distribution with \textbf{expected distance stretch} $\alpha_{\E}$ if $\E_{T \sim \calT}\left[ \1{u, v \in V(T)} \cdot \dist_T(u, v) \right] \le \alpha_{\E} \cdot d_G\at{h}(u, v)$ for all $u, v \in V(G)$. Finally, we say $\calT$ has \textbf{exclusion probability} $\eps$ if $\Pr_{T \sim \calT}[v \in V(T)] \le \eps$ for all $v \in V(G)$.

\subsubsection{Mixture Metrics and Padded Decompositions for Hop-Constrained Distances}\label{sec:decomposition-lemma}
% In this section we prove a useful graph decomposition lemma that can handle multidimensional weights on edges. Notably, the Lemma will furnish dealing with $\mathrm{congestion} + \mathrm{dilation}$ constraints (or hop-bounded constraints) by reserving one dimension for hops and another dimension for edge costs.

To better understand the structure of hop-constrained distances, we develop a decomposition lemma which gives structure both in terms of weights and hops. In particular, we call a collection of disjoint vertex sets $C_1 \sqcup C_2 \sqcup \ldots \sqcup C_k$ a \textbf{partial vertex partition}; $C_1 \sqcup C_2 \sqcup \ldots \sqcup C_k$ is a \textbf{complete vertex partition} if $\bigcup_i C_i = V$. In a nutshell, we decompose the vertices of a weighted graph $G$ into a partial vertex partition where (1) both the hop diameter and weight diameter of all $C_i$'s is small, (2) $C_i$ and $C_j$ for $i \neq j$ are well-separated both in terms of hops and weight and (3) almost every vertex is in the partial vertex partition. Our decomposition combines two simple ingredient. 

Our first ingredient is what we call the mixture metric which is obtained by mixing together hop lengths and weights in the following way.
\begin{definition}[Mixture Metric]
  Given a weighted graph $G = (V, E, w)$, a hop scale $h > 0$, and a weight scale $b > 0$, we define a \textbf{mixture weight} $w' : E \to \mathbb{R}_{\ge 0}$ of an edge $e \in E$ as $w'(e) := 1/h + w(e)/b$. The shortest path metric induced by $w'$ is called the \textbf{mixture metric} $d' : V \times V \to \mathbb{R}_{\ge 0}$.
\end{definition}
The utility of the mixture metric is given by three easy to verify facts: It is a metric and so is amenable to standard metric decomposition theorems; if $d'(u, v) \le \alpha$ in the mixture metric with hop scale $h$ and weight scale $b$, then $d\at{\alpha \cdot h}(u, v) \le \alpha \cdot b$; if $d'(u, v) > \alpha$, then $d\at{\alpha \cdot h / 2}(u, v) > \alpha \cdot b / 2$.

Our second ingredient is the well-studied padded decomposition~\cite{gupta2003bounded,abraham2019cops}. Given a metric space $(V, \dist)$ we denote the ball of radius $r \ge 0$ around $x \in V$ with $B_{\dist}(x, r) := \{ y \in V \mid \dist(x, y) \le r \}$ . Next, let $C = C_1 \sqcup \ldots \sqcup C_k$ be a (partial or complete) vertex partition. Then, for a subset $U \subseteq V$, we say that \textbf{$U$ is broken in $C$} if $|\{ i \mid U \cap C_i \neq \emptyset \}| > 1$. We also denote this event by $U \not \subseteq C$ and its logical negation by $U \subseteq C$. With this notation, we define padded decompositions:
\begin{definition}[Padded Decompositions]
  \label{def:padded}
  Let $(V, \dist)$ be a metric space and let $\mcC$ be a distribution over complete vertex partitions. $\mcC$ is a $(\pad,\Delta)$-padded decomposition if:
  \begin{enumerate}
  \item \textbf{Diameter:} $\max_{u, v \in C_i} \dist(u, v) \le \Delta$ for each $C = C_1 \sqcup C_2 \sqcup \ldots \sqcup C_k$ in the support of $\mcC$ and $i \in [k]$.
  \item \textbf{Paddedness:} $\Pr_{C\sim \mcC}[B_{\dist}(v, r) \not \subseteq C] < \frac{r \cdot \pad}{\Delta}$ for each $v \in V$ and every $r > 0$.
  \end{enumerate}
\end{definition}
In other words, each part of a partition in $\mcC$ has diameter at most $\Delta$ and the probability of a node being within $r$ from a node in a different part is at most $\frac{r \pad}{\Delta}$. The value $\pad$ is known as the \textbf{padding parameter}.\footnote{We note that our definition of $\pad$ slightly differs from that of other papers, albeit only by a constant factor.}
%The value $\pad$ is known as the \textbf{padding parameter} and, while the best possible value for arbitrary graphs is $\pad = \Theta(\log n)$, for many special cases better bounds are known. E.g., for metrics of doubling constant $\lambda$ the padding parameter is $\Theta(\log \lambda)$~\cite{gupta2003bounded}; for genus-$g$ graphs $\pad = \Theta(\log g)$~\cite{lee2010genus}; for $K_r$-free graphs $\pad = O(r)$~\cite{abraham2019cops}. However, since in this paper we consider (complete) weighted graphs, we assume $\pad = \Theta(\log n)$. 
Combining padded decompositions with our mixture metric and its properties as observed above gives our decomposition lemma.%; for the sake of generality we give our results in terms of $\pad$.
\begin{restatable}{lemma}{hopconstrainedDecomp}\label{lem:hopconstrained-decomp}
  Let $G = (V, E, w)$ be a weighted graph with padding parameter $\pad$. For any hop constraint $h > 0$, weight diameter $b > 0$, and exclusion probability $\gamma > 0$, there exists a distribution $\calC$ over partial vertex partitions where for every $C = C_1 \sqcup \ldots \sqcup C_k$ in the support of $\mcC$:
\begin{enumerate}
    \item \textbf{Hop-Constrained Diameter:} $\dist\at{h}_{G}(u, v) \le b$ for $i \in [k]$ and $u,v \in C_i$;
    \item \textbf{Hop-Constrained Paddedness:} $\dist\at{h \frac{\gamma}{2\pad}}_{G}(u, v) \ge b \cdot \frac{\gamma}{2\pad}$ for every $u \in C_i$ and $v \in C_j$ where $i \neq j$.
\end{enumerate}
And:
\begin{enumerate}\setcounter{enumi}{2}
    \item \textbf{Exclusion probability:} $\Pr_{C\sim \calC}[ v \not \in \bigcup_{i \in [k]} C_i ] \le \gamma$ for each $v \in V$ where $C = C_1 \sqcup \ldots \sqcup C_k$;
    \item \textbf{Path preservation:} $\Pr_{C \sim \calC}[V(P)\text{ is broken in } C] \le ( \hop(P)/h + w(P)/b ) \cdot \pad$ for each path $P$.
\end{enumerate}
\end{restatable}
\gdef\hopconstrainedDecompProof{
\begin{proof}
  Let $\dist'$ be the mixture metric of $G$ with hop scale $h$ and weight scale $b$ and let $\Delta := 2 \pad$.
  We first take a (distribution over) $(\pad, \Delta)$-padded decompositions $C' = C_1' \sqcup C_2' \sqcup \ldots \sqcup C_k'$ using $\dist'$ as the underlying metric. Next, we construct $C_i \subseteq C_i'$ by starting with $C_i := C_i'$ and removing all vertices $v \in C_i'$ where $B_{\dist'}(v, 2 \gamma) \not \subseteq C_i'$. Now $\Pr[v \not \in \bigcup_{i \in [k]} C_i] \le \frac{2 \gamma \cdot \pad}{\Delta} \le \gamma$ for each vertex $v \in V$, as stipulated by (3).
  
  Fix $u, v \in C_i$. Since every $C_i'$ has $d'$-diameter at most $\Delta$, there exists a sequence of edges $P = (e_1, e_2, \ldots, e_\ell)$ between $u$ and $v$ whose $\dist'$-length is at most $\Delta$. Therefore:
  \begin{align*}
    \Delta \ge \sum_{i=1}^\ell\left( \frac{\Delta}{h} + \frac{\Delta \cdot w(e_i)}{b} \right) = \frac{\Delta \cdot \hop(P)}{h} + \frac{\Delta \cdot w(P)}{b} .
  \end{align*}
  In other words, $\hop(P) \le h$ and $w(P) \le b$, implying that $d_G\at{h}(u, v) \le b$ for any $u, v \in C_i'$. Therefore, the same claim holds for $u, v \in C_i \subseteq C_i'$, giving (1).

  %, i.e., we can drive the probability of this event down to $\gamma$ by choosing a sufficiently large $O$-notation constant in the definition of $\Delta = O(\log n)$. In other words, $\Pr[v \not \in \bigcup_{i=1}^q C_i] \le \gamma$ for each vertex $v \in V$, as stipulated.

  % Finally, let $u \in C_i$ and $v \not \in P_i$. Since we have that $B_{dist'}(u.t) \subseteq P_i$, we have that $dist'(u,v) > t$. Consequently, it must hold that $\mdist_{\boldw}(u, v) \not \preceq b$, as the converse implies that

  For $u \in C_i$ and $v \in C_j$ where $i \neq j$ we argue that $d\at{\gamma h / \Delta}(u, v) > \gamma b / \Delta$, i.e.\ (2). Suppose for the sake of contradiction that $d_G\at{\gamma h / \Delta}(u, v) \le \gamma b / \Delta$. It follows that there exists a path $P$ with $\hop(P) \le \gamma h / \Delta$ and $w(P) \le \gamma b / \Delta$. However, the $d'$-length of $P$ is at most $\frac{\hop(P) \Delta}{h} + \frac{w(P) \Delta}{b} \le 2 \gamma$. Thus, we have contradicted how we constructed $C_i$ from $C_i'$. Hence $d_G\at{\gamma h / 2\pad}(u, v) > \gamma b / 2 \pad$ since $\Delta = 2 \pad$.

  Finally, consider a path $P$ from $u$ to $v$ and let $\delta' := \hop(P)\Delta/h + w(P)\Delta/b$. If $P$ is broken in $C_1 \sqcup \ldots \sqcup C_k$ then $B_{d'}(u, \delta') \not \subseteq P$. We therefore have (4), namely \[\Pr[P\text{ is broken in }C_1 \sqcup \ldots \sqcup C_k] \le \Pr[B_{\dist'}(u, \delta') \not \subseteq P] < \frac{\delta' \cdot \pad}{\Delta} = \frac{\delta'}{2} \le ( \hop(P)/h + w(P)/b ) \cdot \pad. \qedhere\]
\end{proof}
}
\fullOnly{\hopconstrainedDecompProof}

Lastly, we note that it is known that every metric has padded decompositions with padding parameter $O(\log n)$ and so our decomposition lemma holds with $\pad = O(\log n)$.
%Such decompositions were presented, for example, by \cite{gupta2003bounded, abraham2019cops}.
\begin{lemma}[\cite{gupta2003bounded,abraham2019cops}]
    \label{lemma:abraham-padded}
    Every metric on $n$ points admits a $(\pad,\Delta)$-padded decomposition for $\pad=O(\log n)$ and any $\Delta>0$. Furthermore, such a decomposition can be computed in polynomial time.
\end{lemma}

\subsubsection{Constructing Tree Metrics for Hop-Constrained Distances and the Proof of \Cref{thm:mainmetric}}\label{sec:mainproof}

Next, we recursively apply the random partial vertex partitions of \Cref{lem:hopconstrained-decomp} to obtain a distribution over families of laminar subsets of nodes of $G$. This distribution will naturally correspond to a distribution over well-separated tree metrics which approximate $h$-hop constrained distances. In particular, a rough outline of our construction is as follows: we start with a large weight diameter $\Delta \le \poly(n)$ and hop constraint about $h$ and compute the partial vertex partition $C_1 \sqcup \ldots \sqcup C_k \subseteq V(G)$ of \Cref{lem:hopconstrained-decomp}. We remove from our process any vertices not in our partial vertex partition. We then recurse on each part $C_i$ while keeping our hop constraint constant but shrinking $\Delta$ by a factor of $2$. We combine the recursively constructed trees by hanging the roots of the returned trees off of the root of a fixed but arbitrary tree with edges of length $\Delta$. The recursion stops when each $C_i$ is a singleton. The resulting tree metric is partial since each application of \Cref{lem:hopconstrained-decomp} removes a small fraction of nodes. We illustrate our construction in \Cref{fig:mainEmbed} and proceed to prove \Cref{thm:mainmetric}.

\begin{figure}
    \centering
    \begin{subfigure}[b]{0.24\textwidth}
    \centering
    \includegraphics[width=\textwidth,trim=130mm 0mm 130mm 0mm, clip]{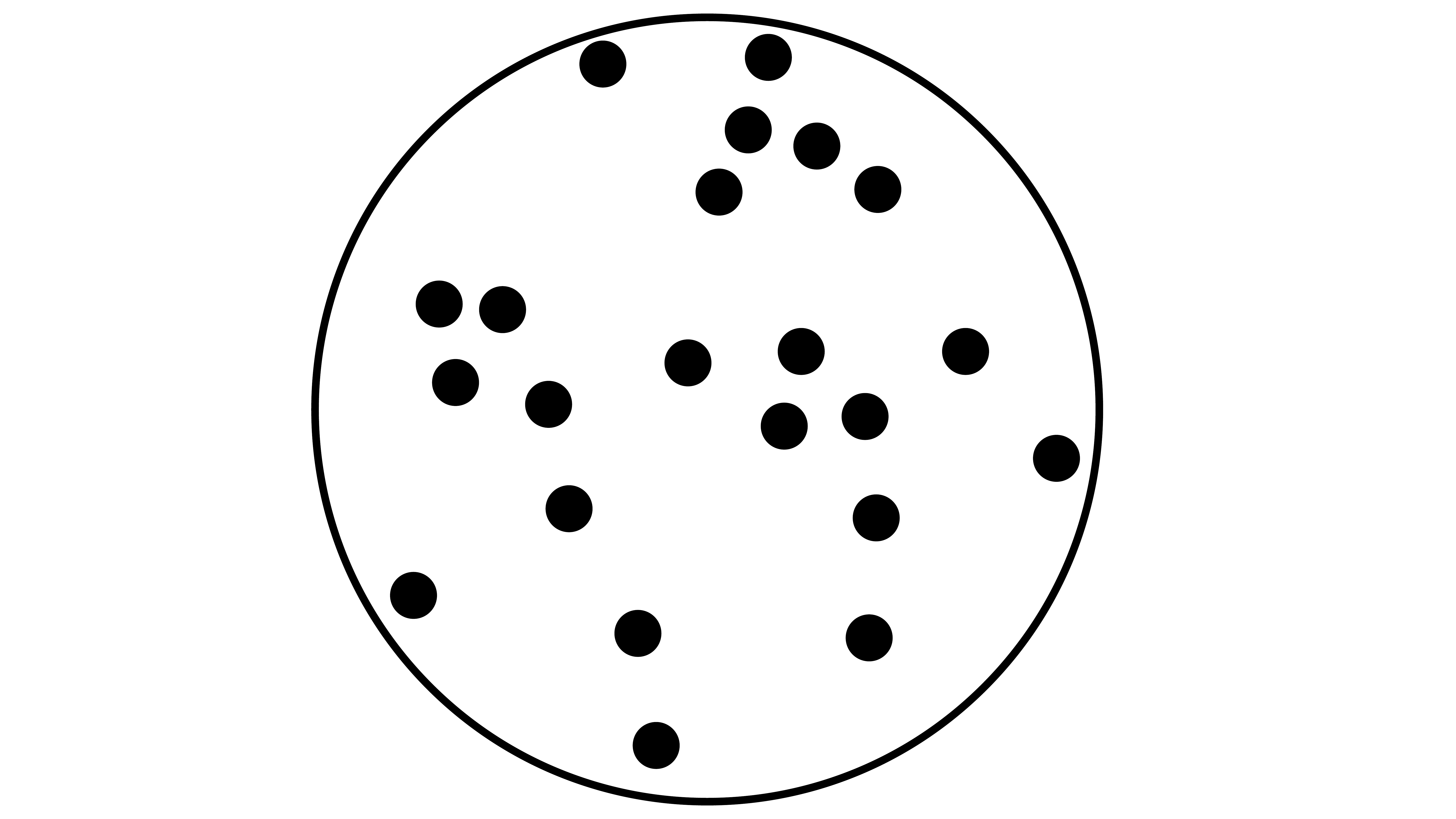}
    \caption{Graph $G$}
\end{subfigure}
    \hfill
    \begin{subfigure}[b]{0.24\textwidth}
    \centering
    \includegraphics[width=\textwidth,trim=130mm 0mm 130mm 0mm, clip]{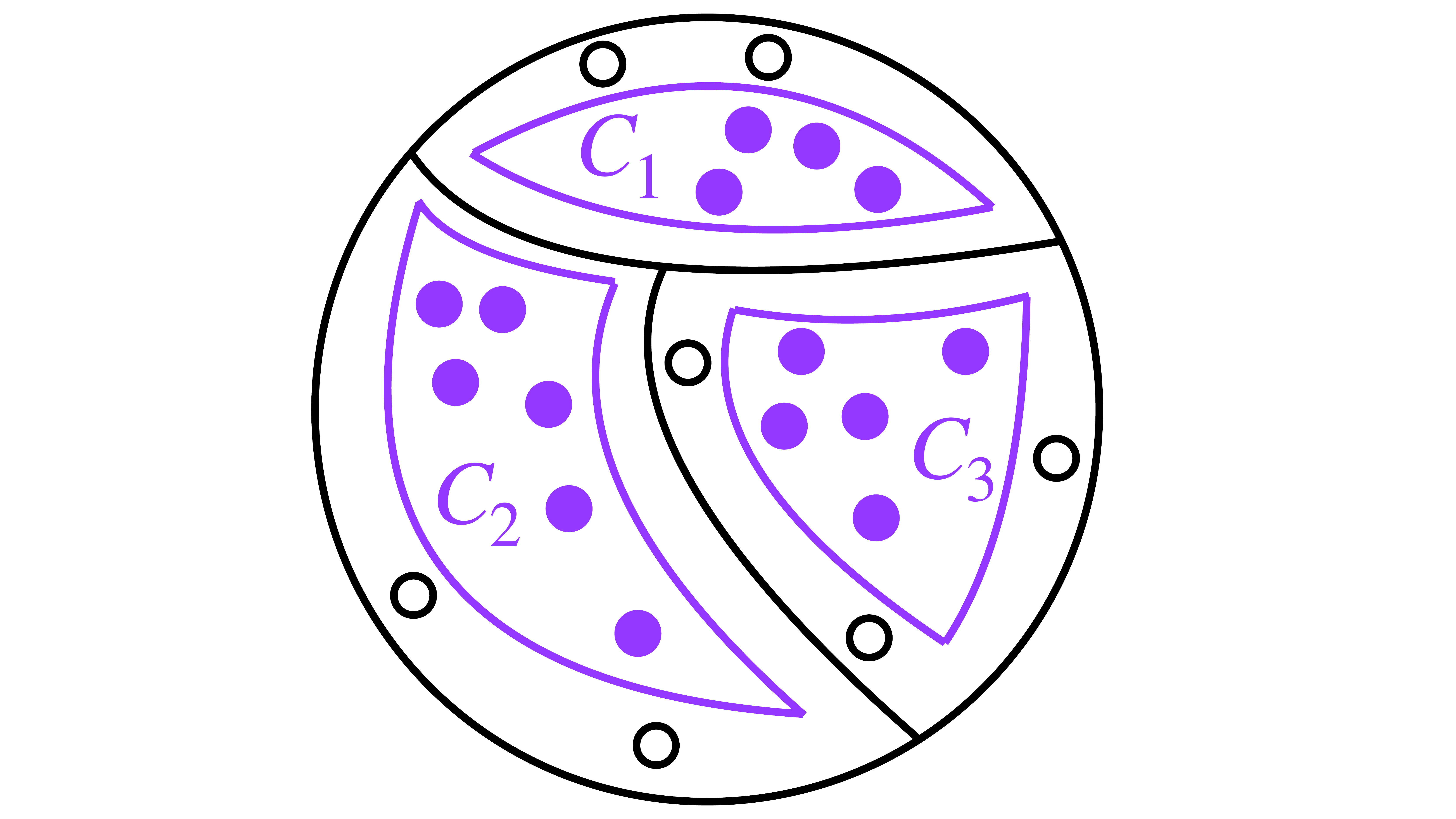}
    \caption{Lem. \ref{lem:hopconstrained-decomp} decomposition}
\end{subfigure}
    \hfill
    \begin{subfigure}[b]{0.24\textwidth}
        \centering
        \includegraphics[width=\textwidth,trim=130mm 0mm 130mm 0mm, clip]{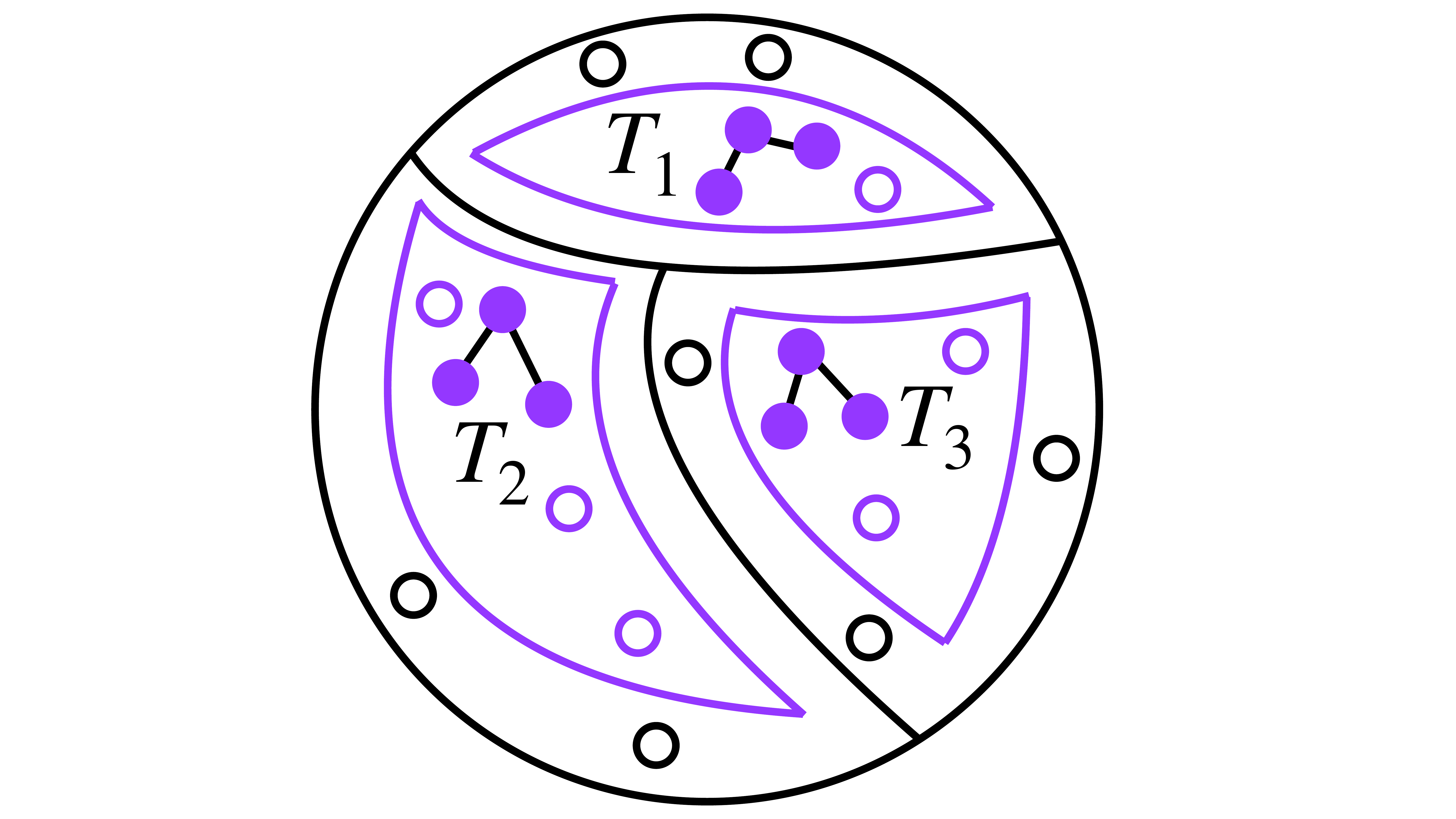}
        \caption{Recursing}
    \end{subfigure}
    \hfill
    \begin{subfigure}[b]{0.24\textwidth}
    \centering
        \includegraphics[width=\textwidth,trim=130mm 0mm 130mm 0mm, clip]{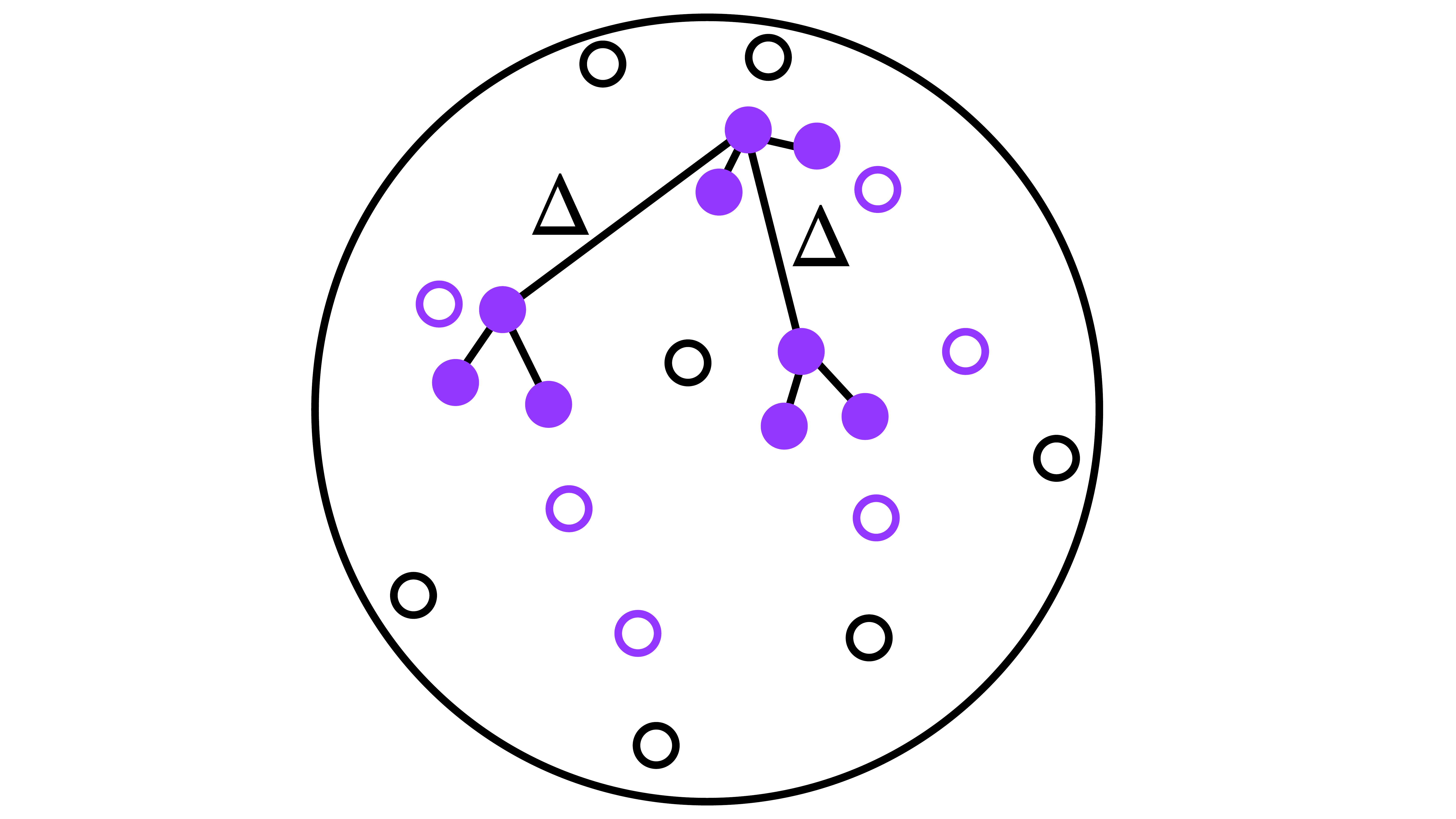}
    \caption{Merging recursions}
\end{subfigure}
    \caption{An illustration of the top-level recursive call of the embedding of \Cref{thm:mainmetric} on graph $G$ (edges omitted from illustration). Vertices in the partial vertex partition of \Cref{lem:hopconstrained-decomp} given in purple. Vertices removed from the process given as empty circles and all other vertices given as filled-in circles.}
    \label{fig:mainEmbed}
\end{figure}

\begin{proof}[Proof of \Cref{thm:mainmetric}]
  % We can assume that the hop diameter of the graph is at most $h$ (actually, at most $2$) by creating a virtual node and connecting it to all other nodes with weight $n^{O(1)}$ for some sufficiently large constant larger than the cost-diameter of $G$.
 
  We describe a recursive and randomized procedure that induces a distribution over well-separated rooted trees where each tree can be interpreted as a partial tree metric with the required properties. Given hop constraint $h'$, weight diameter $\Delta$ and vertex set $V' \subseteq V$ where $d_G\at{h'}(u,v) \leq \Delta$, our procedure returns a rooted tree $(V(T), E(T), w_T)$ satisfying $V(T) \subseteq V'$. Let $\pad$ be the padding parameter of $G$; we will give our proofs in terms of $\pad$ and then conclude by applying \Cref{lemma:abraham-padded}. We fix $h' := h \cdot \kappa$ where we define $\kappa := O(\eps^{-1} \pad \log n)$ throughout the procedure. We emphasize that $h'$ will also be the same for all of our recursive calls. The construction procedure is initially invoked with the parameters $V' := V$ and weight scale $\Delta$ equal to the smallest power of $2$ which is at least the aspect ratio $L \le \poly(n)$. That is, $\Delta \in [L, 2L) \ge \max_{u, v} \dist_G\at{h'}(u, v)$.

  \textbf{Construction procedure:} We use the decomposition of \Cref{lem:hopconstrained-decomp} with hop constraint $h'$, weight diameter $\Delta / 2$, and exclusion probability $\gamma :=\eps/O(\log n)$ (for a sufficiently large hidden constant) to obtain a partial vertex partition $C_1 \sqcup C_2 \sqcup \ldots \sqcup C_k \subseteq V'$ where, plugging in our choice of parameters and the guarantees of \Cref{lem:hopconstrained-decomp}, we have:
  \begin{enumerate}
    \item $\max_{u, v \in C_i} \dist_G\at{h'}(u, v) \le \Delta / 2$;
    \item $\dist_G\at{h}(C_i, C_j) = \dist_G\at{h'/\kappa}(C_i, C_j) \ge \Delta/(2\kappa)$ for each $i, j \in [k]$ where $j \neq i$;
    \item $\Pr[v \not \in \bigcup_{i=1}^k C_i] \le \frac{\eps}{O(\log n)}$ for all $v \in V$.
  \end{enumerate}
  We recursively construct $k$ rooted trees $T_1 = (V_1, E_1, w_1), \ldots, T_k = (V_k, E_k, w_k)$ by calling the same procedure with our distance scale set to $\Delta' \gets \Delta / 2$ on sets $C_1, \ldots, C_k$. We construct the tree $T = (V(T), E(T), w_T)$ returned by the procedure by connecting the roots of $T_2, \ldots, T_k$ to the root of $T_1$ via a tree edge of weight $\Delta$.
  %Specifically, let $r$ be the root of $T_1$ and $r'$ be the root of $T_i$.
  % We associate the $T$-tree edge $\{r, r'\}$, i.e., $T^G_{\{r, r'\}}$ to the shortest $h'$-hop-constrained path in $G$ between $r$ and $r'$.
  % In other words, $\hop(T^G_{\{r, r'\}}) \le h'$ and $d_T(r, r') \le \Delta$. Finally, we let $V(T) := (V_1 \cup \ldots \cup V_k) \cap (C_1 \cup \ldots \cup C_k)$ and return $(V(T), E(T), w_T)$.
  The procedure is stopped when the set of nodes $V'$ is a singleton, at which point the trivial one-node tree is returned.

  % \textbf{Hop stretch analysis.} \gznote{FIX!?!?!?} Suppose that the final returned partial tree embedding is $T = (V(T), E(T), w_T)$. We first note that $T$ is well-separated (since $\Delta$ was divided by $2$ in each level of the recursion). Next, fix some nodes $u, v \in V(T)$. Clearly, there can be only $O(\log n)$ levels of recursion since the aspect ratio is polynomially bounded, and the diameter $\Delta$ drops by a factor of $2$ in each level. Since each tree-edge corresponds to a path of at most $h'$ hops we have that $\hop(T^G_{u,v}) \le O(h' \log n) = h \cdot \kappa \cdot O(\log n) = h \cdot O(\eps^{-1} \pad \log^2 n) = h \cdot O(\eps^{-1} \log^3 n)$, which proves the hop stretch guarantee.
  % % Consequently, $\dist_T(u, v) \ge w_G(T^G_{u,v}) \ge \dist_G\Oat{h \eps^{-1} \log n \log^2 L}(u, v) = \dist_G\Oat{h \log^3 n}$ since $T^G_{u,v}$ corresponds to a valid path in $G$ between $u$ and $v$ of $O(h \log^3 n)$ hops. This proves the lower bound of the distance stretch guarantee.
  \textbf{Exclusion probability analysis:} Consider a recursive call  with $v \in V'$ and suppose that the partial vertex partition in the call is $C_1 \sqcup \ldots \sqcup C_k$. By the properties of the partition, $\Pr[v \not \in \bigcup_{i=1}^k C_i] \le \eps/O(\log n)$ (for a sufficiently large constant). First, we note that $v \not \in V(T)$ if and only if there is a recursive call where $v \in V' \setminus (\bigcup_i C_i)$, which happens with probability $\eps / O(\log n)$. Since $v$ is in a unique recursive call on each level and there are $O(\log n)$ levels, we conclude via a union bound that this happens in at least one level with probability at most $\eps$, proving that the exclusion probability of each node $v \in V$ is a most $\eps$.

  \textbf{Worst-case distance stretch and hop stretch analysis:} In the final tree $T$, for two nodes $u, v \in V(T)$ let $e_{u, v} := \arg\max \{ w_T(e) \mid e \in T_{u, v} \}$ be the heaviest weight tree edge on the unique tree path between $u$ and $v$. The weights $w_T$ are strictly decreasing powers of 2 on any root-leaf path. Therefore, $w_T(e_{u, v}) \le d_T(u, v) \le O(w_T(e_{u, v}))$. Edge $e_{u,v}$ was created via a recursive call with the parameters $V'$ and $\Delta$ where $V' \subseteq V$, $u, v \in V'$ and $d\at{h'}(u', v') \le \Delta = w_T(e_{u,v})$ for all $u', v' \in V'$. Let $C_1 \sqcup \ldots \sqcup C_k$ be the partial vertex partition created by this recursive call where each $C_i$ has weight diameter $\Delta / 2$ and $d\at{h}(C_i, C_j) \ge \Delta / (2\kappa)$ when $i \neq j$. Since $u, v \in V(T)$ we have that $u \in C_i$ and $v \in C_j$ for $i \neq j$ (since otherwise $e_{u, v}$ would not be created by this recursive call), hence $d_G\at{h}(u, v) \ge \frac{\Delta}{2\kappa} = \frac{w_T(e_{u,v})}{2\kappa} = \Theta(\frac{d_T(u, v)}{\kappa}$). Consequently, $d_T(u, v) \le O(\kappa \cdot  \dist_G\at{h}(u, v))$. Furthermore, since $u, v \in V'$ we have that $d\at{h'}(u, v) \le \Delta \le d_T(u, v)$, which can be rewritten as $d\at{\beta h}(u, v) \le d_T(u, v)$ for $\beta := O(\kappa)$. Combining the two bounds on $d_T$ we have that both the worst-case distance stretch $\alphawc$ and hop stretch $\beta$ are $O(\kappa) = O(\eps^{-1} \log n \cdot \pad)$ which gives the desired bound when we plug in the $\pad = O(\log n)$ padded decomposition of \Cref{lemma:abraham-padded}.

  \textbf{Expected distance stretch analysis:} Let $\Delta_l$ be the weight diameter of recursive calls at level $l \in [O(\log n)]$. In particular, $\Delta_1 \in (L, 2L]$ and $\Delta_{l+1} = \Delta_l / 2$. Fix $u, v \in V$, let $P$ be a path in $G$ between $u$ and $v$ with at most $h$ hops and weight $\delta := \dist_G\at{h}(u, v)$ and let $e_{u,v}$ be defined---as in the worst-case distance stretch analysis---as the heaviest weight tree edge between $u$ and $v$. As in the worst-case stretch analysis, it suffices to bound $w_T(e_{u,v})$. We now partition the $O(\log n)$ levels into three phases $H_1 \sqcup H_2 \sqcup H_3$ where $l \in H_1$ iff $\Delta_l > \delta \cdot (2\kappa)$, $l \in H_3$ iff $\Delta_l \le \delta \cdot (2 \pad)$ and $l \in H_2$ in the remaining case where $\Delta_l \in (\delta \cdot 2\pad, \delta \cdot 2\kappa]$. We proceed to bound the probability that $e_{u,v}$ is created by a recursive call in $H_1$, $H_2$ and $H_3$ which, in turn, gives a bound on the expected distance between $u$ and $v$.
 
 We begin with calls at levels in $H_1$. In particular, we argue that a call at level $l \in H_1$ cannot create the edge $e_{u, v}$ (i.e., it cannot be that $\Delta_l = w_T(e_{u, v})$). This follows from the worst-case distance stretch analysis, which stipulates that $d_G\at{h}(u, v) \ge \Delta_l / (2\kappa)$. However, this would yield $d_G\at{h}(u, v) > \delta$, which is a contradiction. Therefore, the contribution of edges corresponding to levels in $H_1$ to $w_T(e_{u,v})$ is $0$:
  \begin{align*}
    \sum_{l \in H_1} \Pr[e_{u,v} \text{ created by level $l$ call}] \cdot \Delta_l \cdot \1{u, v \in V(T)} = 0
  \end{align*}
  
  Next, suppose that $l \in H_2$ and suppose $e_{u, v}$ was created via a level $l$ call with the vertex set $V'$ and partial vertex partition $C_1 \sqcup \ldots \sqcup C_k \subseteq V'$. If this is the case, the path $P$ between $u$ and $v$ is broken in $C_1 \sqcup \ldots \sqcup C_k$, which by \Cref{lem:hopconstrained-decomp} happens with probability at most
  \begin{align*}
    \pad\left(\frac{\hop(p)}{h'} + \frac{w_G(p)}{\Delta_l/2}\right) \le \pad \left(\frac{h}{h'} + \frac{\delta}{\Delta_l / 2}\right) = \frac{\pad}{\kappa} + \frac{\pad \delta}{\Delta_l/ 2}.
  \end{align*}
  Moreover, note that $|H_2| = O(\log ( \kappa / \pad )) = O(\log ( \eps^{-1} \log n ) )$ since $\Delta_{l+1} = \Delta_l / 2$. Therefore:
  \begin{align*}
    \sum_{l \in H_2} \Pr[e_{u,v} \text{ created by level $l$ call}] \cdot \Delta_l \cdot \1{u, v \in V(T)} & \le \sum_{l \in H_2} \left( \frac{\pad}{\kappa} + \frac{\pad\delta}{\Delta_l/2} \right) \cdot \Delta_l \cdot 1 \\
                                                                                                           & \le \sum_{l \in H_2} \left( \Delta_l \cdot \frac{\pad}{\kappa} + 2\pad\delta \right) \\
                                                                                                           & \le \frac{\pad}{\kappa} \cdot (\delta \cdot 2\kappa) + 2\pad\delta |H_2| ) \\
                                                                                                           & \le \delta \cdot O(\pad \log(\eps^{-1} \log n)).
  \end{align*}
  
  Lastly, for $H_3$ notice that we can coarsely upper bound $\sum_{l \in H_3} \Pr[e_{u,v} \text{ created by level $l$ call}] \cdot \Delta_l \cdot \1{u, v \in V(T)}$ as $\sum_{l \in H_3} \Delta_l \leq \delta \cdot 4 \pad$ by our choice of $H_3$ and the fact that our weight diameters are geometrically decreasing.
  
  Combining our upper bounds on the probability that $e_{u,v}$ is created in each level gives an upper bound on the expectation of $w_T(e_{u, v})$, which in turn bounds the expected value of $d_T(u, v)$ since $\dist_T(u, v) = O(w_T(e_{u, v}))$. In the following we let $(\ldots)$ stand for $\Pr[e_{u,v} \text{ created by level $l$ call}] \cdot \Delta_l \cdot \1{u, v \in V(T)}$.
  \begin{align*}
    \E[w_T(e_{u, v}) \cdot \1{u, v \in V(T)}] & \le \sum_{l=1}^{O(\log n)} \Pr[e_{u,v} \text{ created by level $l$ call}] \cdot \Delta_l \cdot \1{u, v \in V(T)} \\
                                        & \le \sum_{l \in H_1} (\ldots) + \sum_{l \in H_2} (\ldots) + \sum_{l \in H_3} (\ldots) \\
                                        & \le 0 + \delta \cdot O(\pad \log(\eps^{-1} \log n)) + \delta \cdot (4\pad) \\
                                        & = \delta \cdot O(\pad \log(\eps^{-1} \log n)) 
  \end{align*}
  Plugging in the padded decompositions of \Cref{lemma:abraham-padded}, we conclude that the expected distance stretch is $O(\pad \log(\eps^{-1} \log n)) = O(\log n \log(\eps^{-1} \log n))$, as required.
\end{proof}

\section{$h$-Hop Partial Tree Embeddings}\label{sec:hop-bounded-hsts}
In the preceding section we demonstrated that hop-constrained distances can be well-approximated by distributions over partial tree metrics. In this section we describe how this result gives embeddings which can be used for hop-constrained network design problems. In particular, in \Cref{sec:projection-injection} we will define $h$-hop partial tree embeddings which are partial tree metrics along with a mapping of each edge in the tree metric to a path in $G$. As an (almost) immediate corollary of our results in the previous section, we have that one can produce such an embedding where $h$-hop distances are approximately preserved by $T$ and each path to which we map an edge has a low number of hops and less weight than the corresponding edge in $T$.

However, ultimately we are interested in using these embeddings to instantiate the usual tree embedding template and the above properties alone are not sufficient to do so. In particular, recall that in the usual tree embedding template for network design we embed our input graph into a tree, solve our problem on the tree and then project our solution back onto the input graph. If the problem which we solve on the tree has a much greater cost than the optimal solution on our input graph then our solution has no hope of being competitive with the optimal solution. Thus, we require some way of projecting the optimal solution of $G$ onto our embeddings in a way that produces low-cost, feasible solutions for our tree problems.

When tree embeddings are not partial---as in FRT---such a projection is trivial. However, the partial nature of our embeddings along with the fact that we must preserve ``$h$-hop connectivity'' makes arguing that such a low cost solution exists significantly more challenging than in the FRT case. Somewhat surprisingly, we show that a natural projection of the optimal solution onto $T$ produces an appropriate subgraph of $T$, despite the fact that an FRT-like charging argument seems incapable of proving such a result. Our proofs will be based on what may be viewed as a hop-constrained version of Euler tours which we call $h$-hop connectors. We give further intuition and details in \Cref{sec:cost-preservence}. Thus, while \Cref{sec:projection-injection} is a straightforward extension of our results from the previous section, the primary technical contribution of this section is the projection result of \Cref{sec:cost-preservence} which shows that, indeed, these embeddings may be used for tree-embedding algorithms in the usual way.

\subsection{Defining $h$-Hop-Partial Tree Embeddings}\label{sec:projection-injection}

%for a solution $H \subseteq G$ it is typically sufficient to map every edge $\{u, v\}\in E(H)$ to the unique path between $u$ and $v$ on the tree $T$ and taking the union over all edges $E(H)$.
%
%tree, resulting in a subtree $F \subseteq F$, and every subtree $F' \subseteq T$  can be ``injected'' back to the graph, resulting in a subgraph $H'$. Furthermore, (2) projecting and injecting a subgraph/subtree leaves the weight of subgraph (approximately) unchanged.
%
%the cost of the projection of $H$ and $H$approximates the cost of
%
%and (2) the projection 
%
%exists a subtree $F \subseteq T$ and (2) $F$ approximates the optimal solution (at least in expectation).

% What do I want to say here.. in classic applications we have the following?
% Graph G -> tree T such that for any subgraph H\subseteq G there exists a subtree F\subseteq 

We begin by defining our partial tree embeddings and proceed to argue that we can map from the trees in these embeddings to our graphs in a weight and connectivity-preserving fashion.

\begin{definition}[Partial Tree Embedding]\label{def:PartialTreeEmbedding}
A partial tree embedding on weighted graph $G = (V(G), E(G), w_G)$ consists of a rooted and weighted tree $T = (V(T), E(T), w_T)$ with $V(T) \subseteq V(G)$ and a path $T_e^{G} \subseteq G$ for every $e \in E(T)$ between $e$'s endpoints satisfying $w_G(T^G_e) \le w_T(e)$.
\end{definition}

We extend the notation from \Cref{def:PartialTreeEmbedding} to nodes in $T$ which are not adjacent: for any two vertices $u, v \in V(T)$, if $e_i$ is the $i$th edge in $T_{uv}$ (ordered, say, from $u$ to $v$) then $T_{uv}^G := T_{e_1}^G \oplus T_{e_2}^G \oplus \ldots$ where $\oplus$ is concatenation.

%Any partial tree metric $d$ approximating $d_G\at{h}$ on $V_d$ with hop-stretch $\beta$ and distance stretch $\alpha$ guarantees that any two nodes $u,v \in V_d$ can be mapped to a path in $G$ with at most $\beta h$ hops and distance at most $\alpha d_G\at{h}(u,v)$. Most applications, including all the ones presented in this paper, require some additional structure on these paths and a more explicit and formal way to talk about the mappings between the trees and $G$. This section introduces the notation necessary for this. \Cref{sec:connectivity} then shows how to map general edge sets in $G$ onto the trees in a cost and connectivity preserving way...

We now define hop and distance stretch of partial tree embeddings analogously to how we defined these concepts for partial metrics.

\begin{definition}[$h$-Hop Partial Tree Embedding]\label{def:HopAndDistanceStretch}
A partial tree embedding $(T, \{T_e^G\}_{e \in E(T)})$ is an $h$-hop partial tree embedding with \textbf{distance stretch} $\alpha \geq 1$ and \textbf{hop stretch} $\beta \ge 1$ for graph $G = (V(G), E(G), w_G)$ if
\begin{enumerate}
    \item $d_G\at{\beta h} \leq d_T(u, v) \le \alpha \cdot d^{(h)}(u,v)$ for all $u, v \in V(T) \subseteq V(G)$;
    \item $\hop(T^G_{uv}) \le \beta h$ for all $u, v \in V(T) \subseteq V(G)$.
\end{enumerate}
\end{definition}
Notice that the above definitions show that one can map subgraphs of a partial tree embedding $(T, \{T_e^G\}_{e \in E(T)})$ for $G$ to subgraphs of $G$ in a cost and connectivity preserving way. In particular, given a $T' \subseteq T$ we have that $H := \bigcup_{e \in E(T')}T_e^G$ satisfies (1) $w_G(H) \leq w_T(T')$ and (2) if $u$ and $v$ are connecting in $T'$ then $\hop_H(u,h) \leq \beta h$. In the next section we give a much more involved and interesting proof showing that one can also project from subgraphs of $G$ to $T$ in a cost and connectivity preserving way.

The next observation confirms that, up to an $O(\log n)$, hop stretch and distance stretch for $h$-hop partial tree embeddings and partial metrics are equivalent, provided the relevant trees are well-separated.

\begin{lemma}\label{lem:equivStretch}
    Let $G$ be a weighted graph and let $h \geq 1$ be a hop constraint. 
    \begin{itemize}
        \item If $(T, \{T_e^G\}_{e \in E(T)})$ is a partial tree embedding with distance stretch $\alpha$ and hop stretch $\beta$ then $T$ is a partial tree metric which approximates $d_G\at{h}$ with distance stretch $\alpha$ and hop stretch $\beta$.
        \item Conversely, if $T$ is a partial tree metric with hop diameter $D_T := \hop(T)$ which approximates $d_G\at{h}$ with distance stretch $\alpha$ and hop stretch $\beta$ then there is a collection of paths $\{T_e^G\}_{e \in E(T)}$ where $(T, \{T_e^G\}_{e \in E(T)})$ is a partial tree embedding with distance stretch $\alpha$ and hop stretch $D_T \cdot \beta$.
    \end{itemize}
\end{lemma}

%\begin{lemma}
%For any weighted graph $G = (V(G), E(G), w_G)$, any hop constraint $h \ge 1$, and any $h$-hop partial tree embedding $T$ with hop stretch $\beta$ and distance stretch $\alpha$ the weighted tree $T$ is always also a partial tree metric of $V(G)$ which approximates $d_G\at{h}$ with hop stretch $\beta$ and distance stretch $\alpha$ on $V(T)$.
%
%On the other hand, any partial tree metric of $V(G)$ given by a weighted tree $T'$ which approximates the $h$-constrained distances $d_G^{(h)}$ of $G$ on $V(T')$ with hop stretch $\beta$ and distance stretch $\alpha$ has a natural edge mapping function $T'^G$ under which $T'$ becomes an $h$-hop partial tree embedding with hop stretch $2\cdot \text{depth}(T') \cdot \beta$ and distance stretch $2\cdot \text{depth}(T') \cdot \alpha$. If $T'$ is well-separated then it can be naturally completed into an $h$-hop partial tree embedding with hop stretch $O(\log n) \cdot \beta$ and distance stretch $O(\alpha)$.  
%\end{lemma}
\begin{proof}
Let $(T, \{T_e^G\}_{e \in E(T)})$ be a partial tree embedding. Then we immediately have that $T$ is a partial tree embedding with distance stretch $\alpha$ and hop stretch $\beta$ by definition of a partial tree embedding and partial tree metric.

%For $u,v \in V(T) \subseteq V(G)$ the length of $T^G_{uv}$ in $G$ is exactly the sum of the lengths of paths $T_e^G$ for the tree edges $e \in E(T)$ on the unique path between $u$ and $v$ in $T$. The length $w_G(T^G_e)$ of each of these $T_e^G$ paths is, by \Cref{def:PartialTreeEmbedding}, upper bounded by $w_T(e)$ and $d_T(u,v)$ is the sum of these upper bounds. Overall $T^G_{uv}$ is a path in $G$ with length at most $d_T(u,v)$ in $G$. Since $T$ is an $h$-hop partial embedding with hop-stretch $\beta$ we also have that $\hop(T^G_{uv}) \le \beta h$. The path $T^G_{uv}$ therefore certifies that $d_G^{(\beta h)}(u,v) \leq d_T(u,v)$. Similarly, $d_T(u, v) \le \alpha \cdot d^{(h)}(u,v)$ is directly guaranteed by \Cref{def:HopAndDistanceStretch} and so $T$ is indeed the appropriate partial tree metric.

On the other hand, let $T$ be a partial tree metric which approximates the $h$-hop constrained distances $d_G^{(h)}$ of $G$ on $V(T)$ with distance stretch $\alpha$ and hop stretch $\beta$. By definition, for every edge $e \in E(T)$ with $e =\{u,v\}$ we have $d_G^{(\beta h)}(u,v) \leq w_T(e)$. In particular, there exists a path between the endpoints of $e$ with at most $\beta h$ hops and length at most $w_T(e)$ in $G$. Defining $T_e^G$ to be this path for every edge $e \in E(T)$ completes $T$ into a partial tree embedding. The distance stretch of this partial tree embedding is trivial by definition. Similarly, for $u$ and $v$ not adjacent in $T$ we have, by definition of $T_{uv}^G$ that $T_{uv}^G$ consists of at most $\beta D_T \cdot h$ hops as required.
\end{proof}

%\gznote{notation of definition 1} Notice that a $h$-hop partial tree embedding $T$ with hop stretch $\beta$ and deterministic distance stretch $\alphawc$ has $d^{(\beta h)}(u,v) \le d_T(u,v) \le \alphawc \cdot d^{(h)}(u,v)$ for all $u, v \in V(T)$ since $d_T(u, v) = w_T(T_{u,v}) \ge w_G(T^G_{u, v}) \ge d_G\at{\hop(T^G_{u, v})}(u, v) \ge d_G\at{\beta h}(u, v)$.

%We also consider distributions $\calT$ over partial embeddings.
% First, we say that a distribution has $h$-hop stretch $\beta$ or deterministic distance stretch $\alphawc$ if each tree embedding in the support $T \in \supp(\calT)$ has those properties.
%We also define a few quality functions that only apply to distributions.

Analogously to our results for partial metrics, we will also talk about the exclusion probability of distributions over partial tree, the distances they induce and how well they approximate hop-constrained distances; in particular, the following definitions are analogous to \Cref{lem:partialMetricDist} and \Cref{def:partialmetricapproximation} respectively.
%\Cref{lem:equivStretch} guarantees that a distribution over partial tree embeddings corresponds to a distribution over partial tree metrics. Specifically, a distribution $\mcD$ over partial tree embeddings in which $(T, \{T_e^G\}_e)$ has probability $p_T$ corresponds to the distribution over partial tree metrics where $T$ has probability $p_T$. Thus, we may extend our definitions regarding distributions over partial tree metrics to distributions over partial tree embeddings. 
For the sake of presentation, here and later in the paper we let $(T, \cdot)$ be shorthand for $(T, \{T_e^G\}_{e \in E(T)})$.

\begin{definition}[Distances of Partial Tree Embedding Distributions]
    Let $\mathcal{D}$ be a distribution of partial tree embeddings on weighted graph $G = (V, E, w)$. We say $\mathcal{D}$ has \textbf{exclusion probability} $\eps$ if for all $v \in V$ we have $\Pr_{(T, \cdot) \sim \mathcal{D}}[v \in V(T)] \geq 1 - \eps$. If $\epsilon \leq \frac{1}{3}$ then we say that $\mathcal{D}$ induces the distance function $d_{\mathcal{D}}$ on $V$, defined as
    $$d_{\mathcal{D}}(u,v) :=  \E_{(T, \cdot) \sim \mathcal{D}}\left[d_T(u,v) \cdot \1{ u,v \in V(T)} \right ].$$
\end{definition}

\begin{definition}[Stretch of Partial Tree Embedding Distribution]
    A distribution $\mathcal{D}$ of $h$-hop partial tree embeddings on $V$ with exclusion probability at most $\frac{1}{3}$ approximates $d^{(h)}$ on weighted graph $G=(V,E,w)$ for hop constraint $h \ge 1$ with \textbf{worst-case distance stretch} $\alpha_{WC} \ge 1$ and \textbf{hop stretch} $\beta \ge 1$ if each $(T, \cdot)$ in the support of $\mathcal{D}$ approximates $d^{(h)}_G$ on $V(T)$ with distance stretch $\alpha_{WC}$ and hop stretch $\beta$, i.e. for each $(T, \cdot)$ in the support of $\mcD$ and all $u,v \in V(T)$ we have \[d_G^{(\beta h)}(u,v) \leq d_T(u,v) \leq \alpha \cdot d_G^{(h)}(u,v).\]
    Furthermore, $\mcD$ has \textbf{expected distance stretch} $\alpha_{\E}$ if for all $u,v \in V$ we have \[d_{\mathcal{D}}(u,v) \leq \alpha_{\E} \cdot d_G^{(h)}(u,v).\]
\end{definition}

%\begin{definition}
%    Let $\mcD$ be a distribution over partial tree embeddings and let $\mcD'$ be the corresponding distribution over partial metrics. We say that $\mathcal{D}$ has exclusion probability $\eps$, worst-case distance stretch $\alpha_{WC}$, expected distance stretch $\alpha_{\E}$ and hop-stretch $\beta$ iff $\mcD'$ does.
%\end{definition}

Concluding, we have that there exists an efficiently-computable distribution over partial tree embeddings with poly-logarithmic stretches.

\begin{theorem}\label{thm:mainEmbed}
  Given weighted graph $G = (V, E, w)$, $0 < \epsilon < \frac{1}{3}$ and root $r \in V$, there is a poly-time algorithm which samples from a distribution over $h$-hop partial tree embeddings whose trees are well-separated and rooted at $r$ with exclusion probability $\eps$, expected distance stretch $\alpha_{\E} = O(\log n \cdot \log \frac{\log n}{\epsilon})$, worst-case distance stretch $\alpha_{WC} = O(\frac{\log ^ 2 n}{ \epsilon})$ and hop stretch $\beta = O(\frac{\log ^ 3 n}{\epsilon})$.
\end{theorem}
\begin{proof}
  We begin by remarking that \Cref{thm:mainmetric} can be adapted so that all trees are rooted at $r$ in the following way. First, we can assume that $r \in V(T)$ by resampling trees until $r$ is in $V(T)$. By a union bound, this increases the exclusion probability by a factor of at most 2, leaves the hop stretch and worst-case distance stretch unchanged, and increases the expected distance stretch by a factor of at most $\frac{1}{1-\eps} = O(1)$; these modifications to our sampling process leave the statement of our theorem unchanged. 
  
  Now, suppose that a sampled tree has $r \in V(T)$; we will observe that $r$ can be assumed to be the root of $T$. In particular, recall that in the construction of $T$ in \Cref{thm:mainmetric} we recursively constructs trees $T_1, \ldots, T_k$ on the parts of a partial vertex partition and then outputs a tree by connecting the root of $T_2, \dots, T_k$ to the root of $T_1$. We note that $T_1$ is chosen arbitrarily, and so we can choose $T_1$ to be the tree containing $r$. Since we may assume inductively that $r$ is the root of $T_1$, the tree we return has $r$ as its root. Choosing a root in this way does not change the guarantees of our partial tree metrics.
    
  Our result then follows immediately from the fact that well-separated trees have hop diameter $O(\log n)$, \Cref{lem:equivStretch}, \Cref{thm:mainmetric} and the observation that the construction procedures of \Cref{thm:mainmetric} and \Cref{lem:equivStretch} are poly-time.
\end{proof}

\subsection{Projecting From The Graph to $h$-Hop Partial Tree Embeddings}
\label{sec:cost-preservence}
In this section we show how to project the optimal solution for a hop-constrained problem onto a partial tree embedding to get a low-cost subgraph which will be feasible for the optimization problems on trees which we later solve. In particular, we show that it is possible to project \emph{any} subgraph $H \subseteq G$ onto an $h$-hop partial tree embedding $(T, \cdot)$ with worst-case distance stretch $\alpha$ in a way that $\alpha$-approximately preserves the cost of $H$ and preserves ``$h$-hop connectivity'': that is, the projection of $H$ will have cost at most $O(\alpha \cdot w_G(H))$ and if $u$ and $v$ are within $h$ hops in $H$ then they will be connected by the projection of $H$ onto our embedding.

In the (non-partial) tree embedding setting where we typically only care about the connectivity structure of nodes---as in FRT---such a projections is trivial. In particular, if $T$ is a tree drawn from the FRT distribution then an edge $e \in E(G)$ can be projected onto the simple tree path $T_{uv} \subseteq T$ between $u$ and $v$ in $T$ and the resulting path will have expected weight $O (\log n \cdot w_G(e))$. Thus, we can project a subgraph $H \subseteq G$ to $T(H) := \bigcup_{\{u,v\} \in E(H)} T_{uv}$. If $u$ and $v$ are connected in $H$ then they are connected in $T(H)$ and so the connectivity of nodes is preserved. Moreover, we can upper bound the weight of $T(H)$ by summing up $w_T(T_{uv})$ over all $\{u,v\} \in E(H)$ to get that, in expectation, $w_T(T(H)) \leq O(\log n \cdot w_G(H))$ and so the cost of the projection is appropriately low.

We might naturally try to use the same projection as is used in the FRT case but only for the nodes embedded by $T$. Specifically, suppose that $T$ is now the tree of a partial tree embedding with worst-case distance stretch $\alpha$. Then, we could project $H$ to $T(H) := \bigcup T_{uv}$ where  the $\bigcup$ is taken over all $u,v$ such that $\{u,v\} \in E(H)$ \emph{and} $ u, v \in V(T)$. Although we trivially have that $w_T(T(H)) \leq \alpha \cdot w_G(H)$ by summing up over edges in $E(H)$, such a projection has no hope of preserving $h$-hop-connectivity as required: if, for example, $u$ and $v$ are connected by exactly one path in $H$ with $h$ hops then if there is even a single node along this path which is not in $V(T)$ then $u$ and $v$ may not be connected in $T(H)$.

We could try to fix these connectivity issues by forcing all vertices in $T$ which are within $h$ hops in $H$ to be connected in $T$ as captured by the following definition.
\begin{definition}[$T(H, h)$]\label{dfn:proj}
    Let $(T, \cdot)$ be a partial tree embedding. Then $T(H, h) := \bigcup T_{uv}$ where the $\bigcup$ is taken over $u,v$ such that $u,v \in V(T)$ and $\hop_H(u,v) \leq h$.
\end{definition}
$T(H, h)$ trivially preserve $h$-hop connectivity as needed: if $u$ and $v$ are connected by an $h$-hop path in $H$ then they will be connected in $T(H, h)$. However, while $T(H, h)$ preserves $h$-hop connectivity, it seems to yield a subgraph of $T$ of potentially unboundededly-bad cost. For example, let $h = 3$ and suppose $H$ is a spider graph with $O(n)$ nodes in which one leg connects vertex $r$ to center $c$ with a cost $1$ edge and the remaining $i$th leg connects $c$ to $u_i$ to $v_i$ with a sufficiently small $\epsilon > 0$ cost edge. Further, suppose that $V(T)$ consists of $r$ and all $v_i$. This example is illustrated in \Cref{fig:spider}. $T(H, h)$ will buy $T_{rv_i}$ for every $i$ since there is an $h$-hop path from $r$ to $v_i$  (all such pairs illustrated in \Cref{subfig:congestionSpiderPairs}). Our worst-case distance guarantee ensures that $w_T(T_{rv_i})$ is at most $\alpha \cdot d_G(v_i, r) \approx \alpha$ and so we might hope to bound the cost of $T(H, h)$ as within $O(\alpha)$ times $w_G(H)$. However, if we try to apply the usual FRT-type proof and upper bound the cost of $T(H, h)$ in $T$ as $\sum d_T(r, v_i)$ then our sum comes out to $O(\alpha \cdot n)$. On the other hand, $w_G(H) \approx 1$ and so $w_T(T(H, h))$ is a factor of $O(n \cdot \alpha)$ larger than $w_G(H)$ while we would like it to only be an $O(\alpha)$ factor larger. Thus, whereas FRT can charge each path in the projection of $H$ to a unique edge of $H$, the partialness of our embedding means that we must charge paths in $T(H, h)$ to paths in $H$. These paths in $H$ may induce large congestion---as illustrated in \Cref{subfig:congestionSpider}---which causes us to ``overcharge'' edges of $H$.
\begin{figure}
        \begin{subfigure}[b]{0.32\textwidth}
        \centering
        \includegraphics[width=\textwidth]{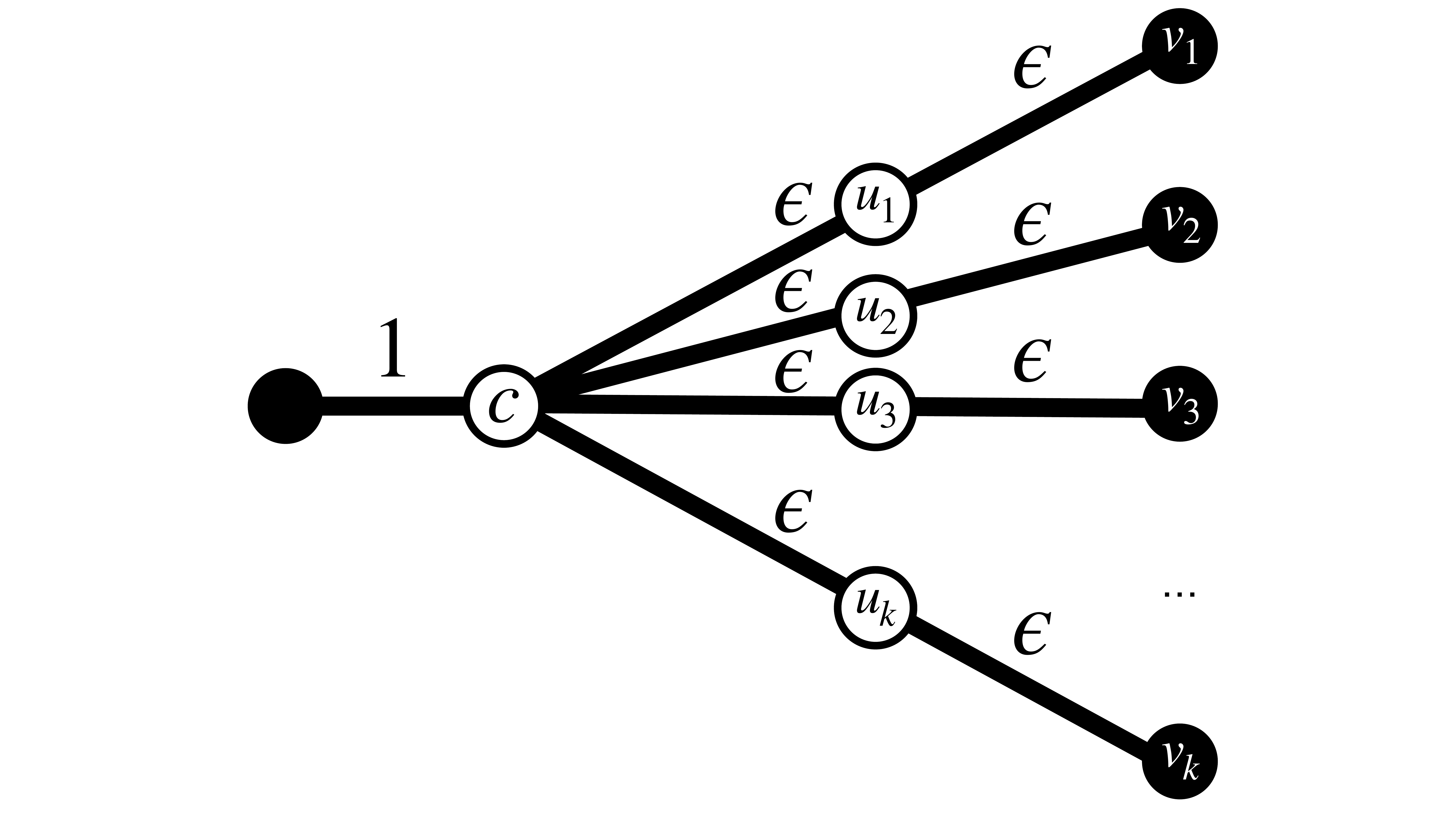}
        \caption{Graph $H$}\label{fig:spider}
    \end{subfigure}
    \hfill
    \begin{subfigure}[b]{0.32\textwidth}
        \centering
        \includegraphics[width=\textwidth]{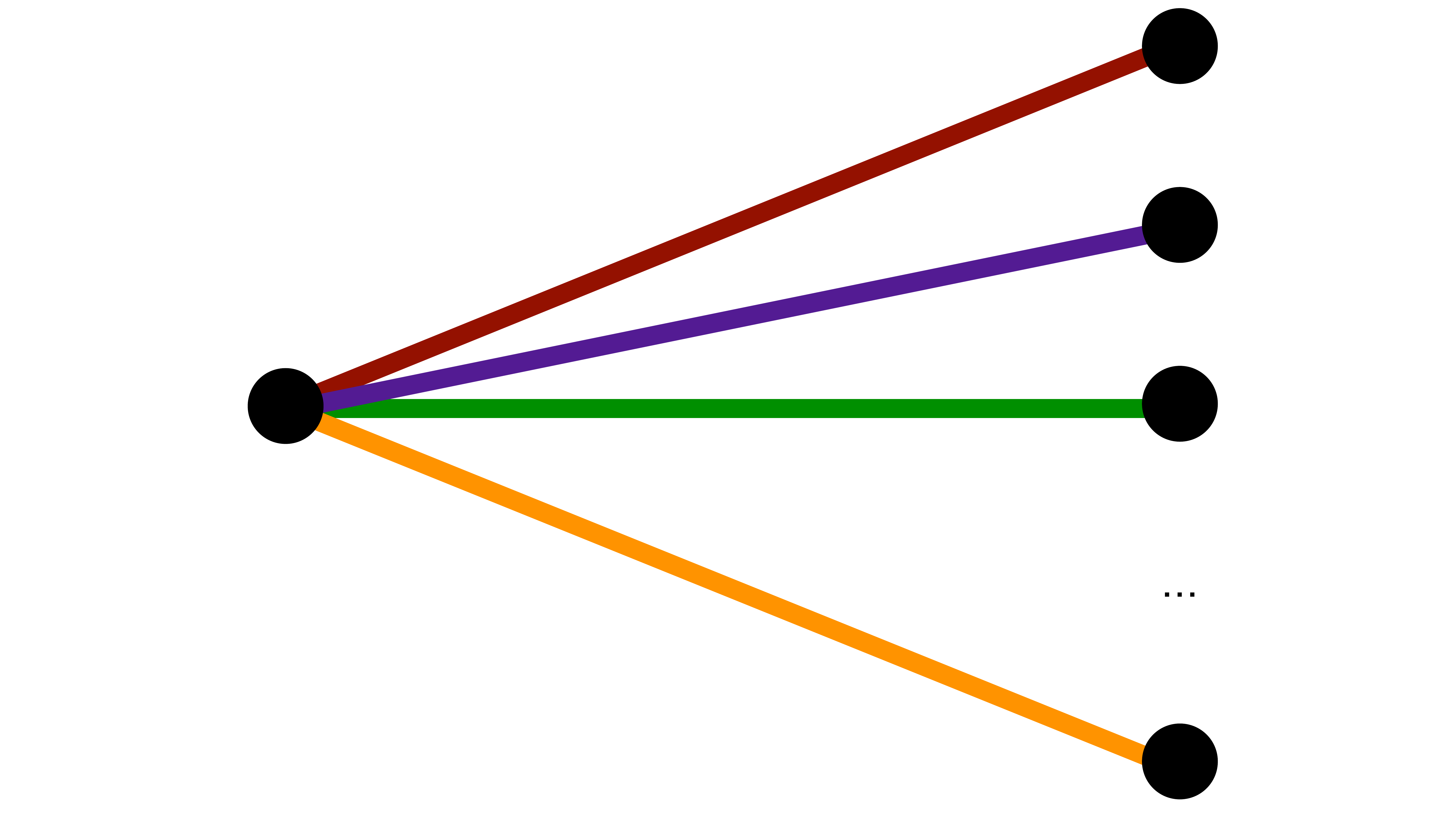} 
        \caption{Pairs in Charged Sum}\label{subfig:congestionSpiderPairs}
    \end{subfigure}
    \hfill
    \begin{subfigure}[b]{0.32\textwidth}
        \centering
        \includegraphics[width=\textwidth]{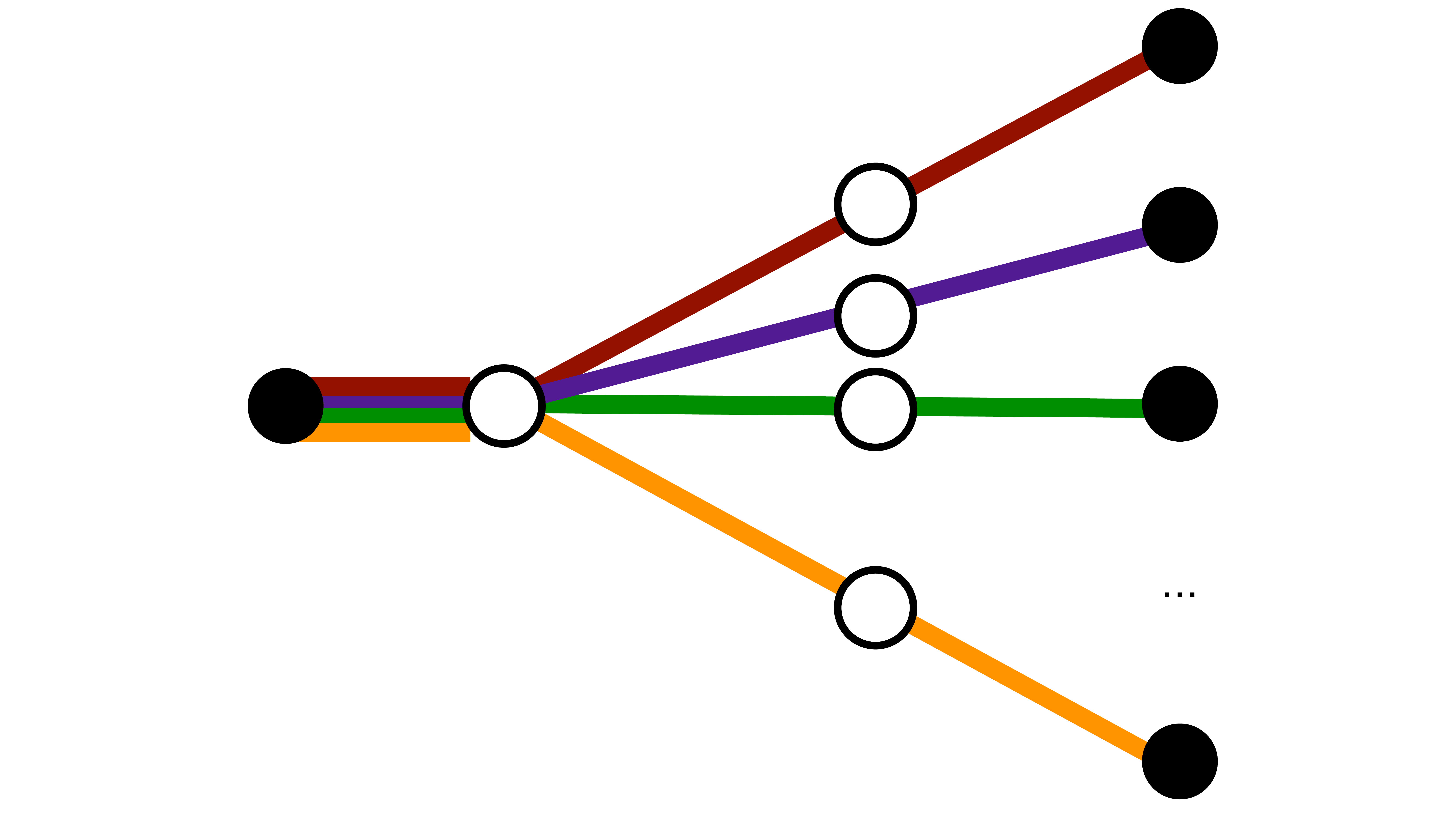} 
        \caption{Congestion of Charged Paths}\label{subfig:congestionSpider}
    \end{subfigure}
    \centering
    \caption{A counter-example to the naive charging argument for $T(H, h)$ where $\Theta(k) = \Theta(n)$. Edges labeled with their weights, vertices of $V(T)$ given as solid black circles and vertices of $V \setminus V(T)$ given as white-filled circles. Paths colored according to their corresponding pair.}
    
\end{figure}

Surprisingly, in what follows we show that, while the above naive charging argument cannot succeed, a more nuanced proof shows that the above $T(H, h)$ is, in fact, competitive with the optimal solution up to small constants in the hop and distance stretch.

\begin{restatable}{theorem}{projectionStruct}\label{thm:projSoln}
    Fix $h \geq 1$, let $H$ be a subgraph of weighted graph $G = (V, E, w_G)$ and let $(T, \cdot)$ be an $8h$-hop partial tree embedding of $G$ with worst-case distance stretch $\alpha$. Then $w_T(T(H,h)) \leq 4 \alpha \cdot w_G(H)$.
\end{restatable}
The basic idea of our proof will be to identify a collection of low congestion paths in $H$ to which we can charge $T(H, h)$.

\subsubsection{Warm-Up: Low Diameter Tree Case}
To illustrate this idea we begin by showing how to prove \Cref{thm:projSoln} in the simple case where $G$ is a tree with diameter at most $h$. In particular, on a tree of diameter at most $h$ we can mitigate the congestion of charged paths by buying an Euler tour restricted to our embedded nodes; conveniently $T(H, h)$ will also be a subgraph of the projection of such an Euler tour onto $T$.

More specifically, suppose $G$ is a tree with diameter at most $h$ and let $(T, \cdot)$ be a partial tree embedding of $G$. Let $G_2$ be the multigraph of $G$ where each edge is doubled. Let $t = (v_1, v_2, \ldots)$ be an Euler tour of $G_2$ and let $t' = (w_1, w_2, \ldots)$ be the vertices of $V(T)$ visited by this tour in the order in which they are visited. That is, $t'$ is gotten from $t$ be deleting from it all vertices not in $V(T)$ while leaving the ordering of the remaining vertices unchanged. Notice that vertices in $V(T)$ might occur multiple times in $t'$. We let $P_\ell$ be the path in $G$ between $w_\ell$ and $w_{\ell+1}$ and let $\mcP := \{P_\ell\}_\ell$. Next, consider $T(\mcP)$ which is the union of $T_{uv}$ for every $u,v$ where $u$ and $v$ form the endpoints of some path in $\mcP$.

First, notice that $T(H, h) \subseteq T(\mcP)$. This follows since every $u, v \in V(T)$ which are within $h$ hops (namely all $u,v \in V(T)$) are also visited by $t'$ and so if $T_{uv}$ is included in $T(H,h)$ then it will also be included in $T(\mcP)$. Next, notice that $w_T(\mcP) \leq 2  \alpha \cdot w_G(H)$ since our Euler tour when projected onto $G$ visited each edge at most twice. This proves \Cref{thm:projSoln} for the $h$-diameter tree case. 

\subsubsection{$h$-Hop Connectors}
The key observation of the above warm-up is that Euler tours allow us to mitigate the congestion induced in our charging arguments by providing a low-congestion collection of paths. We abstract such a collection of paths out in the form of what we call $h$-hop connectors.

For undirected and unweighted graph $G=(V,E)$ with $W \subseteq V$, we let $\mcP^{(h)}(W)$ be all simple paths between vertices in $W$ with at most $h$ hops. That is, each $P \in \mcP^{(h)}(W)$ has vertices in $W$ as its first and last vertices and satisfies $|P \cap W| = 2$ and $\hop(P) \leq h$. Given a collection of paths $\mcP$ in $G$ between vertices in $W$, we abuse notation and let $(W, \mcP)$ be the graph with vertex set $W$ and an edge $\{u,v\}$ iff there is a $P \in \mcP$ with endpoints $\{u,v\}$. We will refer to $(W, \mcP\at{h}(W))$ as the $h$-hop connectivity graph of $W$. We let $c_e(\mcP) := |P \in \mcP: e \in P|$ be the congestion of $e$ with respect to a collection of paths $\mcP$. With this notation in hand, we give our definition of $h$-hop connectors which we illustrate in \Cref{fig:hHopConnEGs}.

% For graph $G = (V,E)$ we say that $W \subseteq V$ is $h$-hop connected if the graph with vertex set $W$ with an edge between $u, v \in W$ iff $\hop_G(u, w) \leq h$ is connected. 

\begin{definition}[$h$-Hop Connector]
    Let $G = (V,E)$ be an undirected and unweighted graph, let $h \geq 1$ and let $W \subseteq V$. An $h$-hop connector $\mcP$ of $W$ with congestion $C$ and hop stretch $\beta$ is a collection of paths in $G$ between vertices of $W$ such that:
    \begin{enumerate}
        \item \textbf{Connecting:} all $u,v \subseteq W$ which are connected in $(W, \mcP^{(h)}(W))$ are connected in $(W, \mcP)$;
        \item \textbf{Edge Congestion:} For all $e \in E$ we have $c_e(\mcP) \leq C$;
        \item \textbf{Hop Stretch:} $\hop(P) \leq \beta \cdot h$ for all $P \in \mcP$.
    \end{enumerate}
\end{definition}

\begin{figure}
    \centering
    \begin{subfigure}[b]{0.22\textwidth}
        \centering
        \includegraphics[width=\textwidth]{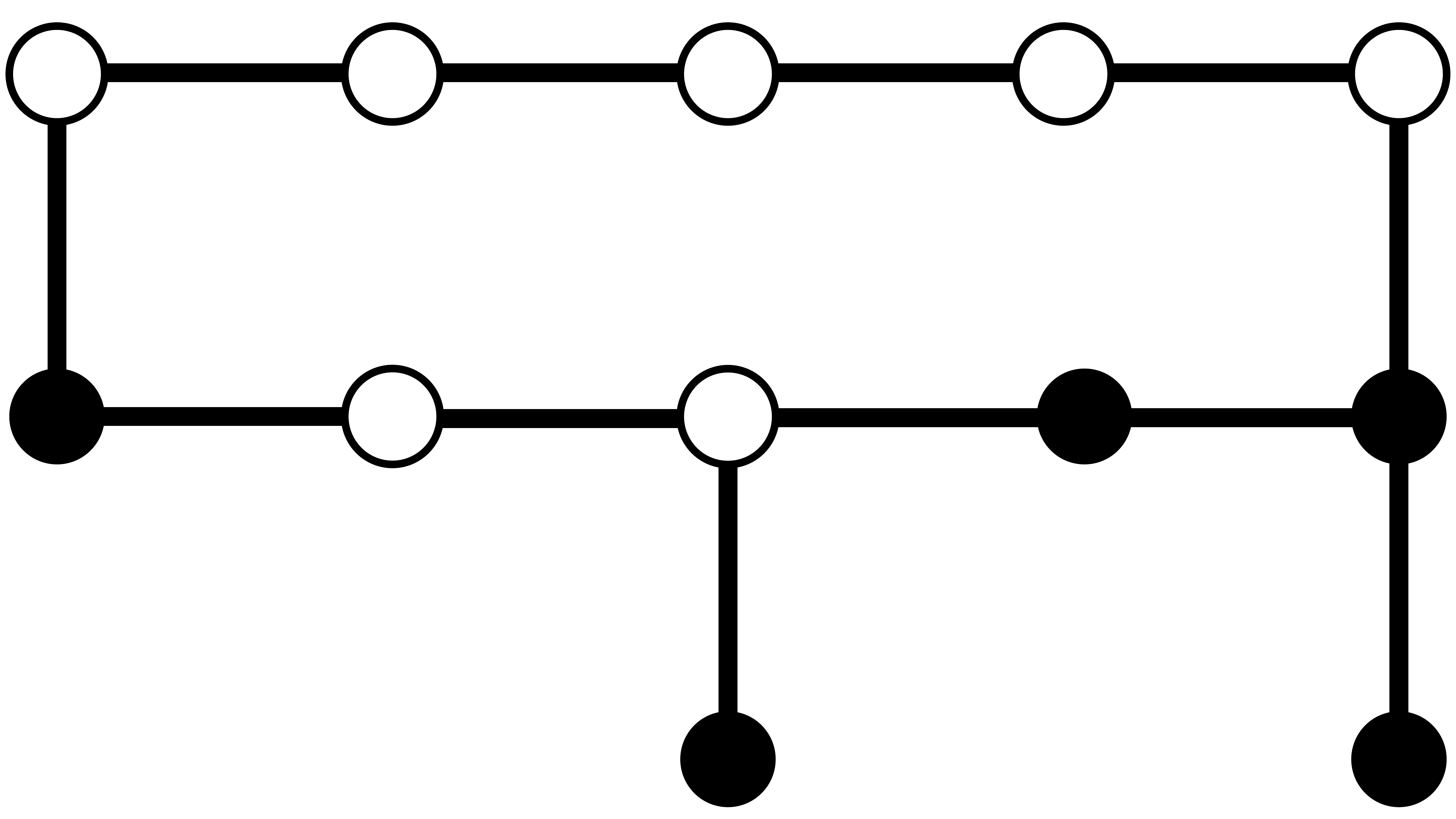}
        \caption{Graph $G$}
    \end{subfigure}
    \hfill
    \begin{subfigure}[b]{0.22\textwidth}
        \centering
        \includegraphics[width=\textwidth]{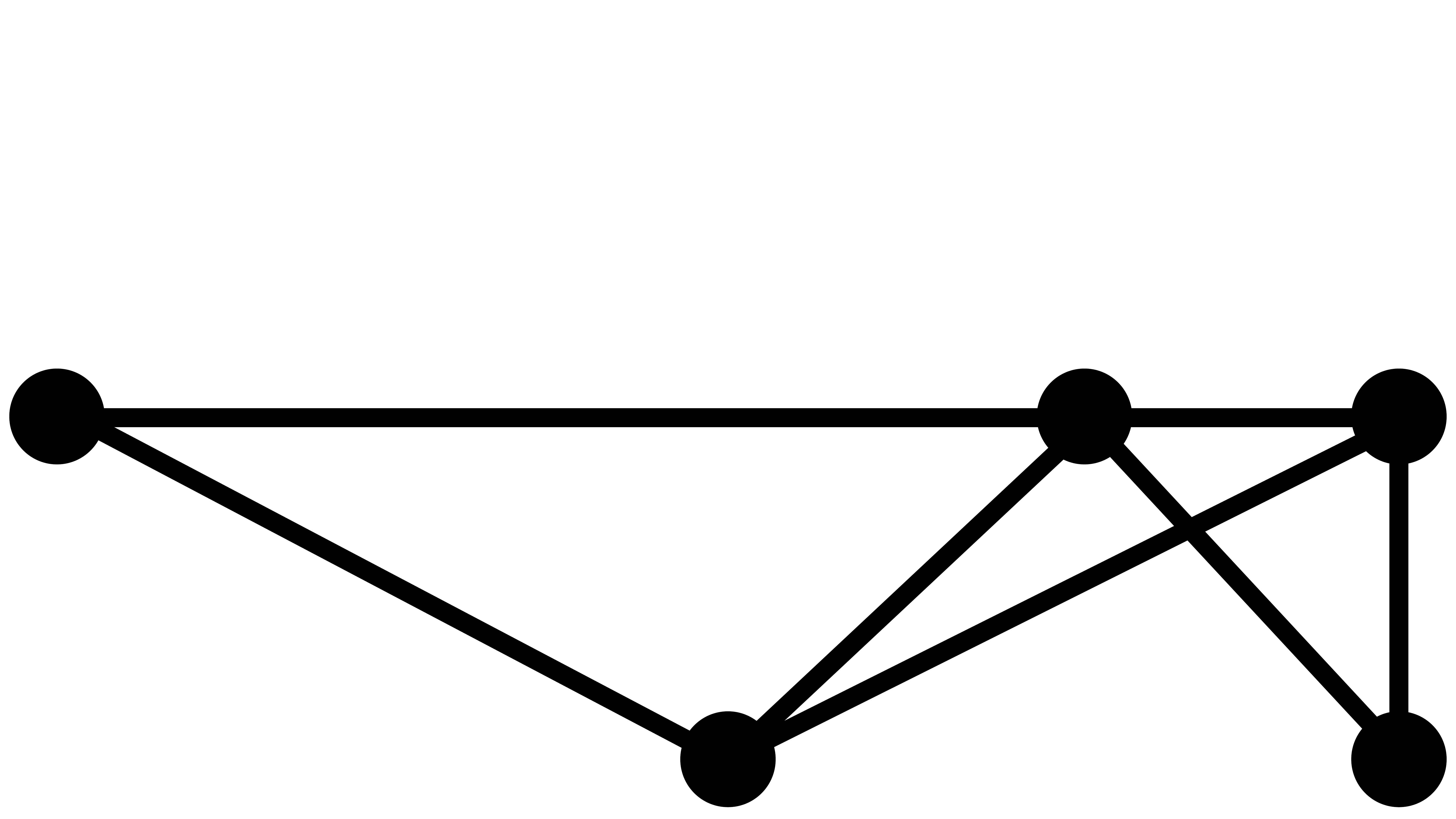}
        \caption{$(W, \mcP_h)$}
    \end{subfigure}
    \hfill
        \begin{subfigure}[b]{0.22\textwidth}
        \centering
        \includegraphics[width=\textwidth]{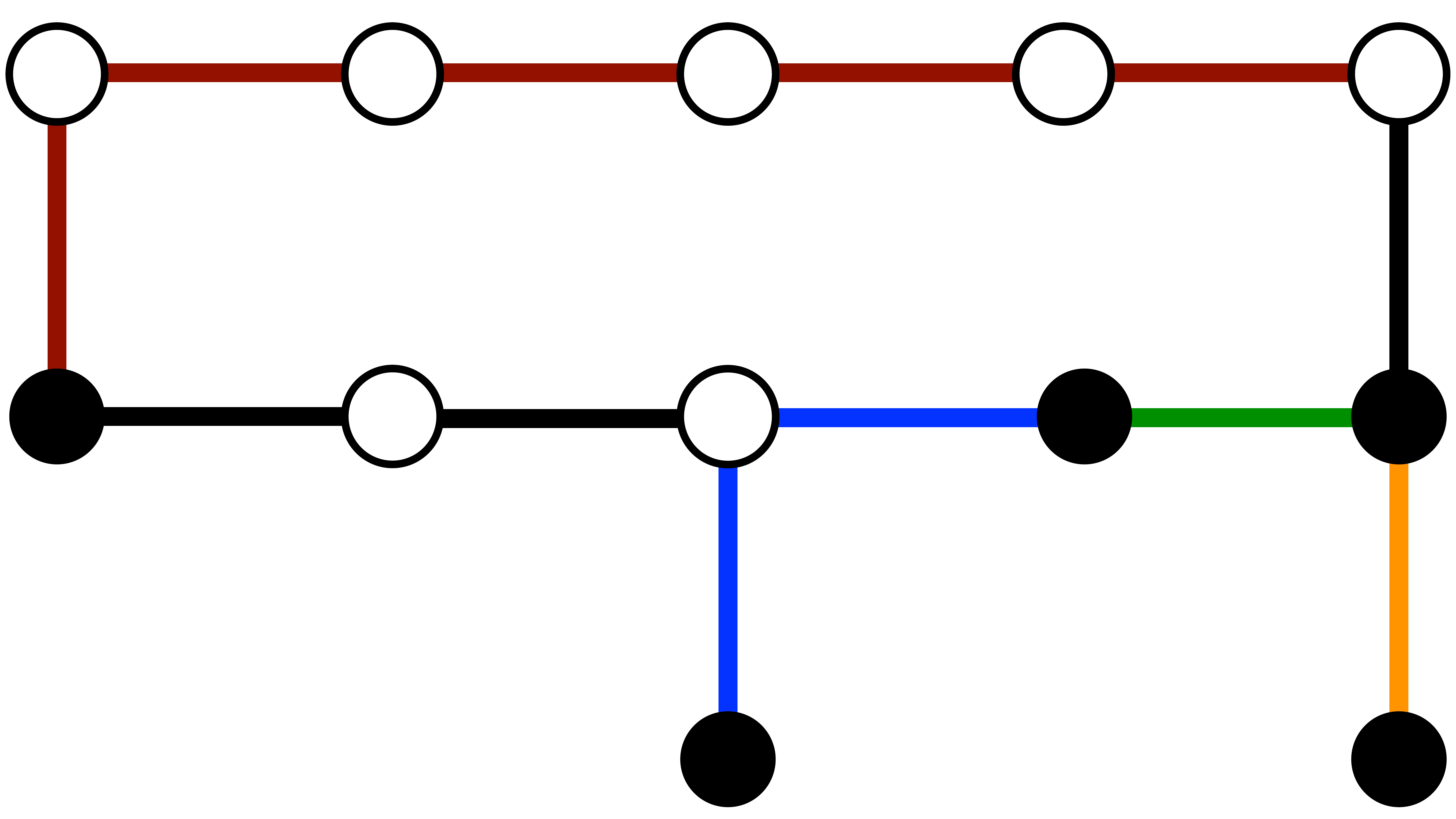}
        \caption{$\mcP$ in $G$}
    \end{subfigure}
    \hfill
    \begin{subfigure}[b]{0.22\textwidth}
        \centering
        \includegraphics[width=\textwidth]{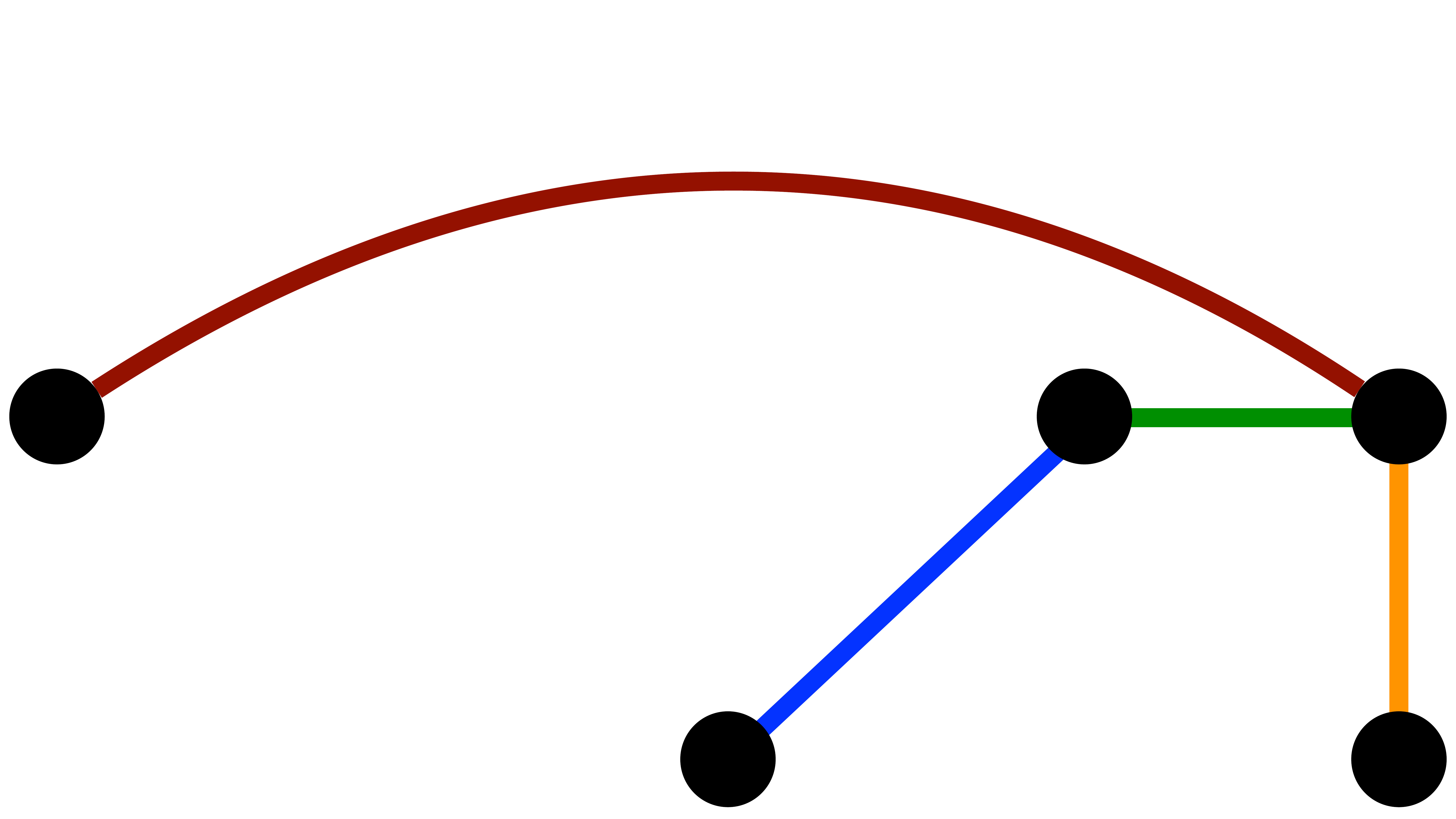}
        \caption{$(W, \mcP)$}
    \end{subfigure}
    \caption{An illustration of an $h$-hop connector with congestion $1$ and hop stretch $2$ on a graph $G$ for a vertex set $W \subseteq V(G)$ with $h = 3$. Vertices of $W$ given as solid black circles; all other vertices of $G$ given as white circles. Edges in $(W, \mcP)$ and paths in $\mcP$ colored according to their correspondence.}
    \label{fig:hHopConnEGs}
\end{figure}
    
It is easy to observe that the existence of good $h$-hop connectors are sufficient to show \Cref{thm:projSoln}.
\begin{lemma}\label{lem:hHopGivesProj}
    Fix $h \geq 1$, let $H \subseteq G$ be a subgraph of weighted graph $G = (V, E, w_G)$ and let $(T, \cdot)$ be a $(\beta h)$-hop partial tree embedding of $G$ with worst-case distance stretch $\alpha$. If $H$ has an $h$-hop connector on $V(T)$ with hop stretch $\beta$ and congestion $C$ then $w_T(T(H,h)) \leq C \alpha \cdot w_G(H)$.
\end{lemma}
\begin{proof}
    Let $\mcP$ be the stated $h$-hop connector, let $S := \{(u,v) : (u, \ldots, v) \in \mcP \}$ be the endpoints of its path and let $T(\mcP) := \bigcup_{(u,v) \in S} T_{uv}$ be the subgraph of $T$ corresponding to $\mcP$. By the connecting property of our $h$-hop connector any $u,v$ which are within $h$ hops in $H$ must also be connected in $T(\mcP)$ and so $T(H, h) \subseteq T(\mcP)$. Combining this with the edge congestion and hop stretch of our $h$-hop connector with the worst-case distance stretch of $T$ we have $w_T(T(H, h)) \leq w_T(T(\mcP)) \leq C \alpha \cdot w_G(H)$.
\end{proof}

Thus, we devote the remainder of this section to showing that every graph has an $h$-hop connector with hop stretch $8$ and congestion $4$.

%\subsubsection{Low Diameter Graph Case}

A simple proof similar to the above warm-up shows that trees with low diameter have good $h$-hop connectors.
\begin{lemma}\label{lem:hHopTree}
    Let $G = (V,E)$ be a tree with diameter at most $\beta h$ for $h \geq 1$. Then, $G$ has an $h$-hop connector with congestion at most $2$ and hop stretch at most $\beta$ for every $W \subseteq V$.
    
    %Then, $G$ has an $h$-hop connector with congestion at most $2$ and hop stretch at most $2\beta$ for every $W \subseteq V$. Moreover, if $G$ is a tree then it has an $h$-hop connector with congestion at most $2$ and hop stretch at most $\beta$ for every $W \subseteq V$.
\end{lemma}
\begin{proof}
    Suppose $G$ is a tree. Let $G_2$ be the multigraph of $G$ where each edge is doubled. Let $t = (v_1, v_2, \ldots)$ be an Euler tour of $G_2$ and let $t' = (w_1, w_2, \ldots)$ be the vertices of $W$ visited by this tour in the order in which they are visited. That is, $t'$ is gotten from $t$ be deleting from it all vertices not in $W$ while leaving the ordering of the remaining vertices unchanged. Notice that vertices in $W$ might occur multiple times in $t'$. We let $P_\ell$ be the path in $G$ between $w_\ell$ and $w_{\ell+1}$ and let $\mcP := \{P_\ell\}_\ell$. 
    
    Since every vertex in $W$ occurs at least once in $t'$ we have that all vertices in $W$ are connected in $(W, \mcP)$. Since $t$ used each edge of $G_2$ once, it follows that $c_e(\mcP) \leq 2$. Lastly, $\hop(P_\ell) \leq \beta h$ for all $P_\ell \in \mcP$ since each $P_\ell$ is a simple path in a tree with diameter at most $\beta h$.
%    Next, suppose $G$ is a general graph and let $T$ be its BFS tree rooted arbitrarily. Since $T$ has diameter at most $2 \beta h$, the desired result is immediate from the tree case.
\end{proof}

%By taking a BFS tree we can generalize this to all graphs of low diameter.
%\begin{corollary}\label{lem:hHopLowDiamGraph}
%    Let $G = (V,E)$ be a graph with diameter at most $\beta \cdot h$ for $h \geq 1$. Then, $G$ has an $h$-hop connector with congestion at most $2$ and hop stretch at most $2\beta \cdot h$ for every $W \subseteq V$.
%\end{corollary}
%\begin{proof}
%
%\end{proof}

%\subsubsection{General Graph Case}\label{sec:constHHop}

We proceed to show how to construct an $h$-hop connector with congestion $4$ and hop stretch $8$ on any graph. We first reduce the general graph case to the forest case: we show that, up to a factor of $2$ in the hop stretch, every graph $G$ has as a subgraph a forest $F$ where an $h$-hop connector for $F$ is an $h$-hop connector for $G$. We then reduce the forest case to the low diameter tree case by cutting each tree in $F$ at $O(h)$-spaced annuli from an arbitrary root so that the resulting trees have low diameter. We apply \Cref{lem:hHopTree} to the resulting low-diameter trees. More specifically, we perform these cuts and applications of \Cref{lem:hHopTree} twice with two different offsets to get back paths $\mcP_1$ and $\mcP_2$; we then take our $h$-hop connector to be $\mcP := \mcP_1 \cup \mcP_2$. We illustrate this strategy in \Cref{fig:hHopConnConst}.

\begin{figure}
    \centering
    \begin{subfigure}[b]{0.32\textwidth}
        \centering
        \includegraphics[width=\textwidth]{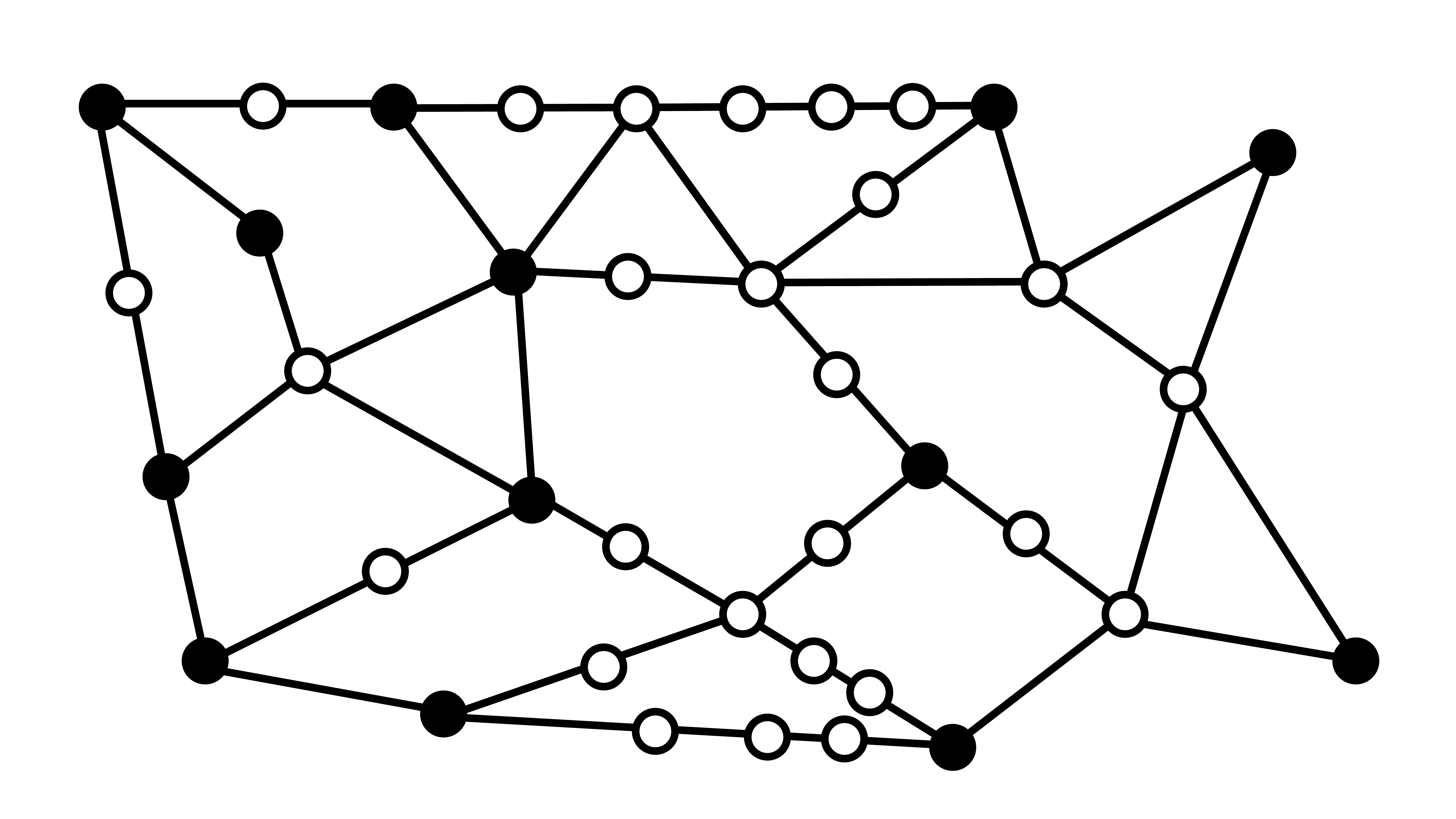}
        \caption{Graph $G$}
    \end{subfigure}
    \hfill
    \begin{subfigure}[b]{0.32\textwidth}
        \centering
        \includegraphics[width=\textwidth]{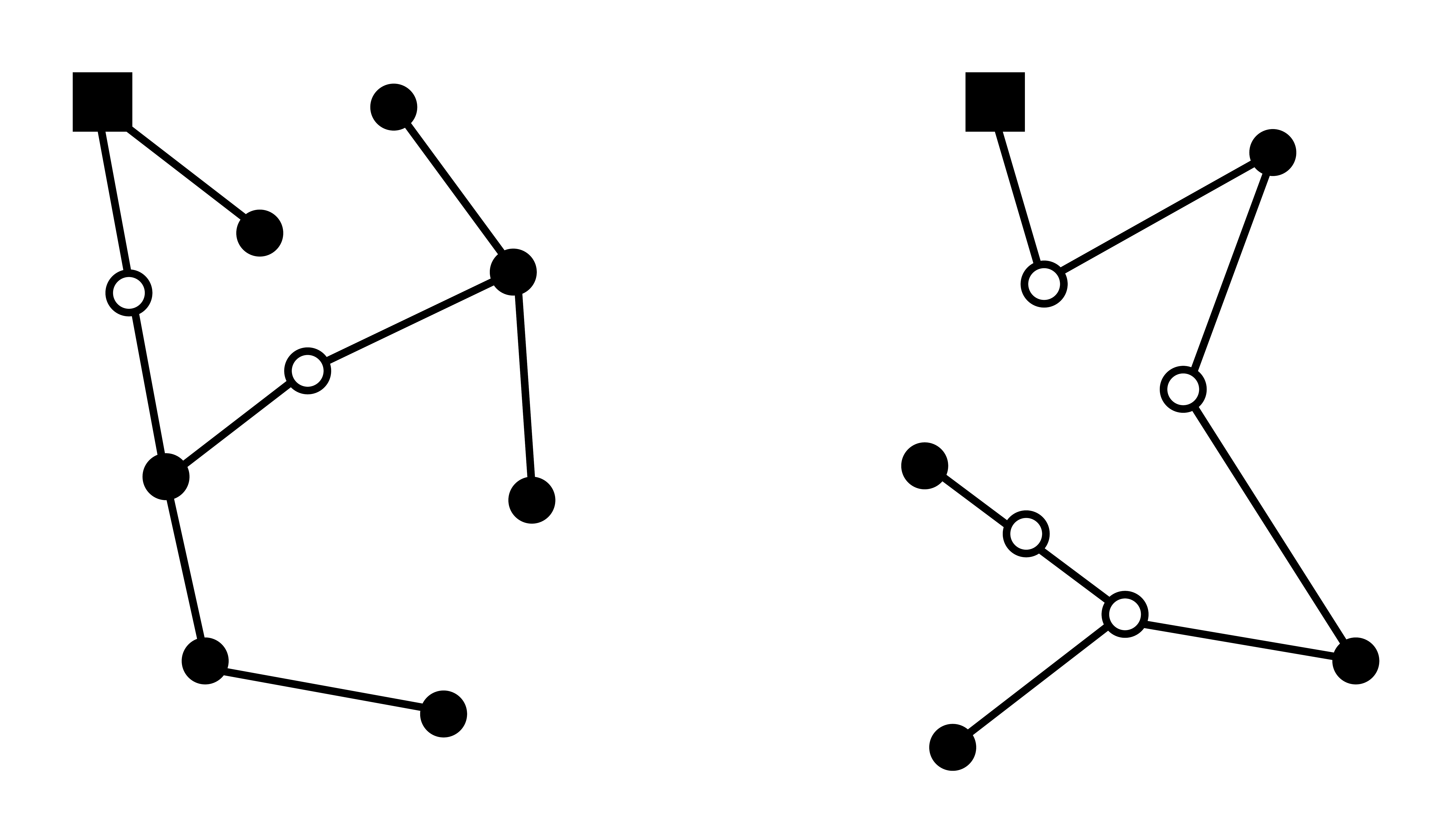}
        \caption{Forest $F$}
    \end{subfigure}
    \hfill
    \begin{subfigure}[b]{0.32\textwidth}
        \centering
        \includegraphics[width=\textwidth]{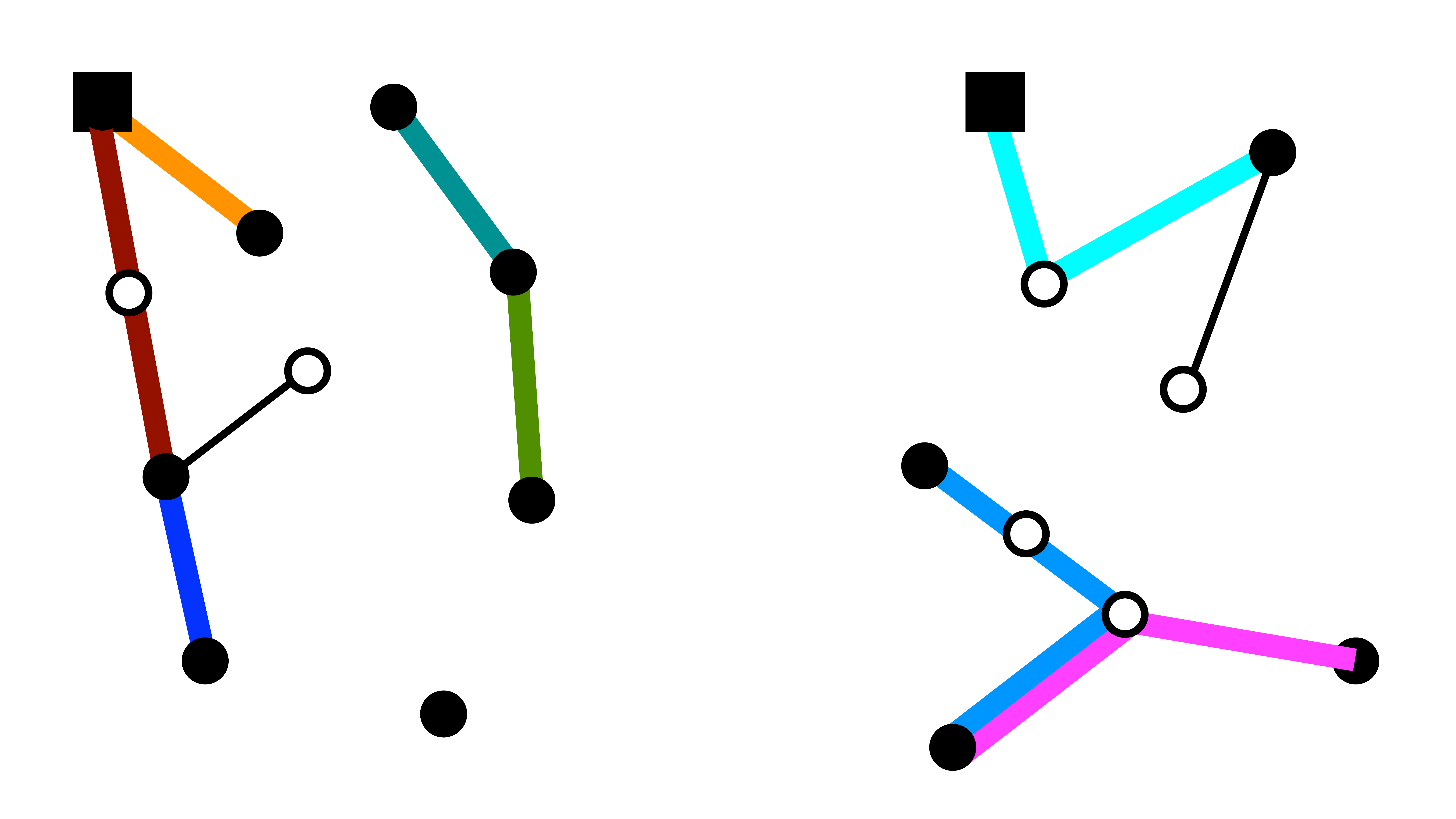}
        \caption{First offset and $\mcP_1$}
    \end{subfigure}
    \hfill
    \begin{subfigure}[b]{0.32\textwidth}
        \centering
        \includegraphics[width=\textwidth]{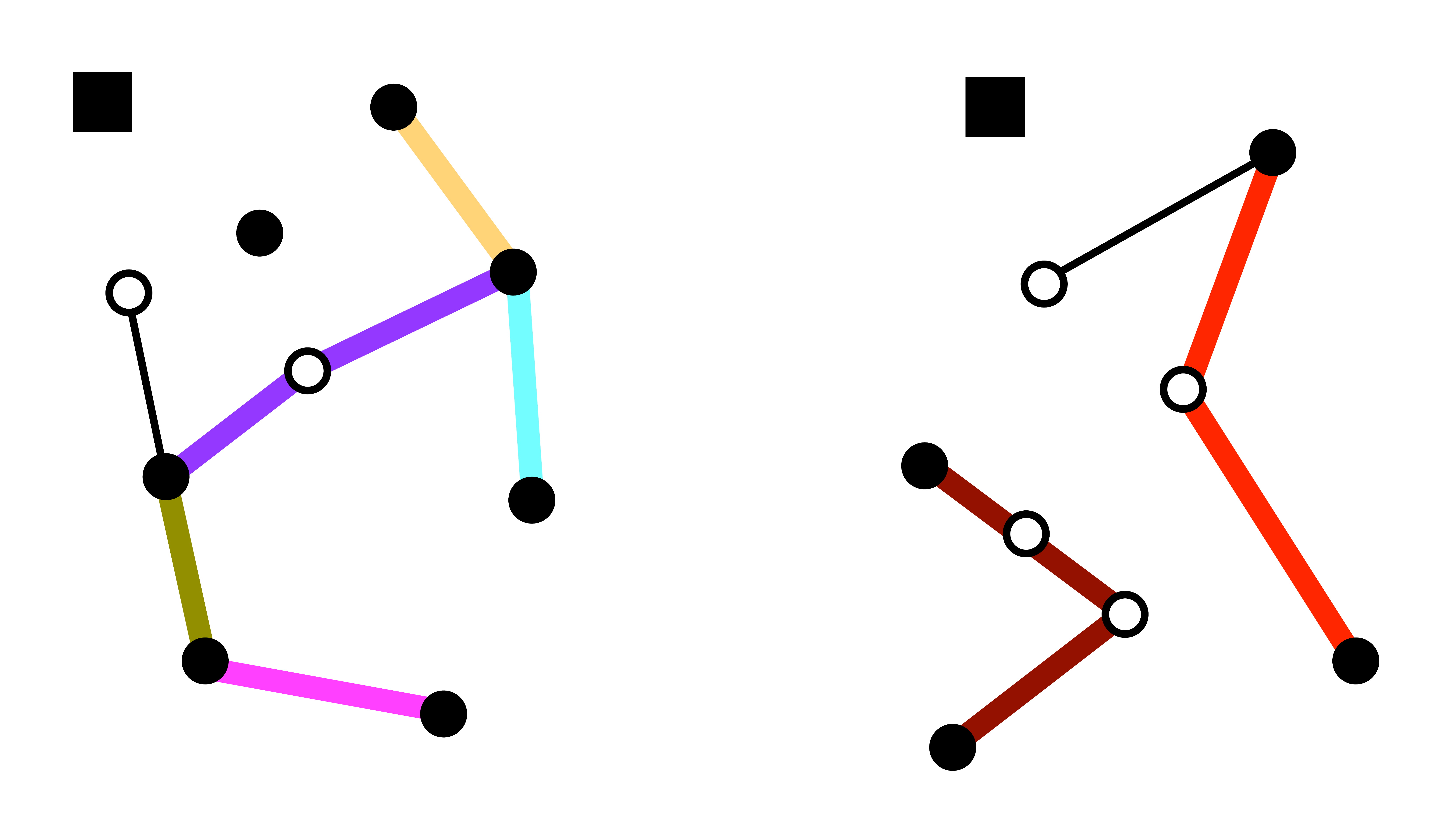}
        \caption{Second offset and $\mcP_2$}
    \end{subfigure}
    \begin{subfigure}[b]{0.32\textwidth}
      \centering
        \includegraphics[width=\textwidth]{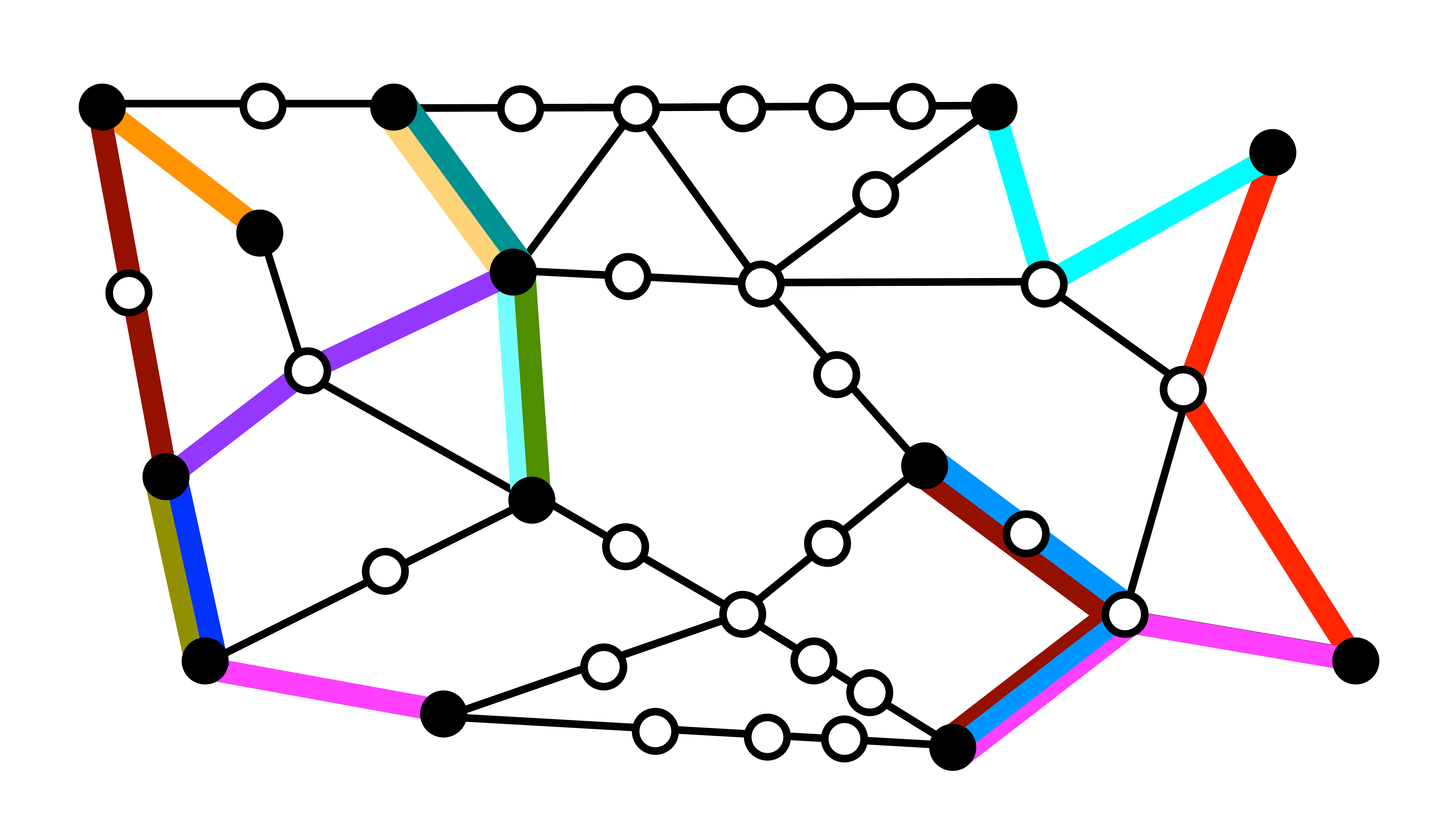}
        \caption{$\mcP = \mcP_1 \cup \mcP_2$ on $G$}
    \end{subfigure}
    \begin{subfigure}[b]{0.32\textwidth}
        \centering
        \includegraphics[width=\textwidth]{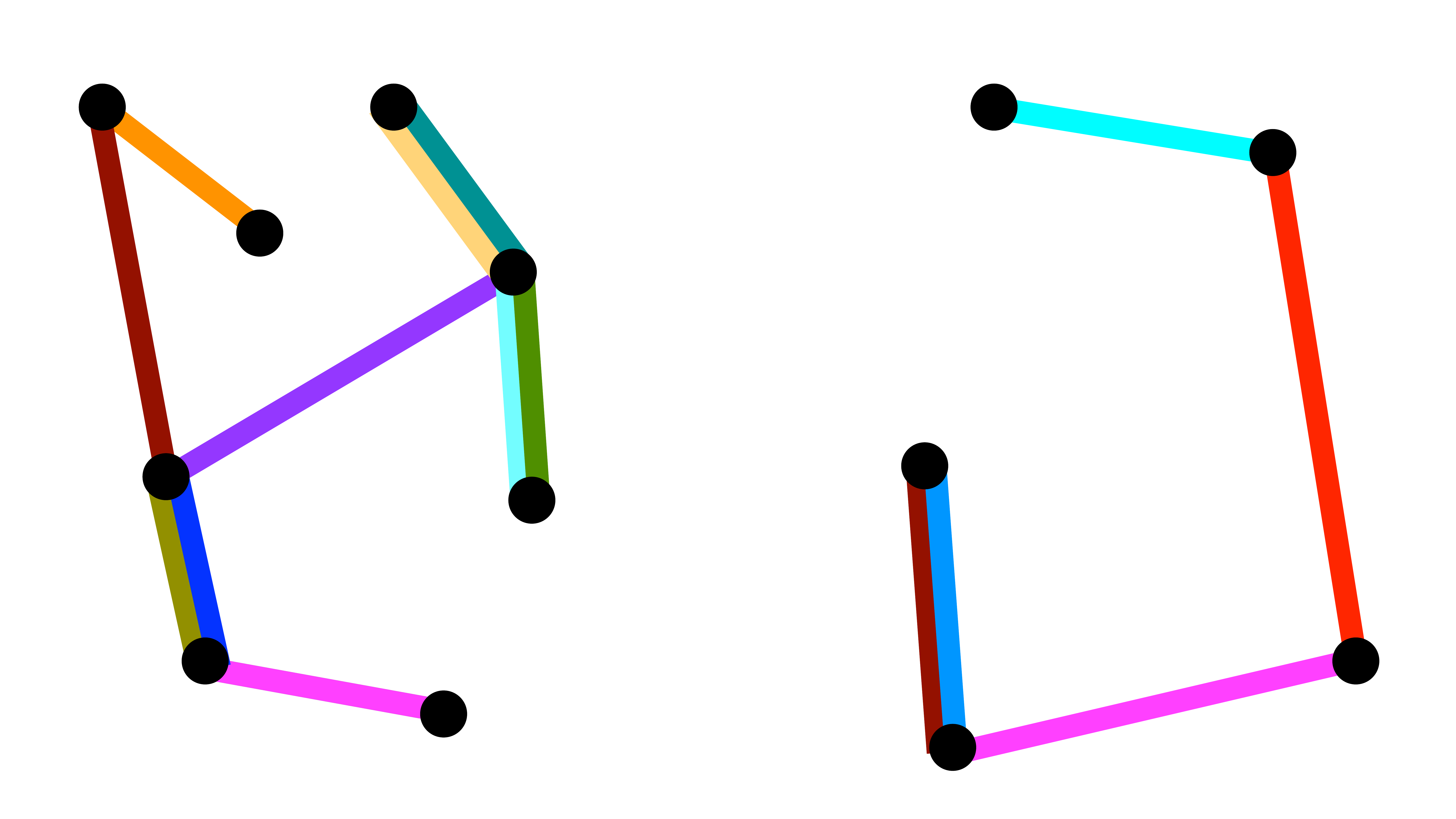}
        \caption{$(W, \mcP)$}
    \end{subfigure}
    \caption{An illustration of how to compute an $h$-hop connector on an arbitrary graph $G$ for $W \subseteq V(G)$ with $h = 3$. Vertices of $W$ given as solid black circles; roots of $F$ given as black squares; all other vertices of $G$ given as white circles. Paths in $\mcP_1$ and $\mcP_2$ colored to correspond to their edges in $(W, \mcP)$.}
    \label{fig:hHopConnConst}
\end{figure}

We begin with a simple technical lemma which shows that the graphs induced by the connected components of the $h$-hop connectivity graph are disjoint. For a collection of paths $\mcP$ in $G$ we let $G[\mcP] := G[\bigcup_{P \in \mcP} V(P)]$ be the graph induced by the union of all such paths.

\begin{lemma}\label{lem:disjointParts}
    Let $G =(V, E)$ be a graph, let $W \subseteq V$, and let $U$ and $U'$ be the vertices of two distinct connected components of $(W, \mcP^{(h)}(W))$. Then $G[\mcP^{(h)}(U)]$ and $G[\mcP^{(h)}(U')]$ are vertex-disjoint.
\end{lemma}
%\begin{lemma}\label{lem:disjointPaths}
%    Let $G =(V, E)$ be a graph, let $W \subseteq V$, and let $P, P' \in \mcP_h(W)$. If $V(P) \cap V(P') \neq \emptyset$ then the edges corresponding to $P$ and $P'$ in $(W, \mcP_h(W))$ are in the same connected components.
%\end{lemma}
\begin{proof}
    It suffices to show that for any $P, P' \in \mcP^{(h)}(W)$ if $V(P) \cap V(P') \neq \emptyset$ then the edges corresponding to $P$ and $P'$ in $(W, \mcP^{(h)}(W))$ are in the same connected components in $(W, \mcP^{(h)}(W))$. Let $u, v \in W$ be the endpoints of $P$ and let $u', v' \in W$ be the endpoints of $P'$. Suppose some $x \in V(P) \cap V(P')$. Let $P_{ux}$ and $P_{xv}$ be the subpaths of $P$ from $u$ and $v$ to $x$ respectively and define $P_{u'x}'$ and $P_{xv'}'$ symmetrically. Then, without loss of generality $P_{ux}$ and $P_{u'x}'$ both have at most $h/2$ edges meaning the concatenation of $P_{ux}$ and $P_{u'x}'$ at $x$ is in $\mcP\at{h}(W)$. It follows that $u$ and $u'$ are in the same connected component in $(W, \mcP^{(h)}(W))$ and therefore, $v$ and $v'$ are also in this component.
\end{proof}

Applying the above lemma, we show that, up to a factor of $2$ in the hop stretch, we may assume that our graph is a forest. We let $\mcP_G^{(h)}(W)$ be all paths with at most $h$ hops between vertices in $W$ \emph{in graph $G$}.
\begin{lemma}\label{lem:WLOGTree}
    Let $G =(V, E)$ be a graph, let $W \subseteq V$. Then there exists a subgraph $F \subseteq G$ which is a forest where $u,v \in W$ are connected in $(W, \mcP^{(h)}_G(W))$ iff $u, v$ are connected in $(W, \mcP\at{2h}_F(W))$.
\end{lemma}
\begin{proof}
    We will iteratively construct $F$. Specifically, for each connected component of $(W, 
    \mcP^{(h)}_G(W))$ with vertex set $U$ we will maintain a collection of paths $\mcP_U$ where these paths are all contained in $G[\mcP^{(h)}(U)]$ and $F$ is the graph induced by the union of all these paths. It follows that by \Cref{lem:disjointParts} if $G[\mcP_U]$ is a tree then  the connected components of our final solution are indeed a forest. We will maintain the following invariants for our $\mcP_U$s where $\hop_G(v, U) := \min_{u \in U} \hop_G(v, u)$:
    \begin{enumerate}
        \item $U' := U \cap V(G[\mcP_U])$ is connected in $(U', \mcP_U)$;
        \item $G[\mcP_U]$ is a tree;
        \item $\hop(P) \leq 2h$ for every $P \in \mcP_U$;
        \item $\hop_{G[\mcP_U]}(v, U) \leq h$ for every $v \in V(G[\mcP_U])$.
        %         \item  $P$ and $P'$ are vertex disjoint if $P \in \mcP_U$ and $P' \in \mcP_{U'}$ for two distinct vertex sets $U$ and $U'$ of connected components in $(U, \mcP_U)$.
    \end{enumerate}

    We initialize $\mcP_U$ to contain a path consisting of exactly one (arbitrary) vertex in $U$. Notice that our construction trivially satisfies these invariants initially.
    
    Next, we repeat the following until $U' = U$. Let $u$ be a vertex in $U \setminus U'$ where $u$ has a path $P$ of at most $h$ hops to a vertex in $U'$; such a $u$ and $P$ must exist by the definition of $U$. Let $x$ be the first vertex in $P \cap G[\mcP_U]$ where we imagine that $P$ starts at $u$ and let $P_{ux}$ be the subpath of $P$ from $u$ to $x$. By invariant $4$ we also know there is some path in $G[\mcP_U]$ from $x$ to a $u' \in U'$ with at most $h$ hops; call this path $P_{xu'}$ and let $P'$ be the concatenation of $P_{ux}$ and $P_{xu'}$; we add $P'$ to $\mcP_U$. Notice that this adds $u$ to $U'$ and so this process will eventually terminate at which point $U' = U$.
    
    Let us argue that our invariants hold. Our first invariant holds since before adding $u$ to $U'$, $U'$ was connected and after adding $u$ to $U'$, $u$ is connected to $u'$ by $P'$. Our second invariant holds since $x$ was the first vertex in $G[\mcP_U]$ incident to $P$. Our third invariant holds since $P_{ux}$ and $P_{xu'}$ were each of at most $h$ hops. Our fourth invariant holds since the only new vertices we add to $G[\mcP_U]$ are the vertices of $P_{ux}$, all of which are within $h$ hops of $u$.
    
    Lastly, notice that once $U' = U$ for every $U$, our claim follows from invariants 1,2 and 3.
\end{proof}

By turning our graph into a forest with \Cref{lem:WLOGTree} and then cutting the constituent trees at $O(h)$-spaced level sets with two different initial offsets, we can conclude that every graph has $h$-hop connectors with constant congestion and hop stretch.
\begin{restatable}{lemma}{hHopConst}\label{lem:hHopConst}
    Let $G = (V,E)$ be a graph. Then $G$ has an $h$-hop connector with congestion $4$ and hop stretch $8$ for every $W \subseteq V$.
\end{restatable}

\begin{proof}
    By \Cref{lem:WLOGTree} we know that there is a forest $F$ such that $u,v$ are connected in $(W, \mcP^{(h)}_G(W)$ iff $u, v$ are connected in $(W, \mcP\at{2h}_F(W))$. Let $T$ be a tree in this forest and notice that to get an $h$-hop connector on $G$ with hop stretch $8$ and congestion $4$, it suffices to find a $2h$-hop connector on $T$ with hop stretch $4$ and congestion $4$.
    
    We do so as follows. Root $T$ arbitrarily at root $r$ and let $T_1, T_2, \ldots$ be the subtrees resulting from cutting $T$ once every $4h$ levels and let $T_1', T_2', \ldots$ be the subtrees resulting from cutting $T$ every $4h$ levels with an initial offset of $2h$. That is, $T_i = T[V(T_i)]$ and $v \in V(T_i)$ iff $4h(i-1) \leq d_T(v,r) < 4h \cdot i$ and $T_i' = T[V(T_i')]$ where $v \in V(T_i')$ iff $\max(4h(i-1) - 2h, 0) \leq d_T(v,r) < 4h \cdot i - 2h$. Notice that each $T_i$ and $T_i'$ has diameter at most $4 (2h)$. Thus, by \Cref{lem:hHopTree} we know that each $T_i$ and $T_i'$ have $2h$-hop connectors $\mcP_i$ and $\mcP_i'$ with congestion at most $2$ and hop stretch at most $4$. Thus, we let $\mcP_1 := \{ \mcP_i\}_i$ and $\mcP_2 := \{\mcP_i'\}_i$ and we let our $h$-hop connector for $T$ be $\mcP := \mcP_1 \cup \mcP_2$.
    
    Let us argue that $\mcP$ is a $2h$-hop connector on $T$ with hop stretch $8$ and congestion $4$. Our congestion bound is immediate from \Cref{lem:hHopTree} and the fact that each edge occurs in at most $2$ trees among all $T_i$ and $T_i'$. To see why $\mcP$ is connecting notice that if $u,v$ are within $2h$ hops of one another in $T$ by some path $P$ then this path must be fully contained in some $T_i$ or $T_i'$; it follows that $u$ and $v$ will be connected in some $\mcP_i$ or $\mcP_i'$ and so connected in $\mcP$. Lastly, our hop bound is immediate by \Cref{lem:hHopTree} since each $T_i$ and $T_i'$ has diameter at most $4(2h)$.
\end{proof}

Combining \Cref{lem:hHopConst} with \Cref{lem:hHopGivesProj} immediately gives \Cref{thm:projSoln}.

Before proceeding to our applications, we remark on a subtle issue regarding independence and expected distance stretch versus worst case distance stretch. \Cref{thm:projSoln} bounded the cost of projecting a subgraphs of $G$ onto a partial tree embedding of $G$ based on the tree embedding's worst-case distance stretch; one might naturally wonder if similar results are possible in terms of the expected distance stretch of a distribution over partial tree embeddings. Here, dependence issues and the partialness of our embeddings work against us. Specifically, one would have to argue that $T(H,h)$---and, in particular, the relevant $h$-hop connector for $T(H, h)$---has low cost in expectation where $(T, \cdot)$ is drawn from a distribution. However, while it is true that for a fixed $H$ and $T$ the relevant $h$-hop connector for $H$ and $T$ has low cost in expectation over the entire distribution of tree embeddings, it need not be the case that this $h$-hop connector has low cost when we condition on the fact that $T$ is the tree we drew from our distribution. In short, \Cref{lem:hHopGivesProj} seems to fail to hold for the expectation case.

\section{Applications of $h$-Hop Partial Tree Embeddings}\label{sec:firstAppl}

In this section we apply our embeddings of $d^{(h)}$ to give approximation algorithms for hop-constrained versions of several well-studied network design problems; namely, oblivious hop-constrained Steiner forest, hop-constrained group Steiner tree, hop-constrained $k$-Steiner tree and hop-constrained oblivious network design. If unspecified, $\OPT$ will stand for the optimal value of the relevant hop-constrained problem throughout this section. We improve our results for hop-constrained group Steiner tree in a later section (\Cref{sec:groupSTBetter}).
% As above, we assume throughout this section that all of our graphs are complete; this assumption is without loss of generality since in each of the problems we study we may always add high cost (i.e.\ cost $L$) edges to our graph to make it complete without changing the optimal solution or which edges are chosen by our approximation algorithms; we may assume that no optimal solution buys any cost $L$ edges. Throughout this section we let $\OPT$ be the cost of the optimal solution of the discussed problem.\gznote{I would remove all of this. It should be clear to the readers we solve the general case even thought formally our graph is weighted and complete (the argument is standard and already talked about in the notation). I think this makes the paper longer and puts unnecessary attention to it.}\enote{I'm fine removing this. However, I think when we state the fact that we assume complete graphs we should emphasize that we can only do so in the applications we consider because for \emph{these} problems it is WLOG. As in, there are other problems where the graph is not complete WLOG.}

\subsection{Oblivious Hop-Constrained Steiner Forest}\label{sec:steinFor}

%\textbf{Problem.} Given a weighted graph $G = (V(G), E(G), w_G)$, a hop constraint $h \ge 1$, and a set of $k$ ``demand pairs'' $\{(s_i, t_i)\}_{i \in [k]}$ we want to find the subgraph $F \subseteq E(G)$ of minimum weight $w_G(F)$ such that all $k$ pairs $s_i, t_i$ are connected with path of at most $h$ hops in $F$.

%We begin with the simplest application of our tree embeddings, hop-constrained Steiner forest.

In this section we give our approximation algorithms for oblivious hop-constrained Steiner forest. While we give our results for oblivious hop-constrained Steiner forest, it is easy to see that an approximation algorithm for the oblivious version gives an approximation algorithm with the same approximation ratios for the online and offline versions of the problem; to our knowledge nothing was known for any of these variants prior to our work.

\textbf{Problem:} In Steiner forest we are given a weighted graph $G = (V, E, w)$.
\begin{itemize}
    \item \textbf{Offline:} In offline Steiner forest we are also given a collection of pairs of nodes $\{(s_i, t_i)\}_i$. Our goal is to find a subgraph $H \subseteq G$ so that every $s_i$ is connected to every $t_i$ in $H$.
    \item \textbf{Online:} In online Steiner forest in each time step $t=1,2,\ldots$ a new pair of vertices $(u_t, v_t)$ is revealed and we must maintain a solution $H_t$ for each $t$ where $H_{t-1} \subseteq H_t$ which connects pairs in $\{(u_1, v_1), \ldots, (u_t, v_t)\}$. %The competitive ratio of $\{H_t\}_t$ is the $\max_t w(H_t)/ \OPT_t$ where $\OPT_t$ is the minimum cost subgraph of $G$ such that $\hop_{F_t}(s_i, t_i) \leq h$ for every $i \leq t$.
    \item \textbf{Oblivious:} In oblivious Steiner forest we must specify a path $P_{uv}$ for each pair of vertices $(u,v) \in V \times V$ before seeing any demands. The demands $\{(s_i, t_i)\}_i$ are then revealed, inducing our solution $H := \bigcup_{i}P_{s_i t_i}$.
\end{itemize}
In all three problems the cost of our solution $H$ is $w(H) := \sum_{e \in E(H)}w(e)$. In the oblivious and offline versions, our approximation ratio is $w(H) / \OPT$ where $\OPT$ is the cost of the optimal offline solution for the given demand pairs. The competitive ratio of our solution in the online case is $\max_t w(H_t)/ \OPT_t$ where $\OPT_t$ is the minimum cost subgraph of $G$ connecting pairs in $\{(u_1, v_1), \ldots, (u_t, v_t)\}$. 

In the hop-constrained versions of each of these problems we are additionally given a hop constraint $h \geq 1$ and if $(s_i, t_i)$ is a demand pair then our solution $H$ must satisfy $\hop_H(s_i, t_i) \leq h$ for all $i$. The optimal solution against which we measure our approximation ratio is similarly hop-constrained.

Notice that, unlike in the Steiner forest problem where we may assume without loss of generality that each connected component of $H$ is a tree, in hop-constrained Steiner forest each connected component of $H$ might not be a tree.

\textbf{Related Work:} We give some brief highlights from work in Steiner forest and hop-constrained Steiner forest: while NP-hard \cite{agrawal1995trees} gave the first constant approximation for offline Steiner forest; \cite{berman1997line} gave an (optimal) $O(\log k)$ approximation for online Steiner forest and \cite{gupta2006oblivious} gave the first non-trivial approximation algorithm for oblivious Steiner forest, an $O(\log ^ 2 n)$ approximation. There has also been quite a bit of work on approximation algorithms for $h$-spanners which can be seen as a special case of offline hop-constrained Steiner forest; see, for example, \cite{dinitz2016approximating} and references therein. Notably for our purposes, \cite{elkin2000strong} and \cite{dinitzmin} show that unless $\text{NP} \not \subseteq \text{BPTIME}(2^{\poly \log n})$ hop-constrained Steiner forest admits no $O(2^{\log ^ {1-\epsilon}n})$ approximation; this immediately rules out the possibility of a poly-log (unicriteria) approximation for hop-constrained Steiner forest.

\textbf{Algorithm:} Roughly, our algorithm follows the usual tree-embedding template: we first apply our $h$-hop partial tree embeddings to reduce oblivious hop-constrained Steiner forest to oblivious Steiner forest on a tree; we then observe that oblivious Steiner forest is trivially solvable on trees and project our solution back to $G$. The only minor caveats are: (1) since our tree embeddings will only embed a constant fraction of nodes, we must repeat this process $O(\log n)$ times and (2) for each tree embedding we must use \Cref{thm:projSoln} to argue that there is a cheap, feasible solution for the relevant Steiner forest problem on each tree.

Formally, our algorithm to compute our solution $H$ is as follows. We begin by applying \Cref{thm:mainEmbed} to sample $8h$-hop partial tree embeddings $T_1, T_2, \ldots, T_k$ where $k := O(\log n)$ for a sufficiently large hidden constant, $\eps = .1$ and an arbitrary root. Given $u,v \in V$, assign the pair $(u,v)$ to an arbitrary $T_j$ such that $u,v \in V(T_j)$ (we will argue that such a $T_j$ exists with high probability). Next, we let our path for $u, v$ be $P_{uv} := (T_j)^G_{uv}$ the projection of the tree path between $u$ and $v$ onto $G$.

We now give the analysis of our algorithm.
\begin{restatable}{theorem}{steinFor}\label{thm:SFAlg}
    There is a poly-time algorithm which given an instance of $h$-hop-constrained oblivious Steiner forest returns a collection of paths such that the induced solution $H$ for any demand set satisfies $w(H) \leq O(\OPT \cdot \log ^ 3 n )$ and $\hop_H(s_i, t_i) \leq O(h \cdot \log ^ 3 n)$ with high probability.
\end{restatable}
\begin{proof}
    We use the above algorithm. We begin by arguing that $H$ connects every $s_i$ to $t_i$ for every $i$ with high probability with a path of at most $O(\log ^ 3 n \cdot h)$ edges. Fix a vertex $v$. A standard Chernoff-and-union-bound-type argument shows that $v$ is in at least $.8k$ of the $T_j$ with high probability. Specifically, let $X_j$ be the random variable which indicates if $v$ is in $V(T_j)$, let $X := \sum_j X_j$ and apply a Chernoff bound to $X$. 
    
%    By \Cref{thm:mainEmbed} we have that $\E[X_j] = .9$ and by linearity of expectation we have $\E[X] = 900 \log n$. Since each $X_j$ is independent it follows by a Chernoff bound that
%            \begin{align*}
%            \Pr(X \leq 800 \log n) &\leq \Pr(X \leq (1-.1)\E[X]) \\
%            &\leq \exp(-\E[X] (.1)^2/ 3)\\
%            &\leq \exp (-3 \log n)\\
%             &= \frac{1}{n^3}
%            \end{align*}
    Taking a union bound over all $v$ we have that with high probability every $v$ is in at least $.8k$ of the $T_j$. Since we have $k$ total $T_j$, by the pigeonhole principle it follows that any pair of vertices $(s_i, t_i)$ simultaneously occur in at least $.6k$ of the $T_j$, meaning that for each such pair there is a $T_j$ where we buy $(T_j)_{s_it_i}^G$ and so $s_i$ will be connected to $t_i$ in our solution. Since $\hop((T_j)_{s_it_i}^G) \leq O(h \cdot \log ^ 3 n)$ by \Cref{thm:mainEmbed}, it follows that $\hop_H(s_i, t_i) \leq O(h \cdot \log ^ 3 n)$.
    
    We next argue that our solution satisfies the stated cost bound. Let $H_{T_j}$ be the minimal subgraph of $T_j$ connecting all pairs assigned to $T_j$ and let $H_j := \bigcup_{e \in H_{T_j}} (T_j)^G_e$ be the projection of $H_{T_j}$ onto $G$. Notice that it suffices to argue that $w_{T_j}(H_{T_j}) \leq O(\OPT \cdot \log ^ 2 n)$ for every $j$ since if this held we would have by \Cref{thm:mainEmbed} that the cost of our solution is $w(H) \leq \sum_j w(H_j) \leq \sum_j \sum_{e \in H_{T_j}} w((T_j)_e^G) \leq \sum_j \sum_{e \in H_{T_j}} w_{T_j}(e) = \sum_j w_{T_j}(H_{T_j}) \leq O(\OPT \cdot \log ^ 3 n)$.  However, applying \Cref{thm:projSoln} to the optimal solution $H^*$ on $G$ shows that $T(H^*, h)$ is a feasible solution for the Steiner forest problem on $T_j$ which connects all pairs assigned to $T_j$ with cost at most $O(\log ^ 2 n \cdot \OPT)$. Since $H_{T_j}$ is the optimal solution for such a Steiner forest problem, it follows that $w_{T_j}(H_{T_j}) \leq O(\OPT \cdot \log ^ 2 n)$ as required.
\end{proof}

\subsubsection{Bicriteria Min-Cost Spanner Approximations}
We end this section by remarking that our hop-constrained Steiner forest algorithm gives new bicriteria approximation algorithms for spanner problems. Notably, the aforementioned $\Omega(2^{\log ^ {1 - \epsilon}n})$ hardness of approximation reductions break down for bicriteria approximation algorithms. For this reason, \cite{chlamtavc2016lowest} state the following regarding bicriteria approximation algorithms for spanner problems:

\begin{quoting}
    Obtaining good bicriteria approximations, or proving that they cannot exist, is an extremely interesting area for future research...
\end{quoting}

As a corollary to our hop-constrained Steiner forest problem result, we give a new such bicriteria approximation algorithms for spanners. Specifically, in the minimum cost client-server $h$-spanner problem we are given a client graph $G_c = (V, E_c)$ and a weighted server graph $G_s = (V_s, E_s, w)$ and integer $h \geq 1$. We must find a subgraph $H$ of $G_s$ which minimizes $w(H)$ subject to the constraint that for each $\{u, v\} \in E_c$ we have $\hop_H(u, v) \leq h$. \cite{elkin1999client} and \cite{elkin2005approximating} studied the unit cost version of this problem for $h= 2$ and $h > 3$, giving bicriteria algorithms in the latter, but, to our knowledge no (poly-log, poly-log) bicriteria approximation algorithms are known for either the unit-cost version of this problem or the min-cost spanner problem (i.e. this problem when $E_s = E_c$).
 
By creating an offline hop-constrained Steiner forest problem which has a demand pair for each client edge, it is easy to see that \Cref{thm:SFAlg} gives such a bicriteria approximation algorithm for min-cost client-server $h$-spanner.
\begin{corollary}
    There is a poly-time algorithm for min-cost client-server $h$-spanner which returns an $H \subseteq G_s$ where $w(H) \leq O(\OPT \cdot \log ^ 3 n)$ and for each $\{u, v\} \in E_c$ we have $\hop_H(u, v) \leq O(h\cdot \log ^ 3 n)$.
\end{corollary}

\subsection{Hop-Constrained Group Steiner Tree}\label{sec:groupST}

As both set cover and Steiner tree are special cases of it, the group Steiner tree problem is one of the most general covering problems. In this section, we give $(O(\poly \log n), O(\poly \log n))$ bicriteria approximation algorithms for the hop-constrained variant of group Steiner tree. We later give an improved approximation based on ``$h$-hop repetition tree embeddings'' which build on our $h$-hop partial tree embeddings. However, we include this result to highlight the fact that $h$-hop partial tree embeddings alone are sufficient for solving many hop-constrained problems.

\textbf{Problem:} In the group Steiner Tree problem we are given a weighted graph $G = (V, E, w)$ as well as pairwise disjoint groups $g_1, g_2, \ldots, g_k \subseteq V$ and root $r \in V$. We let $N:= \max_i |g_i|$. Our goal is to find a tree $T$ rooted at $r$ which is a subgraph of $G$ and satisfies $T \cap g_i \neq \emptyset$ for every $i$. We wish to minimize our cost, $w(T) := \sum_{e \in E(T)} w(e)$. In hop-constrained group Steiner tree we are additionally given a hop bound $h \geq 1$ and we must ensure that $\hop_T(g_i, r) \leq h$ where $\hop_T(g_i, r) := \min_{v \in g_i \cap V(T)} \hop_T(v, r)$.\footnote{The assumption that the tree is rooted in group Steiner tree is without loss of generality as we may always brute-force search over a root. Similarly, the assumption that all groups are pairwise disjoint is without loss of generality since if $v$ is in groups $\{g_1, g_2, \ldots \}$ then we can remove $v$ from all groups and add vertices $v_1, v_2, \ldots$ to $G$ which are connected only to $v$ so that $v_i \in g_i$ and $w((v, v_i)) = 0$ for all $i$. For the unrooted hop-constrained group Steiner tree problem we might define the problem as unrooted group Steiner tree but with the additional constraint that $T$ has diameter at most $h$; all of our results will hold for this unrooted version of the problem though its worth noting that in this case the optimal solution is no longer a tree without loss of generality. The aforementioned transformation also allows us to assume that groups are pairwise disjoint in hop-constrained group Steiner tree at a possible loss of an additive $1$ in our hop stretch.}

Unlike hop-constrained Steiner forest, the optimal solution for hop-constrained Steiner tree is, in fact, a tree. In particular, if $H$ is a feasible solution, then the shortest path tree on $H$ rooted at $r$ is also a feasible solution of cost at most the cost of $H$.

\textbf{Related Work:} \cite{garg2000polylogarithmic} gave the first randomized poly-log approximation for offline group Steiner tree using linear program rounding. \cite{charikar1998rounding} derandomized this result and \cite{chekuri2006greedy} showed that a greedy algorithm achieves similar results. \cite{demaine2009node} gave improved algorithms for group Steiner tree on planar graphs

%	Consider the following hop-bounded version of Group Steiner Tree (HBGST). we are given a graph $G = (V, E)$ with cost function $c : E \to \mathbb{R}_+$ as well as groups $g_1, g_2, \ldots, g_k \subseteq V$ and a hop bound $t$. Our goal is to build a tree $T \subseteq G$ with hop diameter at most $t$ such that $T \cap g_i \neq \emptyset$ for all $i$. Let $S^*$ be the minimum-cost such tree. We assume that such a $S^*$ always exists (i.e.\ that there is always a tree with diameter at most $t$ which connects all groups). 

    \textbf{Algorithm:} Our algorithm will reduce solving hop-constrained group Steiner tree to a series of group Steiner tree problems on trees. For this reason, we restate the following known result for group Steiner tree on trees.
\begin{theorem}[\cite{garg2000polylogarithmic}]\label{thm:SGonTrees}
    There exists a randomized algorithm which with high probability given an instance of group Steiner tree on a tree returns a solution of cost at most $O(\OPT \cdot \log N \log k)$.
\end{theorem}

    We might naively hope to realize the usual tree embedding template: sample $O(\log n)$ partial tree embeddings using \Cref{thm:mainEmbed}, apply \Cref{thm:SGonTrees} to the resulting trees and then project the solutions back to the input graph. However, things are not so simple: since our partial tree embeddings only embed a subset of nodes, the optimal solution on the group Steiner tree problem on each of our partial tree embeddings has no guarantees of being close to optimal. For example, suppose $g_i = \{v_i, v_i'\}$ where $w((r, v_i)) = \epsilon$ for some small $\epsilon > 0$ and $w((r, v_i')) = 1$ for each group $g_i$. Then, if our partial tree embedding embeds $v_i'$ but not $v_i$, connecting $g_i$ to $r$ on our tree will be arbitrarily more expensive than $\OPT$.

    We solve this issue by relaxing the instance of hop-constrained Steiner tree that we solve on each of our trees. Specifically, we will randomly merge groups so that there always exists a low cost group Steiner solution. For example, in the above example we could randomly partition groups into super-groups each consisting of $\Theta(\log n)$ groups. If we then solved group Steiner tree on our tree on these super-groups we would know that---by a standard Chernoff-union bound proof---every super-group has at least one constituent $v_i$ embedded on the tree and so the optimal group Steiner tree problem restricted to embedded nodes on our tree still has low cost.

Formally, our algorithm for constructing our solution $T$ is as follows. Initially every group is active; we let $a$ be the number of active groups throughout our algorithm.
\begin{enumerate}
    \item For phase $j \in [1000 \log n \log k]$ or until $a \leq 10 \log n$
    \begin{enumerate}
        \item For iteration $\ell \in [1000 \log n]$
        \begin{enumerate}
        \item Apply \Cref{thm:mainEmbed} with hop bound $8h$, $\eps = .1$ and root $r$ on $G$ to get partial tree embedding $T_{j\ell}$
%        \item If $r \not \in V(T_l)$ move onto the next iteration
        \item Let $(g_{j\iota}')_{\iota=1}^{k'}$ form a uniformly random partition of the vertices of all active groups where $k' = \lceil \frac{a}{10 \log n}\rceil$ and an active $g_i \subseteq g_{j\iota}'$ with probability $1/k'$
        \item Apply \Cref{thm:SGonTrees} to the group Steiner tree instance on $T_{j\ell}$ with root $r$ and groups $(g_{j\iota}')_{\iota=1}^{k'}$ to get back solution $H_{j\ell}$
        \end{enumerate}
    \item Let $j^* := \argmin_\ell w_{T_{jl}}(H_{j\ell})$
    \item Add $T_e^G$ to $T$ for every $e \in E(H_{jj^*})$
    \item Set all $g_i$ which are now connected to $r$ by $T$ as inactive
    \end{enumerate}
    \item For the up to $10 \log n$ remaining active groups we add to $T$ the shortest path in $G$ from $r$ to each such group with at most $h$ hops
    \item Lastly, we set $T$ to be a BFS tree on $T$ rooted at $r$ to ensure that $T$ is a tree
\end{enumerate}

    We apply \Cref{thm:projSoln} and a standard Chernoff-union-bound-type argument to argue that each of our instances of group Steiner tree on a tree have a cheap solution.
    \begin{lemma}\label{lem:HBGST}
        Fix a phase $j$ and let $\OPT_\ell$ be the cost of the optimal group Steiner tree instance on $T_{j\ell}$ with root $r$ and groups $(g_{j\iota}')_{i=1}^{k'}$. Then $\min_\ell \OPT_\ell \leq O(\log ^ 2 n \cdot \OPT)$ with probability at least $1- \frac{1}{n^5}$.
    \end{lemma}
    \begin{proof}
        Since we have fixed a $j$, for ease of notation we let $T_\ell := T_{j\ell}$ and $g_{\iota}' = g_{j \iota}'$ for the remainder of the proof.
        
        We will construct a solution $T'_\ell$ for every $\ell$ which has cost at most $O(\log ^ 2 n \cdot \OPT)$ and which is feasible for the aforementioned group Steiner instance with probability at least $\frac{1}{3}$. Our claim will then immediately follow from this and the fact that the feasibility of each $T_\ell'$ will be independent, meaning with probability at least $1 - (\frac{1}{3})^{1000 \log n} \geq 1 - \frac{1}{n^5}$ there is some feasible $T_\ell'$ with cost at most $O(\log ^ 2 n \cdot \OPT)$.
        
        Fix an arbitrary $\ell$. Let $T^*$ be the optimal solution to our $h$-hop-constrained group Steiner tree problem on $G$ and let $W := V(T_{\ell})$ be all vertices embedded by $T_{\ell}$. Let $T_{\ell}' := T_j(T^*, h)$ where $T_j(T^*, h)$ is as defined in \Cref{dfn:proj}. By \Cref{thm:projSoln} we have $w_{T_j}(T_j(T^*, h)) \leq O(\log ^ 2 n \cdot \OPT)$ as desired.

        We now argue that $T'_\ell$ is feasible with probability at least $\frac{1}{3}$. It suffices to show that some vertex from $g_\iota'$ is in $W \cap V(T^*)$ for every $\iota \in [k']$ with probability at least $\frac{1}{3}$.
        
        %We first argue that some vertex from $g_\iota'$ is in $W \cap V(T^*)$ for every $\iota \in [k']$ with probability at least $\frac{1}{2}$. 
        Let $\mcI := \{i : g_i \cap W \cap V(T^*) \neq \emptyset \}$ be all groups with at least one embedded vertex from the optimal solution. We know by \Cref{thm:mainEmbed} and linearity of expectation that $\E[|\mcI|] \geq .9 a$ but since $|\mcI| \leq a$, it follows by Markov's inequality that $\Pr(|\mcI| \geq .8 a) \geq \frac{1}{2}$.
        
%        Say that a group is optimally embedded if $g_i \cap V(T^*) \cap W \neq \emptyset$. Let $\mcI$ be the indices of all optimally embedded groups and let $\vec{\mcI}$ range over the support of $\mcI$.
        
        Fix a super-group $g_\iota'$. For group $g_i$, let $X_i$ be the indicator of whether $g_i \subseteq g_\iota'$. Similarly, let $\vec{\mcI}$ be a fixed value in the support of $\mcI$ and let $X^{(\vec{\mcI})}_\iota := \sum_{i \in \vec{\mcI}} X_i$. Notice that $\E[X_\iota^{(\vec{\mcI})}] \geq |\vec{\mcI}| \cdot \frac{1000 \log n}{a}$ and that for a fixed $\vec{\mcI}$ if $\mcI = \vec{\mcI}$ and $X_\iota^{(\vec{\mcI})} \geq 1$ then $T'_\ell$ will connect $g_\iota'$ to $r$.
        
        Since for a fixed $\vec{\mcI}$ we know that each $X_i$ in $\sum_{i \in \vec{\mcI}} X_i$ is independent, a Chernoff-bound shows that
        \begin{align*}
            \Pr\left(X_\iota^{(\vec{\mcI})} \leq |\vec{\mcI}| \cdot \frac{900 \log n}{a}\right) \leq \exp \left(- \frac{(.1)^2 \cdot 1000 \log n \cdot |\vec{\mcI}|}{3a} \right) \leq \exp\left( - \frac{3|\vec{\mcI}| \cdot \log n }{a}\right)\\
        \end{align*}
        
        It follows that if $|\vec{\mcI}| \geq .8a$ we have that $\Pr \left(X_{\vec{\mcI}} = 0 \right) \leq \frac{1}{n^2}$. Combining this with a union bound and the fact that $|\mcI|$ is at least $.8a$ with probability at least $\frac{1}{2}$, we have that $T'_\ell$ contains a vertex from every $g_\iota'$ except with probability
        \begin{align*}
        \sum_{\vec{\mcI}} \Pr(\mcI = \vec{\mcI}) \cdot \Pr(X_\iota^{(\vec{\mcI})} = 0 \text{ for some $\iota$}) &\leq \sum_{(\vec{\mcI})} \Pr(\mcI = \vec{\mcI}) \sum_{\iota} \Pr(X_\iota^{(\vec{\mcI})} = 0)\\
        & = \sum_{\vec{\mcI} : |\vec{I} < .8 a|} \Pr(\mcI = \vec{\mcI}) \sum_{\iota} \Pr(X_\iota^{(\vec{\mcI})} = 0) \\  & \qquad+ \sum_{\vec{\mcI} : |\vec{I} \geq .8 a|} \Pr(\mcI = \vec{\mcI}) \sum_{\iota} \Pr(X_\iota^{(\vec{\mcI})} = 0)\\
        &\leq \frac{1}{2} + \frac{1}{2n^3}\\
        & \leq \frac{2}{3}.
        \end{align*}        

%        Lastly, notice that $r \in T_\ell'$ except with probability $.1$ by \Cref{thm:mainEmbed}. Thus, by a union bound we have that except with probability at most $\frac{1}{3}$ it holds that $T_\ell'$ contains $r$ and a vertex from every $g_\iota'$.
    \end{proof}

    We conclude with our approximation algorithm for hop-constrained group Steiner tree.
    \begin{restatable}{theorem}{HCGS}
        There is a poly-time algorithm which given an instance of $h$-hop-constrained group Steiner tree returns a tree $T$ such that $w(T) \leq O(\log ^ 3 n \log N \log ^2 k \cdot \OPT)$ and $\hop_H(g_i, r) \leq O(h \cdot \log ^ 3 n)$ with high probability for every $g_i$.
    \end{restatable}
	\begin{proof}
        We use the algorithm described above. We first argue our cost bound. By \Cref{lem:HBGST}, \Cref{thm:SGonTrees} and a union bound over all phases we have that with high probability $w_{T_{jj^*}}(H_{jj^*}) \leq O(\log ^ 2 n \log N \log ^2 k \cdot \OPT)$ for every $j$. Since each of the at most $10 \log n$ shortest paths we buy cost at most $\OPT$, we can apply the properties of our embeddings and sum up over all phases, to see that $w_G(T) \leq O(\log ^ 3 n \log N \log ^2 k \cdot \OPT)$. 
        
        Next, notice that $T$ satisfies the stated hop bounds by \Cref{thm:mainEmbed}.
        
        To see why $T$ connects all groups notice that in a given phase $j$ where we have $a$ unconnected groups and $a \geq 10 \log n$, we newly connect at least $\frac{a}{10 \log n} $ groups. Thus, the number of unconnected groups after this iteration is at most $(1 - \frac{1}{10 \log n})a$. Assume for the sake of contradiction that the number of unconnected groups after $1000 \log n \log k$ phases is more than $10 \log n$. Then, we have that the number of unconnected groups after $1000 \log n \log k$ phases is at most,
        \begin{align*}
        k \cdot \left(1 - \frac{1}{10 \log n} \right)^{1000 \log n \log k} \leq k \cdot e^{- \log k} \leq 1
        \end{align*}
        a contradiction.
    \end{proof}

\subsection{Hop-Constrained $k$-Steiner Tree}\label{sec:HCkST}

In this section we give a bicriteria approximation algorithm for the hop-constrained $k$-Steiner tree problem and a relaxed version of it. Notably, unlike most other problems to which we apply our embeddings, the hop-constrained version of $k$-Steiner tree and its relaxed version have previously been studied under the name ``Shallow-Light $k$-Steiner Trees''~\cite{hajiaghayi2009approximating, khani2011improved}. While both our techniques and prior work yield bicriteria approximation algorithms with polylogarithmic guarantees, our techniques are simpler (i.e., follow directly from the theory of partial tree embeddings), and for the relaxed problem give the best known cost approximation (at the cost of a worse hop stretch than known results).

\textbf{Problem:} Let $G = (V, E, w_G)$ be a weighted graph. Given a terminal set $S \subseteq V$, an integer $1 \le k \le |S|$, and a root $r \in V$ we want to find the connected subgraph $H \subseteq G$ which minimizes $w_G(H) := \sum_{e \in E(H)} w_G(e)$ that contains $r$ and has at least $k$ terminals (i.e., $|V(H) \cap S| \ge k$). In the hop-constrained version, we are additionally given a hop constraint $h \ge 1$ and need to satisfy that the hop diameter of $H$ is at most $h$ (i.e., $\hop_H(u, v) \le h$ for all $u, v \in V(H)$). In the relaxed version, we must find an $h$-hop-diameter subgraph $H$ with at least $k/8$ terminals, but we compare our cost to the optimal solution on $k$-terminals whose value we denote $\OPT$. 

%Finally, we define the (hop-constrained) $k'$-relaxed $k$-Steiner tree problem, in which one needs to find $H$ that contains at least $k' \le k$ terminals, but when talking about approximation quality we compare its cost with the optimal solution that connects all $k$ terminals (i.e., the $\OPT$ does not change between the relaxed and non-relaxed version).

\textbf{Related work:} \cite{hajiaghayi2009approximating} solve the relaxed hop-constrained $k$-Steiner tree with $O(\log n)$ hop stretch and $O(\log^3 n)$ cost approximation. They show that the relaxed and non-relaxed problems are equivalent up to a $O(\log k)$ factor in the cost and they use a black-box reduction to reduce the relaxed problem to a new problem (without hop constraints) called the ``buy-at-bulk $k$-Steiner tree problem'' (which we do not define here), achieving a $O(\log^4 n)$ cost approximation and $O(\log ^ 2 n)$ hop approximation. \cite{khani2011improved} improve the hop-constrained $k$-Steiner tree guarantee to $O(\log n)$ hop stretch and $O(\log^2 n$) cost approximation by improving the buy-at-bulk $k$-Steiner tree cost approximations.

%, in which one is given two edge cost functions: a ``buying'' cost $b$ and a ``connecting'' cost $c$; the objective is to find a subgraph $H$ with $k$ terminals which minimize the sum of buying the edges $\sum_{e \in E(H)} b(e)$ and connecting the terminals to the root inside $H$ $\sum_{\text{terminals $v$ in $H$}} d_{H, c}(r, v)$.
%As in previous work, we will first solve the relaxed $k$-Steiner tree problem and then use this solution to solve $k$-Steiner tree.

\textbf{Algorithm for the $h$-hop relaxed $k$-Steiner tree problem:} We sample an $8h$-hop partial tree embedding $T$ of $G$ with root $r$, hop stretch $O(\log^3 n)$, worst-case distance stretch $O(\log^2 n)$, and exclusion probability $\frac{1}{4}$ via \Cref{thm:mainEmbed}. Let $H'$ be the optimal $\frac{k}{8}$-Steiner tree solution (without hop constraints and containing the root) on $T$ with the terminal set $V(T) \cap S$, which can be found with a standard (poly-time) dynamic programming algorithm. We return $H := \bigcup_{e \in E(H')} T^G_e$, i.e., the projection of $H'$ back to $G$.

\begin{lemma}
  There is a poly-time algorithm for relaxed hop-constrained $k$-Steiner tree which produces a solution $H$ that contains $r$, at least $k/8$ terminals, and has hop diameter $O(\log^3 n \cdot h)$. With constant probability, $H$ satisfies $w_G(H) \le O(\log^2 n \cdot \OPT)$.
\end{lemma}
\begin{proof}
  Suppose that $H^*$ is the optimal solution of weight $w_G(H^*) = \OPT$ and hop diameter at most $h$. Furthermore, let $S_\OPT := S \cap V(H^*)$ be the set of terminals in the optimal solution. Since $T$ was sampled from a distribution with exclusion probability $1/4$, we have that $\E[|V(T) \cap S_{\OPT}|] = k/4$, hence we have $\Pr[|V(T) \cap S_{\OPT}| \ge k/8] \ge \frac{k/4 - k/8}{k - k/8} = 1/7$. Furthermore, we can find a ``projection'' $F^* = T(H^*, h) \subseteq T$ of $H^*$ to the tree $T$ (using \Cref{thm:projSoln}) where $w_G(F^*) \le 4\alpha \cdot \OPT$ and $|V(F^*) \cap S| = |V(T) \cap S_{\OPT}| \ge k/8$ with constant probability. Therefore, there exists an optimal solution $F^*$ on $T$ solving the (un-hop-constrained) $k/8$-Steiner tree problem with value at most $4\alpha \cdot \OPT$. Hence $w_G(H') \le 4 \alpha \cdot \OPT$ with constant probability.

  Finally, we project $H'$ back to $G$ to obtain the output $H$ and deduce that the hop diameter is $O(\log^3 n \cdot h)$ and $w_G(H) \le 4 \alpha \cdot OPT$ with constant probability (since the partial embeddings are dominating, i.e., $w_G(T_e^G) \le w_T(e)$), as required.
\end{proof}

\textbf{Algorithm for the (non-relaxed) $h$-hop $k$-Steiner tree problem:} We can easily boost the $k/8$-Steiner algorithm from constant probability to high probability by repeating it $O(\log n)$ times and taking the minimum solution. We can then apply this high-probability algorithm $O(\log k)$ times and take the union of the results as our solution. Note that in each iteration, a $1/8$-fraction of the remaining terminals will be added to the final solution and so by a standard covering argument $O(\log k)$ iterations suffice to cover all terminals. The hop diameter does not increase during the iterations since all solutions share a common root $r$, while the cost of our solution increases by $O(\OPT \cdot \log ^ 2 n)$ in each iteration. This proves the following result.

\begin{lemma}
  There is a poly-time algorithm for the $k$-Steiner tree problem which outputs a solution $H$ with $w_G(H) \le O(\log^2 n \cdot \log k) \cdot \OPT$ and hop diameter at most $O(\log^3 n \cdot h)$ with high probability.
\end{lemma}

\subsection{Hop-Constrained Oblivious Network Design}\label{sec:HCOND}

In this section we give a bicriteria approximation algorithm for hop-constrained oblivious network design which generalizes the hop-constrained version of many well-studied oblivious network design problems.

\textbf{Problem:} In the oblivious network design problem we are given a weighted graph $G = (V(G), E(G), w_G)$, a monotone subadditive function $f : \mathbb{R}_{\ge 0} \to \mathbb{R}_{\ge 0}$ (satisfying $f(a + b) \le f(a) + f(b)$ and $f(a) \le f(a+b)$ for all $a, b \ge 0$). For each pair of vertices $(u,v) \in V \bigtimes V$ we need to select a single path $P_{uv}$ between $u$ and $v$. An adversary then reveals a set of $k$ demand pairs $\{(s_i, t_i)\}_{i=1}^k$, inducing our solution $\bigcup_{i} P_i$ where $P_i : = P_{s_i t_i}$. For an edge $e \in E(G)$ let $\ell_e := |\{i : P_i \ni e\}|$ be the ``load'' of our induced solution: that is, the number of paths passing through $e$. The cost of our induced solution is $\sum_{e \in E(G)} w_G(e) \cdot f(\ell_e)$. In hop-constrained oblivious network design we are additionally given a hop constraint $h \ge 1$ and require that each $P_{uv}$ satisfies $\hop(P_{uv}) \le h$ for all $i$. (Non-oblivious) network design is identical but we are shown the demand pairs before we must fix our paths. We emphasize that $\OPT$ in this section will refer to the cost of the optimal hop-constrained network design problem; that is, the cost of the optimal solution which knows the demand pairs \emph{before} it fixes its paths.

\textbf{Related Work:} \cite{gupta2006oblivious} introduced the oblivious network design problem as an oblivious generalization of many well-studied problems such as Steiner forest and buy-at-bulk network design \cite{awerbuch1997buy}.

\textbf{Algorithm:} We sample $5h$-hop partial tree embeddings with trees $T_1, T_2, \ldots, T_{O(\log n)}$ by applying \Cref{thm:mainEmbed} to $G$ with $\eps = .1$ and an arbitrary root. Next, we fix a $T_j$ and do the following. We let $\mcS_j := \{ (u,v) : u, v \in V(T_j) \}$ be the pairs of vertices embedded by $T_j$. For each pair $(u,v) \in \mcS_j$, we let $P_{uv}' \gets (T_j)_{uv}$ be the unique simple path between $u$ and $v$ in $T_j$ (possibly overwriting a previous $P_{uv}'$). Lastly, we let $P_{uv} := (T_j)_{uv}^G$ be $P_{uv}'$'s projection onto $G$.

Our proof will use the usual partial tree-embedding template along with the idea of mixture metrics which we introduced in \Cref{sec:decomposition-lemma}. We let $\mcI_j$ be the 
\begin{lemma}
  \label{lemma:obl-network-design-cost}
  Given the revealed demand pairs, let $\mcI_j$ be the indices of all pairs with vertices in $T_j$ and let $\ALG_j$ be the optimal cost of the network design problem (without hop constraints) on the tree $T_j$ with demand pairs $\calI_j$. Then $\ALG_j \le O(\log^3 n \cdot \OPT)$ with high probability.
\end{lemma}
\begin{proof}
  We introduce some notation. We write $T := T_j$, $\ALG := ALG_j$, and $\calI := \calI_j$ for brevity. Let $\alpha = O(\log^2 n)$ be the worst-case distance stretch of $T$ (as in \Cref{thm:mainEmbed}). For a tree edge $e = \{u, v\} \in E(T)$ where $u$ is the parent of $v$, we define $S_e \subseteq V(T)$ as the set of nodes in the subtree of $e$ (including $v$, but excluding $u$). Furthermore, given a set $W \subseteq V(T)$ we write $\unmatched(W) := |\{ i \in \calI : \{s_i, t_i\} \cap W = 1 \}|$ as the number of ``unmatched terminals'' in $W$.

  Remember that $T$ is well-separated (as stipulated by \Cref{thm:mainEmbed}), meaning each
  root-to-leaf has edges of weights that are decreasing powers of 2. Specifically, $w_T(e) = 2^p$ for some $p$. From now on we fix a value $p$.

  We define ``mixture weights'' $w'_G$ on the graph $G$: for an edge $e \in E(G)$ we define $w'_G(e) := 1/h + w_G(e) \cdot \frac{ 5 \alpha }{2^p}$. By $d'_G : V(G) \times V(G) \to \mathbb{R}_{\ge 0}$ we denote the distances induced by $w'_G$. Furthermore, given a node $v \in V(G)$, radius $r > 0$, we define the ``ball'' $B'_G(r, v)$ as the set of all (fractional) edges $f$ such that there exists a path $Q$ in $G$ between $v$ and (any endpoint of) $f$ with $w'_G(Q) \le r$. Here we consider an edge to be subdivided into infinitesimal pieces, hence one can talk about a fractional portion of an edge---while this can be made fully formal by considering a version of $G$ where an edge $e \in E(G)$ is subdivided into $\xi \to \infty$ pieces of weight $w_G(e)/\xi$, hops $1/\xi$, and mixture weight $w_G'(e)/\xi$, we choose to keep it slightly informal for simplicity. Next, we extend $B'_G(r, W) := \bigcup_{v \in W} B_G'(r, v)$ for a subset $W \subseteq V(G)$.

  For each tree edge $e \in E(T)$ of weight $w_T(e) = 2^p$ we associate the ball $A_e := B_G'(2, S_e)$ to $e$. % We write $V(A_e)$ to denote the nodes of $V(T)$ inside the annulus. \gznote{are we using this?}
  We show that the balls assigned to different edges $e$ and $f$ of the same weight $2^p$ are disjoint. Clearly, since $T$ is well-separated, there is no root-leaf path that contains both $e$ and $f$. Note that $2^p \le d_T(u, v) \le \alpha \cdot d_G\at{5h}(u, v)$ for $u \in S_e$ (subtree below $e$) and $v \not \in S_e$, implying $d_G\at{5h}(S_e, S_f) \ge \frac{2^i}{\alpha}$. Therefore, any path $Q$ between (a node in) $S_e$ and (a node in) $S_f$ has $\hop(Q) > 5h$ or $w_G(Q) \le \frac{2^i}{\alpha}$. Since $w_G'(Q) = \hop(Q)/h + w_G(Q) \cdot \frac{ 5 \alpha }{2^p}$ we have that $w_G'(Q) \ge 5$. We conclude that $d'_G(S_e, S_f) \ge 5$. Therefore, the balls $A_e$ and $A_f$ (of radius 2) are disjoint.

  % Next, since $d_G\at{4h}(S_e, S_f) \ge \frac{2^i}{\alpha} > \frac{2^i}{3 \alpha} + \frac{2^{i'}}{3 \alpha}$ we conclude that $B_G\at{2h}(\frac{2^i}{3 \alpha}, S_e) \supseteq A_e$ and $B_G\at{2h}(\frac{2^{i'}}{3\alpha}, S_f) \supseteq A_f$ are disjoint, hence the annuli $A_e$ and $A_f$ are disjoint (otherwise there would be a path $Q$ between a node in $S_e$ and $S_f$ of $\hop(Q) \le 4h$ and weight $w_G(Q) \le \frac{2^i}{3 \alpha} + \frac{2^{i'}}{3 \alpha} < d\at{4h}_G(S_e, S_f)$, a contradiction). The other case we consider is when $e$ and $f$ are on the same root-to-leaf path (hence $i > {i'}$ and $e$ is an ancestor of $f$). Let $A_e^C$ be the complement of $A_e$, then we have $A_c^C \supseteq B_G\at{2h}(\frac{2^i}{6\alpha}, S_e) \supseteq B_G\at{2h}(\frac{2^{i'}}{3\alpha}, S_f) \supseteq A_f$, hence the annuli $A_e$ and $A_f$ are also disjoint.

  Let $\{P_i^*\}_i$ be an optimal hop-constrained solution on $G$. Fix a tree edge $e \in E(T)$ of weight $w_T(e) = 2^p$. With a slight abuse of notation, let $P_i^* \cap A_e$ be the sub-path of $P_i^*$ from its start in $S_e$ to the first node (in the infinitesimal graph) not contained in the ball $A_e$. We claim that for each unmatched terminal $s_i \in S_e$ (or $t_i$, but we will WLOG assume it is $s_i$) it holds $w_G(P_i^* \cap A_e) \ge \frac{2^p}{5 \alpha}$. First, since by definition $s_i \in S_e$ and $t_i \not \in S_e$, then $d'_G(s_i, t_i) \ge 5$ (as in the previous paragraph). Therefore, since the radius of $A_e$ is $2 < 5$ we have that $w'_G(P^*_i \cap A_e) \ge 2$. Furthermore, let $Q := P^*_i \cap A_e$, and we have $2 \le w'_G(Q) = \hop(Q)/h + w_G(Q) \cdot \frac{ 5 \alpha }{2^p} \le 1 + w_G(Q) \cdot \frac{ 5 \alpha }{2^p}$ giving us $w_G(Q) \ge \frac{2^p}{5\alpha}$ as claimed.

  Continuing to fix $\{P_i^*\}_i$ and $e \in E(T)$ with $w_T(e) = 2^p$, we define $\ell^*_f$ to be the load of $\{ P_i^* \}_i$ on any edge $f \in E(G)$ and then define $\OPT(p, e) := \sum_{f \in A_e} w_G(f) \cdot \ell_f^*$. Furthermore, since for each $p$ the balls $\{ A_e \}_{e \in E(T), w_T(e) = 2^p}$ are disjoint, we conclude that (for each $p$)
  \begin{align*}
    \OPT = \sum_{f \in E(G)} w_G(f) \ell^*_f \ge \sum_{e \in E(T), w_T = 2^p} \sum_{f \in A_e} w_G(f) \ell^*_f = \sum_{e \in E(T), w_T(e) = 2^p} \OPT(p, e) .
  \end{align*}

  Fix tree edge $e \in E(T)$ with $w_T(e) = 2^p$. As proven before, for each $i$ where $s_i$ or $t_i$ are an unmatched terminal in $S_e$ we have that $w_G(P_i^* \cap A_e) \ge \frac{w_T(e)}{5 \alpha}$ and all such $\{P_i^* \cap A_e)$ are contained within the same set of edges $A_e$. By subadditivity, the contribution to $\OPT(p, e)$ is minimized when the paths $\{ P_i' \cap A_e \}_{\text{unmatched } i \text{ in } S_e}$ are identical for all $i$, leading to a contribution of at least $\OPT(p, e) \ge \frac{2^p}{5\alpha} f(\unmatched(S_e))$.

  Note that there are at most $O(\log n)$ values for $p$ since $T$ is well-separated and the aspect ratio of $G$ is $\poly(n)$. Therefore, the value of the algorithm can be written as
  \begin{align*}
    \ALG & \le \sum_{p = 1}^{O(\log n)} \sum_{e \in E(T), w_T(e) = 2^p} 2^p f(\unmatched(S_e)) \le \sum_{p = 1}^{O(\log n)} \sum_{e \in E(T), w_T(e) = 2^p} O(\alpha) \OPT(p, e) \\
         & = O(\alpha) \sum_{p = 1}^{O(\log n)} \OPT .
  \end{align*}
  Therefore, $\ALG \le O(\alpha \log n \cdot \OPT) = O(\log^3 n \cdot \OPT)$.
\end{proof}
We conclude the analysis of the above algorithm.
\begin{theorem}
  There is a poly-time algorithm for $h$-hop-constrained oblivious network design which with high probability outputs a selection of paths $\{ P_{uv} \}_{(u,v) \in V \times V}$ each with at most $O(\log^3 n \cdot h)$ hops such that the induced solution for any set of demand pairs has cost at most $O(\log^4 n \cdot \OPT)$.
\end{theorem}
\begin{proof}
  First, since we sample from a partial tree distribution $\mcD$ with $\Pr_{(T, \cdot) \sim \mcD}[v \in V(T)] \ge 0.9$, we conclude using a standard Chernoff and union bound that both nodes of each demand pair $(s_i, t_i)$ appear in at least one partial tree embedding $T_j$, with high probability. Therefore, with high probability, each pair is assigned a valid path at least once. Furthermore, since the distribution is $h$-hop with hop stretch $O(\log^3 n)$, we have that $\hop(P_i) \le O(\log^3 n) h$.

  Let $\ALG_j$ be the cost of the (unique) solution in the tree $T_j$ with respect to all demand pairs $\calI_j$. We have $\ALG_j \le O(\log^3 n \cdot \OPT)$ by \Cref{lemma:obl-network-design-cost}. Since for each $j$ we purchase a subset of the paths corresponding to the cost $\ALG_j$ solution on $T_j$ and since projecting such a solution back to $G$ only decreases its cost and there are $O(\log n)$ trees $T_j$, we conclude that our cost is at most $O(\log^4 n) \OPT$.
\end{proof}

\section{$h$-Hop Repetition Tree Embeddings}\label{sec:repTrees}
In the preceding section we showed how many hop-constrained problems reduce to solving $O(\log n)$ non-hop-constrained problems on trees by sampling $O(\log n)$ partial tree embeddings as in \Cref{thm:mainEmbed}. In this section, we show how to compactly represent $O(\log n)$ draws from \Cref{thm:mainEmbed} in a single ``$h$-hop repetition tree embedding'' to reduce many hop-constrained problems to a \emph{single} non-hop-constrained problem on a tree. This will allow us to give several online algorithms for hop-constrained problems and improve the approximation guarantees for offline hop-constrained group Steiner tree which we gave in the preceding section. Our notion of an $h$-hop repetition tree embedding here is an adaptation of a forthcoming paper of ours \cite{haeupler2020repTree} to the hop-constrained setting. Roughly, a repetition tree embedding is a tree embedding where a vertex maps to many copies of itself.

%to adapt our partial tree embeddings for online settings. A common strategy in online algorithms is to sample a single tree from the FRT distribution and then solve the online problem on the FRT embedding, projecting solutions from the tree to the original graph at each time step. However, since our partial embeddings only embed a subset of nodes we cannot simply sample a single tree. It will suffice for online hop-constrained Steiner forest to simply sample $\Theta(\log n)$ trees and run our online algorithm on each tree; however, this strategy will not work for more general connectivity problems as for these problems the exclusion of vertices in each tree can cause the optimal solution on each sampled tree to have cost much higher than the optimal solution on the original graph.

%We adapt a tree embedding which we recently introduced---the repetition tree embedding---to the hop-bounded case \cite{haeupler2020repTree}. Roughly speaking, a repetition tree embedding merges many partial tree embeddings into a single tree on which one can run one's online algorithm.

To precisely define our $h$-hop repetition tree embeddings we need a function $\phi$ from the vertex set $V$ to subsets of $V'$. $\phi(v)$ should be understood as the ``copies'' of $v$ in $V'$. We call such a function a vertex mapping.

% preserves connectivity and weights up to an $O(\log^2 n)$ factor and can 

\begin{definition}[Vertex Mapping]
    Given vertex sets $V$ and $V'$ we say $\phi : V \to 2^{V'}$ is a vertex mapping if $\phi$ is injective, $\phi(v) \neq \emptyset$ for all $v$ and $\{ \phi(v) : v \in V \}$ forms a partition of $V'$. For $v' \in V'$, we use the shorthand $\phi^{-1}(v')$ to stand for the unique $v \in V$ such that $v' \in \phi(v)$.
\end{definition}
We also say that a mapping $\pi : E \to E'$ between edge-sets is monotone if for every $A \subseteq B$ we have that $\pi(A) \subseteq \pi(B)$ and that a collection of edges $F$ connects sets $U$ and $U'$ if there are vertices $u \in U$ and $u' \in U'$ connected by $F$.

        \begin{definition}[$h$-Hop Repetition Tree Embedding]
    Let $G = (V, E, w)$ be a weighted graph with some specified root $r \in V$. An $h$-hop repetition tree embedding with cost stretch $\alpha$ and hop stretch $\beta$ consists of a weighted tree $T = (V', E',w')$, a vertex mapping $\phi : V \to 2^{V'}$ and monotone edge mappings $\pi_{G \to T} : 2^E \to 2^{E'}$ and $\pi_{T \to G} : 2^{E'} \to 2^{E}$ such that:
    \begin{enumerate}
        \item \textbf{$\alpha$-Approximate Cost Preservation}: For any $F \subseteq E$ we have $w(F) \leq \alpha \cdot w'(\pi_{G \to T}(F))$ and for any $F' \subseteq E'$ we have $w'(F') \leq w(\pi_{T \to G}(F'))$.
        \item \textbf{$\beta$-Approximate $h$-Hop-Connectivity Preservation:} For all $F \subseteq E$ and $u,v \in V$ if $\hop_F(u,v) \leq h$, then $\phi(u), \phi(v) \subseteq V'$ are connected via $\pi_{G \to T}(F)$. Symmetrically, for all $F' \subseteq E'$ and $u', v' \in V'$ if $u'$ and $v'$ are connected by $F'$ then $\hop_{\pi_{T \to G}(F')}(\phi^{-1}(u'), \phi^{-1}(v')) \leq \beta h$.
        \item \textbf{Root mapping:} $|\phi(r)| = 1$.
    \end{enumerate}
    We say a repetition tree embedding is efficient if $\phi$, $\pi_{G \to T}$ and $\pi_{T \to G}$ are all deterministically poly-time computable.
\end{definition}

A simple consequence of our main embedding theorem (\Cref{thm:mainEmbed}), along with our projection mapping theorem (\Cref{thm:projSoln}), shows that we can compute $h$-hop repetition tree embeddings with poly-log hop and cost stretch.
\begin{restatable}{theorem}{repTree}\label{thm:repTree}
    Given $h \geq 1$, there is a poly-time algorithm which given any weighted graph $G =(V,E,w)$ and root vertex $r \in V$ computes an efficient $h$-hop repetition tree embedding from $G$ into some weighted and rooted tree $T$ with hop stretch $O(\log ^ 3 n)$ and cost stretch $O(\log ^ 3 n)$ with high probability. Further, $T$ is well-separated and satisfies $|\phi(v)| \leq O(\log n)$ for all $v$.
\end{restatable}
\begin{proof}
    We compute our $h$-hop repetition tree embedding as follows. First, apply \Cref{thm:mainEmbed} to compute $\Theta(\log n)$  $8h$-hop partial tree embeddings $T_1, T_2, \ldots$ with exclusion probability $\epsilon = .01$, root $r$, worst-case distance stretch $O(\log ^ 2 n)$ and hop stretch $O(\log ^ 3 n)$.
    %we discard any trees which do not contain $r$; as $r$ independently occurs in each tree with probability $.99$, a Chernoff bound shows that after discarding trees we still have $\Theta(\log n)$ trees with high probability. 
    We let our repetition tree embedding $T$ be the result of identifying $r$ in each of our $T_i$ as the same vertex; that is, $V(T) = \{r\} \sqcup  \bigsqcup_i V(T_i) \setminus \{r\}$. Let $\phi$ map from a vertex in $V$ to its copies in the natural way. We let $\pi_{G \to T}(F) := \bigcup_i T_i(G[F], h)$ where $T_i(G[F], h)$ is as defined in \Cref{dfn:proj}. We also let $\pi_{T \to G}(F') := \bigcup_{e \in F'} T_{e}^G$. Notice that our mappings are monotone by definition. Also notice that our tree satisfies root mapping by construction. Lastly,  our tree satisfies $O(\log ^ 3 n)$-approximate cost preservation as an immediate consequence of \Cref{thm:mainEmbed}, \Cref{thm:projSoln} and the fact that we sampled $\Theta(\log n)$ trees. Our tree satisfies $O(\log ^ 3)$-approximate $h$-hop-connectivity preservation by \Cref{thm:mainEmbed} and a Chernoff bound which shows that any $u,v$ have some copies that appear in the same $T_i$ with high probability. The well-separatedness and number of copies of each vertex trivially follow from \Cref{thm:mainEmbed}.
\end{proof}

\section{Applications of $h$-Hop Repetition Tree Embeddings}\label{sec:secondAppl}

In this section we apply our $h$-hop repetition tree embeddings to give approximation algorithms for the hop-constrained versions of group Steiner tree, online group Steiner tree, group Steiner forest and online group Steiner forest. As in our previous section, we let $\OPT$ stand for the optimal value of the relevant hop-constrained problem throughout.
% As above, we assume throughout this section that all of our graphs are complete; this assumption is without loss of generality since in each of the problems we study we may always add high cost (i.e.\ cost $L$) edges to our graph to make it complete without changing the optimal solution or which edges are chosen by our approximation algorithms; we may assume that no optimal solution buys any cost $L$ edges. Throughout this section we let $\OPT$ be the cost of the optimal solution of the discussed problem.\gznote{I would remove all of this. It should be clear to the readers we solve the general case even thought formally our graph is weighted and complete (the argument is standard and already talked about in the notation). I think this makes the paper longer and puts unnecessary attention to it.}\enote{I'm fine removing this. However, I think when we state the fact that we assume complete graphs we should emphasize that we can only do so in the applications we consider because for \emph{these} problems it is WLOG. As in, there are other problems where the graph is not complete WLOG.}

\subsection{Hop-Constrained Group Steiner Tree}\label{sec:groupSTBetter}
Here we give an approximation algorithm for hop-constrained group Steiner tree which improves over our result in \Cref{sec:groupST} by using $h$-hop repetition trees. For a problem definition and related work see \Cref{sec:groupST}.

\textbf{Algorithm:} We first sample an $h$-hop repetition tree $T$ with high probability as in \Cref{thm:repTree} with mappings $\pi_{G \to T}$, $\pi_{T \to G}$ and $\phi$ and root $r$. Next, consider the group Steiner tree instance on $T$  whose root is the one vertex in $\phi(r)$ and whose groups are $(g_i')_i$ where $g_i' := \bigcup_{v \in g_i} \phi(v)$. We apply \Cref{thm:SGonTrees} to this group Steiner tree problem to get back tree $T' \subseteq T$ and let $H' := \pi_{T \to G}(T_t')$ be its projection onto $G$. We let our solution $H$ be a BFS tree of $H'$ rooted at $r$ where edges have unit cost in the BFS.

The properties of our $h$-hop repetition tree embeddings immediately show that this algorithm is competitive.
\begin{theorem}
    There is a poly-time algorithm which with high probability given an instance of $h$-hop-constrained group Steiner tree returns a tree $T$ such that $w(T) \leq O(\log ^ 3 n \log N \log k \cdot \OPT)$ and $\hop_H(g_i, r) \leq O(h \cdot \log ^ 3 n)$  for every $g_i$.
\end{theorem}
\begin{proof}
    We use the above algorithm. The root mapping and $\beta$-approximate $h$-hop connectivity preservation properties of our $h$-hop repetition tree embedding along with the feasibility of $T'$ guarantees that $H'$ connects $g_i$ to $r$ for every $i \leq t$ with at most $O(h \cdot \log ^ 3 n)$ hops; it follows that $H$ does the same. The bound on our cost comes from combining the fact that $\pi_{G \to T}(H^*)$ is feasible for the group Steiner tree instance we solve on $T$ where $H^*$ is the optimal solution on $G$, the $O(\log ^ 3 n)$-approximate cost preservation of our $h$-hop tree embedding and the cost guarantee of \Cref{lem:onGS}.
\end{proof}

\subsection{Online Hop-Constrained Group Steiner Tree}\label{sec:groupSTOnline}
In this section we show that our $h$-hop repetition trees reduce solving online hop-constrained group Steiner tree to online group Steiner tree on a tree; we then apply a known solution for online group Steiner tree on trees.

\textbf{Problem}: Online group Steiner tree is the same as group Steiner tree as defined in \Cref{sec:groupST} but where our solution need not be a tree and groups are revealed in time steps $t = 1, 2, \ldots$. That is, in time step $t$ an adversary reveals a new group $g_t$ and the algorithm must maintain a solution $T_t$ where: (1) $T_{t-1} \subseteq T_{t}$; (2) $T_t$ is feasible for the group Steiner tree problem on groups $g_1, \ldots g_t$ and; (3) $T_t$ is competitive with the optimal offline solution for this problem where the competitive ratio of our algorithm is $\max_t w(T_t)/ \OPT_t$ where $\OPT_t$ is the cost of the optimal offline group Steiner tree solution on the first $t$ groups. Online hop-constrained group Steiner tree is the same as group Steiner tree but we are also given a hop-constraint $h \geq 1$ and the optimal solution as well as each of our trees must satisfy $\hop_{T_t}(r, g_i) \leq h$ for every $i \leq t$. %We say that an online algorithm for online hop-constrained group Steiner tree is $(\alpha, \beta)$-competitive if it has competitive ratio $\alpha$ and $\hop_{T_t}(r, g_i) \leq \beta h$ for every $i \leq t$. 
We assume that the possible groups revealed by the adversary are known ahead of time as otherwise this problem is known to admit no sub-polynomial approximations \cite{alon2006general} and let $k$ be the number of possible groups revealed by the adversary.

\textbf{Related work:} \cite{alon2006general} gave the first poly-logarithmic online algorithm for group Steiner tree. Recently, \cite{barta2020online} gave the first online algorithm for this problem which does not have $n$ in its approximation ratio.

\textbf{Algorithm:} We recall a result of \cite{alon2006general} that solves online group Steiner tree on trees.
\begin{theorem}[\cite{alon2006general}]\label{lem:onGS}
    There is a poly-time algorithm for online group Steiner tree on trees with expected competitive ratio $O(\log ^ 2 n \log k)$.
\end{theorem}

We then combine this result with our $h$-hop repetition trees to get our algorithm for online hop-constrained group Steiner tree.

We first sample an $h$-hop repetition tree $T$ with high probability as in \Cref{thm:repTree} with root $r$ and mappings $\pi_{G \to T}$, $\pi_{T \to G}$ and $\phi$. Next, consider the online group Steiner tree instance on $T$  whose root is the one vertex in $\phi(r)$ where group $g_t' := \bigcup_{v \in g_t} \phi(v)$ is revealed in time step $t$. Apply \Cref{lem:onGS} to maintain solution $T_t'$ for this problem in time step $t$ on $T$ and let our solution in time step $t$ on $G$ be $T_t := \pi_{T \to G}(T_t')$.

The properties of our $h$-hop repetition tree embeddings immediately give the desired properties of our algorithm.
\begin{theorem}
    There is a poly-time algorithm for online hop-constrained group Steiner tree which with high probability maintains a solution $\{T_t\}_t$ that is $O(\log k \cdot \log ^ 5 n)$-cost-competitive in expectation where $\hop_{T_t}(r, g_i) \leq O(\log ^ 3 n \cdot h)$ for all $t$ and $i \leq t$.
\end{theorem}
\begin{proof}
    We use the above algorithm. The root mapping and $\beta$-approximate $h$-hop connectivity preservation properties of our $h$-hop repetition tree embedding along with the feasibility of $T_t'$ guarantees that our solution $T_t$ connects $g_i$ to $r$ for every $i \leq t$ with at most $O(h \cdot \log ^ 3 n)$ hops. The bound on our cost comes from combining the fact that $\pi_{G \to T}(T^*_t)$ is feasible for the group Steiner tree instance we solve on $T$ in time step $t$ where $T^*_t$ is the optimal solution on $G$ in time step $t$, the $O(\log ^ 3 n)$-approximate cost preservation of our $h$-hop tree embedding and the cost guarantee of \Cref{lem:onGS}.
\end{proof}

\subsection{Hop-Constrained Group Steiner Forest}\label{sec:GSF}
As both group Steiner tree and Steiner forest are special cases of it, the group Steiner forest is one of the most general studied connectivity problems; for this reason it is also sometimes referred to as the ``generalized connectivity problem.''

\textbf{Problem:} In the group Steiner forest problem we are given a weighted graph $G = (V, E, w)$ as well as pairs of subsets of nodes $(S_1, T_1), (S_2, T_2), \ldots, (S_k, T_k)$ where $S_i, T_i \subseteq V$. Our goal is to find a forest $F$ which is a subgraph of $G$ and in which for each $i$ there is an $s_i \in S_i$ and $t_i \in T_i$ such that $s_i$ and $t_i$ are connected in $F$. We wish to minimize our cost, $w(F) := \sum_{e \in E(F)} w(e)$. In hop-constrained group Steiner forest we are additionally given a hop bound $h \geq 1$ and for every $i$ we must ensure that $\hop_F(s_i, t_i) \leq h$ for some $s_i \in S_i$ and $t_i \in T_i$. We will use the shorthand $\hop_F(S_i, T_i) := \min_{s_i \in S_i, t_i \in T_i} \hop_F(s_i, t_i)$.

\textbf{Related Work:} \cite{alon2006general} introduced the group Steiner forest problem to study online network formation. \cite{chekuri2011set} gave the first poly-log approximation algorithm for group Steiner forest. \cite{naor2011online} gave a worse poly-log approximation for group Steiner forest but one which was based on tree embeddings which will be useful for our purposes.

\textbf{Algorithm:} We use our repetition tree embeddings to reduce hop-constrained group Steiner forest to the tree case. We then apply an algorithm of \cite{naor2011online} which shows how to solve group Steiner forest on trees.

\begin{theorem}[\cite{naor2011online}]\label{lem:GSFonTrees}
    There is a poly-time algorithm for group Steiner forest on trees of depth $d$ which achieves an approximation ratio of $O(d \cdot \log ^ 2 n \log k)$ with high probability.
\end{theorem}

Formally, we first apply \Cref{thm:repTree} to sample a repetition tree $T$ with depth $O(\log n)$, an arbitrary root and mappings $\phi$, $\pi_{G \to T}$ and $\pi_{T \to G}$. Next, we apply  \Cref{lem:GSFonTrees} to solve the group Steiner forest on $T$ with pairs to be connected $(S_1', T_1'), \ldots (S_k', T_k')$ where $S_i' := \bigcup_{v \in S_i} \phi(v)$ and symmetrically $T_i' := \bigcup_{v \in T_i} \phi(v)$. Let $F'$ be the resulting solution on $T$. We return as our solution $F:= \pi_{T \to G}(F')$.

\begin{theorem}
    There is a poly-time algorithm for $h$-hop-constrained group Steiner forest which with high probability returns a solution $F$ such that $w(F) \leq O( \OPT \cdot \log ^ 6 n \log k)$ and $\hop_{F}(S_i, T_i) \leq O(h \cdot \log ^ 3 n)$ for every $i$.
\end{theorem}
\begin{proof}
    We use the above algorithm. A polynomial runtime is immediate from \Cref{thm:repTree} and \Cref{lem:GSFonTrees}. The hop guarantee is immediate from the correctness of the algorithm of \Cref{lem:GSFonTrees} and the properties of $\pi_{T \to G}$ as given in \Cref{thm:repTree}. To see the bound on cost, notice that $\pi_{G \to T}(F^*)$ is feasible for the group Steiner forest problem that we solve on $T$ and has cost at most $O(\log ^ 3 n \cdot \OPT)$ by the properties of $\pi_{G \to T}$ as specified in \Cref{thm:repTree} where $F^*$ is the optimal solution to the input group Steiner forest problem on $G$. The bound then follows from the fact that $T$ is well-separated and so has depth at most $O(\log n)$.
\end{proof}

\subsection{Online Hop-Constrained Group Steiner Forest}\label{sec:onlineGSF}
    In this section we give our algorithm for online hop-constrained group Steiner forest. It follows, almost immediately, from our $h$-hop repetition tree embeddings and a result of \cite{naor2011online} for online group Steiner forest on trees.
    
    \textbf{Problem:} Online group Steiner forest is the same as group Steiner forest as defined in \Cref{sec:GSF} but each pair $(S_t, T_t)$ is revealed at time step $t = 1,2, \ldots$ by an adversary and in each time step $t$ we must maintain a forest $F_t$ which is feasible for pairs $(S_1, T_1), \ldots (S_t, T_t)$ so that $F_{t-1} \subseteq F_t$. The competitive ratio of an online algorithm with solution $\{F_t\}_t$ is $\max_t w(F_t) / \OPT_t$ where $\OPT_t$ is the optimal offline solution for the group Steiner forest problem we must solve in time step $t$. Online hop-constrained group Steiner forest is the same as online group Steiner forest but we are given a hop constraint $h$ and we must ensure that for each $t$ and $i \leq t$ there is some $s_i \in S_i$ and $t_i \in T_i$ such that $\hop_{F_t}(s_i, t_i) \leq h$. 
    %We say that an online algorithm for online hop-constrained group Steiner forest is $(\alpha, \beta)$-competitive if it has a competitive ratio $\alpha$ and for every $t$ and $i \leq t$ we have $\hop_{F_t}(S_i, T_i) \leq \beta h$. 
    We assume that the possible pairs revealed by the adversary are known ahead of time as otherwise this problem is known to admit no sub-polynomial approximations \cite{alon2006general} and let $k$ be the number of possible pairs.
    
    \textbf{Related work:} The existence of a poly-log-competitive online algorithm for group Steiner forest was first posed as an open question by \cite{chekuri2011set}. \cite{naor2011online} answered this question in the affirmative by showing that such an algorithm exists.
    
    \textbf{Algorithm:} We use our repetition tree embeddings to reduce online hop-constrained group Steiner forest to the tree case. We then apply an algorithm of \cite{naor2011online} which shows how to solve group Steiner forest on trees.
    
    \begin{theorem}[\cite{naor2011online}]\label{lem:GSFonTreesOnline}
        There is a randomized poly-time algorithm for group Steiner forest on trees of depth $d$ with expected competitive ratio $O(d \cdot \log ^ 3 n \log k)$.
    \end{theorem}
    
    Formally, we first apply \Cref{thm:repTree} to sample a repetition tree $T$ with depth $O(\log n)$, an arbitrary root and mappings $\phi$, $\pi_{G \to T}$ and $\pi_{T \to G}$. Next, we apply \Cref{lem:GSFonTreesOnline} to solve the online group Steiner forest on $T$ with pairs to be connected $(S_1', T_1'), \ldots (S_t', T_t')$ in time step $t$ where $S_i' := \bigcup_{v \in S_i} \phi(v)$ and symmetrically $T_i' := \bigcup_{v \in T_i} \phi(v)$. Let $F'_t$ be the resulting solution on $T$ in time step $t$. In time step $t$ we return as our solution on $G$ the subgraph $F_t:= \pi_{T \to G}(F'_t)$.
    
    We conclude with the properties of our online group Steiner forest algorithm.
    \begin{theorem}
        There is a poly-time algorithm for online $h$-hop-constrained group Steiner forest which with high probability maintains a solution $\{F_t\}_t$ that is $O(\log ^ 7 n \log k)$-cost-competitive in expectation where $\hop_{F_t}(S_i, T_i) \leq O(h \cdot \log ^ 3 n)$ for all $t$ and $i \leq t$.
    \end{theorem}
    \begin{proof}
        We use the above algorithm. A polynomial runtime is immediate from \Cref{thm:repTree} and \Cref{lem:GSFonTreesOnline}. The hop guarantee is immediate from the correctness of the algorithm of \Cref{lem:GSFonTreesOnline} and the properties of $\pi_{T \to G}$ as given in \Cref{thm:repTree}. To see the bound on cost, notice that $\pi_{G \to T}(F^*_t)$ is feasible for the group Steiner forest problem that we solve on $T$ and has cost at most $O(\log ^ 3 n \cdot \OPT)$ by the properties of $\pi_{G \to T}$ as specified in \Cref{thm:repTree} where $F^*_t$ is the optimal solution to the input group Steiner forest problem on $G$ in time step $t$. The bound then follows from the fact that $T$ is well-separated and so has depth at most $O(\log n)$.
\end{proof}

\section{Conclusion and Future Work}
In this work we showed that, while far from any metric, hop-constrained distances are well-approximated by partial tree metrics. We used this fact to develop new embeddings for hop-constrained distances which we then used to give the first bicriteria (poly-log, poly-log) approximation algorithms for many classic network design problems.

We conclude by giving directions for future work. Reducing the stretch in our embeddings, or proving lower bounds stronger than those immediately implied by the FRT lower bounds is our main open question. Improving the upper bounds in our embeddings---as in the FRT setting---has the benefit that doing so immediately improves the approximation ratios for the many algorithms we gave in this paper. We note that, like the embeddings of \cite{bartal1998approximating}, our embeddings are built around the paddedness of certain decompositions and these embeddings were later improved by FRT \cite{fakcharoenphol2004tight}. One might naturally wonder, then, if an FRT-like analysis might improve our stretch guarantees; from what we can tell no such FRT-type proof seems capable of improving our bounds. Another point to note is that we lose an $O(\log n)$ in the hop stretch when moving from partial tree metrics to partial tree embeddings. This loss does not seem to have an analogue in the (non-partial) tree embedding setting and it is not clear if such a loss is necessary. 

%\enote{BH: show new embeddings exist, get a bunch of poly logs, not clear if they are necessary, would be nice to get this down to optimal with either lower or upper bounds in the same way there was a progression with FRT; something stronger is true wrt the FRT analysis; our construction proof now is slightly different; one could do exactly FRT cutting and we looked at this (take a permutation, cut with FRT but define balls wrt hybrid metric); have a graph where FRT doesn't work; a natural way to do a decomposition with FRT us but there's a lower bound}

Moreover, while tree embeddings have proven useful for many network design problems, there are many other problems such as $k$-server \cite{bansal2011polylogarithmic}, metrical task systems \cite{bartal1997polylog} and requirement cuts \cite{nagarajan2005approximation} where tree embeddings enabled the first poly-log approximations. Thus, while the focus of our paper has been on the hop-constrained versions of network design problems, we expect that our embeddings will prove useful for the hop-constrained versions of many of these other problems.

%\enote{BH: one log n from the padding, one log n from the number of levels; with FRT can convert these to one; this does not seem possible for us}

%\enote{BH: Problems where FRT gave first poly log apx and then later works improved those?}

Lastly, as we discussed at the end of \Cref{sec:hop-bounded-hsts}, our $h$-hop partial tree embeddings are built on the worst-case stretch guarantees of our partial metrics; it would be interesting if it were possible to construct embeddings based on the expected stretch guarantees of our partial metrics. Such a result would immediately give several randomized algorithms for hop-constrained problems with low expected cost.

%Lastly, we note that, as in the (non-partial) tree embedding setting, there are many further natural settings in which to explore hop-constrained tree embeddings. Specifically, we leave the existence of improved $h$-hop partial tree embeddings for planar graphs as an open question for future work.% and $h$-hop partial tree embeddings which are subtrees of the input graph as open questions for future work. \enote{Nix subtree embeddings}

\appendix
\shortOnly{
\section{Deferred Proofs of \Cref{sec:apxHCD}}\label{sec:defProofsApxMetrics}

\noMetric*
\noMetricProof

\trivApx*
\trivApxProof

\hopconstrainedDecomp*
\hopconstrainedDecompProof
}

\bibliographystyle{alpha}
\bibliography{abb,main}

\newcommand{\etalchar}[1]{$^{#1}$}
\begin{thebibliography}{AFHP{\etalchar{+}}05}

\bibitem[AA97]{awerbuch1997buy}
Baruch Awerbuch and Yossi Azar.
\newblock Buy-at-bulk network design.
\newblock In {\em Symposium on Foundations of Computer Science (FOCS)}, pages
  542--547. IEEE, 1997.

\bibitem[AAA{\etalchar{+}}06]{alon2006general}
Noga Alon, Baruch Awerbuch, Yossi Azar, Niv Buchbinder, and Joseph Naor.
\newblock A general approach to online network optimization problems.
\newblock {\em ACM Transactions on Algorithms (TALG)}, 2(4):640--660, 2006.

\bibitem[ABN08]{abraham2008nearly}
Ittai Abraham, Yair Bartal, and Ofer Neiman.
\newblock Nearly tight low stretch spanning trees.
\newblock In {\em Symposium on Foundations of Computer Science (FOCS)}, pages
  781--790. IEEE, 2008.

\bibitem[AFHP{\etalchar{+}}05]{althaus2005approximating}
Ernst Althaus, Stefan Funke, Sariel Har-Peled, Jochen K{\"o}nemann, Edgar~A
  Ramos, and Martin Skutella.
\newblock Approximating k-hop minimum-spanning trees.
\newblock {\em Operations Research Letters}, 33(2):115--120, 2005.

\bibitem[AGG{\etalchar{+}}19]{abraham2019cops}
Ittai Abraham, Cyril Gavoille, Anupam Gupta, Ofer Neiman, and Kunal Talwar.
\newblock Cops, robbers, and threatening skeletons: Padded decomposition for
  minor-free graphs.
\newblock {\em SIAM Journal on Computing}, 48(3):1120--1145, 2019.

\bibitem[AKPW95]{alon1995graph}
Noga Alon, Richard~M Karp, David Peleg, and Douglas West.
\newblock A graph-theoretic game and its application to the k-server problem.
\newblock {\em SIAM Journal on Computing}, 24(1):78--100, 1995.

\bibitem[AKR95]{agrawal1995trees}
Ajit Agrawal, Philip Klein, and Ramamoorthi Ravi.
\newblock When trees collide: An approximation algorithm for the generalized
  steiner problem on networks.
\newblock {\em SIAM Journal on Computing}, 24(3):440--456, 1995.

\bibitem[AN12]{abraham2012using}
Ittai Abraham and Ofer Neiman.
\newblock Using petal-decompositions to build a low stretch spanning tree.
\newblock In {\em Annual ACM Symposium on Theory of Computing (STOC)}, pages
  395--406, 2012.

\bibitem[AT11]{akgun2011new}
{\.I}brahim Akg{\"u}n and Barbaros~{\c{C}} Tansel.
\newblock New formulations of the hop-constrained minimum spanning tree problem
  via miller--tucker--zemlin constraints.
\newblock {\em European Journal of Operational Research}, 212(2):263--276,
  2011.

\bibitem[Bar96]{bartal1996probabilistic}
Yair Bartal.
\newblock Probabilistic approximation of metric spaces and its algorithmic
  applications.
\newblock In {\em Symposium on Foundations of Computer Science (FOCS)}, pages
  184--193. IEEE, 1996.

\bibitem[Bar98]{bartal1998approximating}
Yair Bartal.
\newblock On approximating arbitrary metrices by tree metrics.
\newblock In {\em Annual ACM Symposium on Theory of Computing (STOC)}, pages
  161--168, 1998.

\bibitem[BBBT97]{bartal1997polylog}
Yair Bartal, Avrim Blum, Carl Burch, and Andrew Tomkins.
\newblock A polylog (n)-competitive algorithm for metrical task systems.
\newblock In {\em Annual ACM Symposium on Theory of Computing (STOC)}, pages
  711--719, 1997.

\bibitem[BBMN11]{bansal2011polylogarithmic}
Nikhil Bansal, Niv Buchbinder, Aleksander Madry, and Joseph Naor.
\newblock A polylogarithmic-competitive algorithm for the k-server problem.
\newblock In {\em Symposium on Foundations of Computer Science (FOCS)}, pages
  267--276. IEEE, 2011.

\bibitem[BC97]{berman1997line}
Piotr Berman and Chris Coulston.
\newblock On-line algorithms for steiner tree problems.
\newblock In {\em Annual ACM Symposium on Theory of Computing (STOC)}, pages
  344--353, 1997.

\bibitem[BFG15]{botton2015hop}
Quentin Botton, Bernard Fortz, and Luis Gouveia.
\newblock On the hop-constrained survivable network design problem with
  reliable edges.
\newblock {\em Computers \& Operations Research}, 64:159--167, 2015.

\bibitem[BFGP13]{botton2013benders}
Quentin Botton, Bernard Fortz, Luis Gouveia, and Michael Poss.
\newblock Benders decomposition for the hop-constrained survivable network
  design problem.
\newblock {\em INFORMS journal on computing}, 25(1):13--26, 2013.

\bibitem[BFU20]{barta2020online}
Yair Bartal, Nova Fandina, and Seeun~William Umboh.
\newblock Online probabilistic metric embedding: a general framework for
  bypassing inherent bounds.
\newblock In {\em Annual ACM-SIAM Symposium on Discrete Algorithms (SODA)},
  pages 1538--1557. SIAM, 2020.

\bibitem[BHR13]{bley2013ip}
Andreas Bley, S~Mehdi Hashemi, and Mohsen Rezapour.
\newblock Ip modeling of the survivable hop constrained connected facility
  location problem.
\newblock {\em Electronic Notes in Discrete Mathematics}, 41:463--470, 2013.

\bibitem[BIKP01]{bar2001generalized}
Judit Bar-Ilan, Guy Kortsarz, and David Peleg.
\newblock Generalized submodular cover problems and applications.
\newblock {\em Theoretical Computer Science}, 250(1-2):179--200, 2001.

\bibitem[CCGG98]{charikar1998rounding}
Moses Charikar, Chandra Chekuri, Ashish Goel, and Sudipto Guha.
\newblock Rounding via trees: deterministic approximation algorithms for group
  steiner trees and k-median.
\newblock In {\em Annual ACM Symposium on Theory of Computing (STOC)}, pages
  114--123, 1998.

\bibitem[CD16]{chlamtavc2016lowest}
Eden Chlamt{\'a}{\v{c}} and Michael Dinitz.
\newblock Lowest-degree $ k $-spanner: Approximation and hardness.
\newblock {\em Theory of Computing}, 12(1):1--29, 2016.

\bibitem[CEGS11]{chekuri2011set}
Chandra Chekuri, Guy Even, Anupam Gupta, and Danny Segev.
\newblock Set connectivity problems in undirected graphs and the directed
  steiner network problem.
\newblock {\em ACM Transactions on Algorithms (TALG)}, 7(2):1--17, 2011.

\bibitem[CEK06]{chekuri2006greedy}
Chandra Chekuri, Guy Even, and Guy Kortsarz.
\newblock A greedy approximation algorithm for the group steiner problem.
\newblock {\em Discrete Applied Mathematics}, 154(1):15--34, 2006.

\bibitem[CZ20]{chechik2020dynamic}
Shiri Chechik and Tianyi Zhang.
\newblock Dynamic low-stretch spanning trees in subpolynomial time.
\newblock In {\em Annual ACM-SIAM Symposium on Discrete Algorithms (SODA)},
  pages 463--475. SIAM, 2020.

\bibitem[DBF18]{de2018extended}
J{\'e}r{\^o}me De~Boeck and Bernard Fortz.
\newblock Extended formulation for hop constrained distribution network
  configuration problems.
\newblock {\em European Journal of Operational Research}, 265(2):488--502,
  2018.

\bibitem[DGM{\etalchar{+}}16]{diarrassouba2016integer}
Ibrahima Diarrassouba, Virginie Gabrel, Ali~Ridha Mahjoub, Lu{\'\i}s Gouveia,
  and Pierre Pesneau.
\newblock Integer programming formulations for the k-edge-connected
  3-hop-constrained network design problem.
\newblock {\em Networks}, 67(2):148--169, 2016.

\bibitem[DHK09]{demaine2009node}
Erik~D Demaine, MohammadTaghi Hajiaghayi, and Philip~N Klein.
\newblock Node-weighted steiner tree and group steiner tree in planar graphs.
\newblock In {\em International Colloquium on Automata, Languages and
  Programming (ICALP)}, pages 328--340. Springer, 2009.

\bibitem[DKR12]{dinitzmin}
Michael Dinitz, Guy Kortsarz, and Ran Raz.
\newblock Min-rep instances with large supergirth and the hardness of
  approximating basic spanners.
\newblock In {\em International Colloquium on Automata, Languages and
  Programming (ICALP)}, 2012.

\bibitem[DKR15]{dinitz2015label}
Michael Dinitz, Guy Kortsarz, and Ran Raz.
\newblock Label cover instances with large girth and the hardness of
  approximating basic k-spanner.
\newblock {\em ACM Transactions on Algorithms (TALG)}, 12(2):1--16, 2015.

\bibitem[DMMY18]{diarrassouba2018k}
Ibrahima Diarrassouba, Meriem Mahjoub, A~Ridha Mahjoub, and Hande Yaman.
\newblock k-node-disjoint hop-constrained survivable networks: polyhedral
  analysis and branch and cut.
\newblock {\em Annals of Telecommunications}, 73(1-2):5--28, 2018.

\bibitem[DZ16]{dinitz2016approximating}
Michael Dinitz and Zeyu Zhang.
\newblock Approximating low-stretch spanners.
\newblock In {\em Annual ACM-SIAM Symposium on Discrete Algorithms (SODA)},
  pages 821--840. SIAM, 2016.

\bibitem[EEST08]{elkin2008lower}
Michael Elkin, Yuval Emek, Daniel~A Spielman, and Shang-Hua Teng.
\newblock Lower-stretch spanning trees.
\newblock {\em SIAM Journal on Computing}, 38(2):608--628, 2008.

\bibitem[EP99]{elkin1999client}
Michael Elkin and David Peleg.
\newblock The client-server 2-spanner problem and applications to network
  design.
\newblock In {\em International Colloquium on Structural Information and
  Communication Complexity}, 1999.

\bibitem[EP00]{elkin2000strong}
Michael Elkin and David Peleg.
\newblock Strong inapproximability of the basic k-spanner problem.
\newblock In {\em International Colloquium on Automata, Languages and
  Programming (ICALP)}, pages 636--648. Springer, 2000.

\bibitem[EP05]{elkin2005approximating}
Michael Elkin and David Peleg.
\newblock Approximating k-spanner problems for k> 2.
\newblock {\em Theoretical Computer Science}, 337(1-3):249--277, 2005.

\bibitem[FGH20]{forster2020dynamic}
Sebastian Forster, Gramoz Goranci, and Monika Henzinger.
\newblock Dynamic maintanance of low-stretch probabilistic tree embeddings with
  applications.
\newblock {\em arXiv preprint arXiv:2004.10319}, 2020.

\bibitem[FRT04]{fakcharoenphol2004tight}
Jittat Fakcharoenphol, Satish Rao, and Kunal Talwar.
\newblock A tight bound on approximating arbitrary metrics by tree metrics.
\newblock {\em Journal of Computer and System Sciences}, 69(3):485--497, 2004.

\bibitem[GHR06]{gupta2006oblivious}
Anupam Gupta, Mohammad~T Hajiaghayi, and Harald R{\"a}cke.
\newblock Oblivious network design.
\newblock In {\em Annual ACM-SIAM Symposium on Discrete Algorithms (SODA)},
  pages 970--979, 2006.

\bibitem[GKL03]{gupta2003bounded}
Anupam Gupta, Robert Krauthgamer, and James~R Lee.
\newblock Bounded geometries, fractals, and low-distortion embeddings.
\newblock In {\em Symposium on Foundations of Computer Science (FOCS)}, pages
  534--543. IEEE, 2003.

\bibitem[GKR00]{garg2000polylogarithmic}
Naveen Garg, Goran Konjevod, and R~Ravi.
\newblock A polylogarithmic approximation algorithm for the group steiner tree
  problem.
\newblock {\em Journal of Algorithms}, 37(1):66--84, 2000.

\bibitem[GM03]{gouveia2003network}
Luis Gouveia and Thomas~L Magnanti.
\newblock Network flow models for designing diameter-constrained
  minimum-spanning and steiner trees.
\newblock {\em Networks: An International Journal}, 41(3):159--173, 2003.

\bibitem[Gou95]{gouveia1995using}
Luis Gouveia.
\newblock Using the miller-tucker-zemlin constraints to formulate a minimal
  spanning tree problem with hop constraints.
\newblock {\em Computers \& Operations Research}, 22(9):959--970, 1995.

\bibitem[Gou96]{gouveia1996multicommodity}
Luis Gouveia.
\newblock Multicommodity flow models for spanning trees with hop constraints.
\newblock {\em European Journal of Operational Research}, 95(1):178--190, 1996.

\bibitem[GR01]{gouveia2001new}
Luis Gouveia and Cristina Requejo.
\newblock A new lagrangean relaxation approach for the hop-constrained minimum
  spanning tree problem.
\newblock {\em European Journal of Operational Research}, 132(3):539--552,
  2001.

\bibitem[HHZ20]{haeupler2020repTree}
Bernhard Haeupler, D~Ellis Hershkowitz, and Goran Zuzic.
\newblock Deterministic repetition tree embeddings and online connectivity.
\newblock In {\em arXiv Preprint}, 2020.

\bibitem[HKS09]{hajiaghayi2009approximating}
Mohammad~Taghi Hajiaghayi, Guy Kortsarz, and Mohammad~R Salavatipour.
\newblock Approximating buy-at-bulk and shallow-light k-steiner trees.
\newblock {\em Algorithmica}, 53(1):89--103, 2009.

\bibitem[Kar89]{karp19892k}
Richard~M Karp.
\newblock A 2k-competitive algorithm for the circle.
\newblock {\em Manuscript, August}, 5, 1989.

\bibitem[KKM{\etalchar{+}}12]{khan2012efficient}
Maleq Khan, Fabian Kuhn, Dahlia Malkhi, Gopal Pandurangan, and Kunal Talwar.
\newblock Efficient distributed approximation algorithms via probabilistic tree
  embeddings.
\newblock {\em Distributed Computing}, 25(3):189--205, 2012.

\bibitem[KLS05]{konemann2005approximating}
Jochen K{\"o}nemann, Asaf Levin, and Amitabh Sinha.
\newblock Approximating the degree-bounded minimum diameter spanning tree
  problem.
\newblock {\em Algorithmica}, 41(2):117--129, 2005.

\bibitem[KMP11]{koutis2011nearly}
Ioannis Koutis, Gary~L Miller, and Richard Peng.
\newblock A nearly-m log n time solver for sdd linear systems.
\newblock In {\em Symposium on Foundations of Computer Science (FOCS)}, pages
  590--598. IEEE, 2011.

\bibitem[KP97]{kortsarz1997approximating}
Guy Kortsarz and David Peleg.
\newblock Approximating shallow-light trees.
\newblock In {\em Annual ACM-SIAM Symposium on Discrete Algorithms (SODA)},
  pages 103--110, 1997.

\bibitem[KP09]{kantor2009approximate}
Erez Kantor and David Peleg.
\newblock Approximate hierarchical facility location and applications to the
  bounded depth steiner tree and range assignment problems.
\newblock {\em Journal of Discrete Algorithms}, 7(3):341--362, 2009.

\bibitem[KRS01]{konjevod2001approximating}
Goran Konjevod, R~Ravi, and F~Sibel Salman.
\newblock On approximating planar metrics by tree metrics.
\newblock {\em Information Processing Letters (IPL)}, 80(4):213--219, 2001.

\bibitem[KS11]{khani2011improved}
M~Reza Khani and Mohammad~R Salavatipour.
\newblock Improved approximations for buy-at-bulk and shallow-light k-steiner
  trees and (k, 2)-subgraph.
\newblock In {\em Annual International Symposium on Algorithms and Computation
  (ISAAC)}, pages 20--29. Springer, 2011.

\bibitem[Lei16]{leitner2016layered}
Markus Leitner.
\newblock Layered graph models and exact algorithms for the generalized
  hop-constrained minimum spanning tree problem.
\newblock {\em Computers \& Operations Research}, 65:1--18, 2016.

\bibitem[MRS{\etalchar{+}}98]{marathe1998bicriteria}
Madhav~V Marathe, Ramamoorthi Ravi, Ravi Sundaram, SS~Ravi, Daniel~J
  Rosenkrantz, and Harry~B Hunt~III.
\newblock Bicriteria network design problems.
\newblock {\em Journal of algorithms}, 28(1):142--171, 1998.

\bibitem[NPS11]{naor2011online}
Joseph Naor, Debmalya Panigrahi, and Mohit Singh.
\newblock Online node-weighted steiner tree and related problems.
\newblock In {\em Symposium on Foundations of Computer Science (FOCS)}, pages
  210--219. IEEE, 2011.

\bibitem[NR05]{nagarajan2005approximation}
Viswanath Nagarajan and Ramamoorthi Ravi.
\newblock Approximation algorithms for requirement cut on graphs.
\newblock In {\em Approximation, Randomization and Combinatorial Optimization.
  Algorithms and Techniques}, pages 209--220. Springer, 2005.

\bibitem[NS97]{naor1997improved}
Joseph Naor and Baruch Schieber.
\newblock Improved approximations for shallow-light spanning trees.
\newblock In {\em Symposium on Foundations of Computer Science (FOCS)}, pages
  536--541. IEEE, 1997.

\bibitem[Rac02]{racke2002minimizing}
Harald Racke.
\newblock Minimizing congestion in general networks.
\newblock In {\em Symposium on Foundations of Computer Science (FOCS)}, pages
  43--52. IEEE, 2002.

\bibitem[RAJ12]{rossi2012connectivity}
Andr{\'e} Rossi, Alexis Aubry, and Mireille Jacomino.
\newblock Connectivity-and-hop-constrained design of electricity distribution
  networks.
\newblock {\em European journal of operational research}, 218(1):48--57, 2012.

\bibitem[Rav94]{ravi1994rapid}
R~Ravi.
\newblock Rapid rumor ramification: Approximating the minimum broadcast time.
\newblock In {\em Symposium on Foundations of Computer Science (FOCS)}, pages
  202--213. IEEE, 1994.

\bibitem[TCG15]{thiongane2015formulations}
Babacar Thiongane, Jean-Fran{\c{c}}ois Cordeau, and Bernard Gendron.
\newblock Formulations for the nonbifurcated hop-constrained multicommodity
  capacitated fixed-charge network design problem.
\newblock {\em Computers \& Operations Research}, 53:1--8, 2015.

\bibitem[Vo{\ss}99]{voss1999steiner}
Stefan Vo{\ss}.
\newblock The steiner tree problem with hop constraints.
\newblock {\em Annals of Operations Research}, 86:321--345, 1999.

\bibitem[WA88]{woolston1988design}
Kathleen~A Woolston and Susan~L Albin.
\newblock The design of centralized networks with reliability and availability
  constraints.
\newblock {\em Computers \& Operations Research}, 15(3):207--217, 1988.

\end{thebibliography}

\end{document}